\documentclass[12pt]{article}
\usepackage[a4paper]{geometry}
\usepackage[italian,english]{babel}
\usepackage[T1]{fontenc}

\usepackage{amsmath, accents}
\usepackage{amsfonts}
\usepackage{amstext}
\usepackage{amssymb}
\usepackage{amsthm}
\usepackage{amscd}
\usepackage{mathrsfs}
\usepackage{bbold}
\usepackage{dsfont}
\usepackage{bbm}

\usepackage[pagebackref,draft=false]{hyperref}
\hypersetup{colorlinks,
linkcolor=myrefcolor,
citecolor=mycitecolor,
urlcolor=myurlcolor}

\usepackage[capitalize]{cleveref}
\usepackage{cite}
\usepackage{caption}
\usepackage{etaremune}

\usepackage{xcolor}
\definecolor{myurlcolor}{rgb}{0,0,0.4}
\definecolor{mycitecolor}{rgb}{0,0.5,0}
\definecolor{myrefcolor}{rgb}{0.5,0,0}
\usepackage{graphicx}
\usepackage{tikz}
\usepackage{tikz-cd}
\usetikzlibrary{graphs,decorations.pathmorphing,decorations.markings}

\usepackage{lipsum}
\usepackage{etoolbox}
\usepackage{makeidx}
\usepackage{sectsty}
\usepackage{dsfont}
\usepackage{enumitem} 
\usepackage[]{latexsym}
\usepackage{braket}
\usepackage{caption}
\usepackage[utf8]{inputenx}
\usepackage{lmodern}
\usepackage{textcomp}
\usepackage{microtype}
\usepackage{totcount}
\usepackage{blindtext}

\newtheorem{theorem}{Theorem}[section]
\newtheorem{remark}[theorem]{Remark}

\newtheorem{proposition}[theorem]{Proposition}
\newtheorem{definition}[theorem]{Definition}
\newtheorem{example}[theorem]{Example}

\newtheorem*{proof*}{Proof}

\newcommand{\be}{\begin{equation}}
\newcommand{\ee}{\end{equation}}
\newcommand{\bea}{\begin{eqnarray}}
\newcommand{\eea}{\end{eqnarray}}
\newcommand{\vsp}{\vspace{0.4cm}}

\newcommand{\ac}{\mathscr{S}}

\newcommand{\m}{\mathscr{M}}
\newcommand{\mm}{\mathbb{M}}
\newcommand{\pe}{\mathcal{P}(\mathbb{E})}

\newcommand{\fpe}{\mathcal{F}_{\mathcal{P}(\mathbb{E})}}

\newcommand{\pssos}{\Pi^\star_\Sigma \Omega^\Sigma}
\newcommand{\os}{\Omega^\Sigma}

\newcommand{\ps}{\Pi_\Sigma}
\newcommand{\pss}{\Pi^\star_\Sigma}

\newcommand{\lag}{\mathscr{L}}
\newcommand{\elag}{\mathcal{E}\mathscr{L}}

\newcommand{\dd}{{\rm d}}
\newcommand{\de}{\partial}
\newcommand{\nn}{\nonumber}

\title{The geometry of the solution space of first order Hamiltonian field theories I: from particle dynamics to free Electrodynamics}

\author{F. M. Ciaglia$^{2,6}$ \href{https://orcid.org/0000-0002-8987-1181}{\includegraphics[scale=0.7]{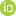}}, F. Di Cosmo$^{1,2,7}$ \href{https://orcid.org/0000-0003-0256-5913}{\includegraphics[scale=0.7]{ORCID.png}}, A. Ibort$^{1,2,8}$ \href{https://orcid.org/0000-0002-0580-5858}{\includegraphics[scale=0.7]{ORCID.png}}, \\ G. Marmo$^{3,4,9}$ \href{https://orcid.org/0000-0003-2662-2193}{\includegraphics[scale=0.7]{ORCID.png}}, L. Schiavone$^{3,5,10}$  \href{https://orcid.org/0000-0002-1817-5752}{\includegraphics[scale=0.7]{ORCID.png}}, A. Zampini$^{3,5,11}$ \href{https://orcid.org/0000-0003-0980-6003}{\includegraphics[scale=0.7]{ORCID.png}} \\
\footnotesize{$^{1}$\textit{ICMAT, Instituto de Ciencias Matem\'{a}ticas (CSIC-UAM-UC3M-UCM), Madrid, Spain}} \\
\footnotesize{$^{2}$\textit{Departamento de Matem\'aticas, Universidad Carlos III de Madrid, Legan\'es, Madrid, Spain}} \\
\footnotesize{$^{3}$\textit{INFN-Sezione di Napoli, Naples, Italy}} \\
\footnotesize{$^{4}$\textit{Dipartimento di Fisica ``E. Pancini'', Universit\`a di Napoli Federico II,  Naples, Italy}} \\
\footnotesize{$^{5}$\textit{Dipartimento di Matematica e Applicazioni "Renato Caccioppoli", Università di Napoli Federico II, Napoli, Italy}} \\
\footnotesize{$^{6}$\textit{e-mail: \texttt{fciaglia[at]math.uc3m.es
}}} \,\, 
\footnotesize{$^{7}$\textit{e-mail: \texttt{fcosmo[at]math.uc3m.es}}} \\
\footnotesize{$^{8}$\textit{e-mail: \texttt{albertoi[at]math.uc3m.es}}} \,\,  
\footnotesize{$^{9}$\textit{e-mail: \texttt{marmo[at]na.infn.it}}} \\ 
\footnotesize{$^{10}$\textit{e-mail: \texttt{lschiavo[at]math.uc3m.es}}} \,\,  
\footnotesize{$^{11}$\textit{e-mail: \texttt{azampini[at]na.infn.it}}} 
}

\begin{document}

\maketitle
\tableofcontents

\begin{abstract}
We analyse the problem of defining a Poisson bracket structure on the space of solutions of the equations of motions of first order Hamiltonian field theories.
The cases of Hamiltonian mechanical point systems -- as a ($0+1$)-dimensional field -- and more general field theories without gauge symmetries are addressed by showing the existence of a symplectic (and, thus, a Poisson) structure on the space of solutions.
Also the easiest case of gauge theory, namely free electrodynamics, is considered: within this problem, a pre-symplectic tensor on the space of solutions is introduced, and a Poisson structure is induced in terms of  a flat connection on a suitable  bundle associated to the theory.
\end{abstract}

\section{Introduction}

\noindent The need for a systematic study of the geometry of the space of solutions of differential equations describing physical systems was first pointed out by \textit{Souriau} in \cite{Souriau1997-Symplectic_view_Physics}.
Using his own words:
\begin{quotation}
\textit{Analytical mechanics is not an outdated theory, but it appears that the categories which one classically attributes to it such as configuration space, phase space, Lagrangian formalism, Hamiltonian formalism, are, simply because they do not have the required covariance; in other words, because these categories are in contradiction with Galilean relativity. A fortiori, they are inadequate for the formulation of relativistic mechanics in the sense of Einstein.} 
\end{quotation}
Coeherently with this approach, his suggestion was to  study the geometry not of the \emph{phase space} usually adopted to describe the Hamiltonian formulation of a dynamical system, but of what he named the \emph{space of motions} for a dynamics, i.e. the space of the solutions of the equations of motion representing the set of physical trajectories  for the system. 
In the present paper we shall refer to it as the \emph{solution space} of the motion, and it will be given as the space of extrema of a suitable action functional encoding the dynamical content of the theory.

\noindent Upon referring to the quantum-to-classical transition in physics, it is interesting to recall  the so called \emph{correspondence principle}, whose first formulation  was elaborated by Bohr (see \cite{Rosen-Nielsen1976-Bohr_papers_corr_princ} for a collection of  papers on it), stating that the predictions of the quantum theory under development should agree, in a high energy limit, with the predictions of the old (i.e. classical) theory. Such a sentence was made more mathematically rigorous by Dirac  \cite{Dirac1925-fundamental_eq_QM}: in a suitable classical limit, the commutator structure on quantum observables (i.e. linear operators on a separable complex Hilbert space) should correspond to a Poisson structure on their classical counterparts.

\noindent The aim of this paper is then to adopt Souriau's point of view and to analyse how it is possible to define a Poisson tensor on the solution space of a suitable class of field theories. 

\noindent This is not a new task:  it is interesting to recall that  a closed  form on the space of solutions for the Yang-Mills equations, for the Einstein equations for general relativity and for the string dynamics was introduced in  \cite{Crnkovic-Witten1986-Covariant_canonical_formalism} and \cite{Crnkovic1988-Symplectic_Cov_Phase_Space}, at least at a formal level from the point of view of differential geometry. Such a closed form turned out to be pre-symplectic, its kernel being related to the  gauge invariance of the equations. The search for a Poisson bracket between functions of the solutions of  field theories motivated by the wish of better understanding their quantum counterpart dates back indeed  to  Peierls,   who introduced covariant commutation rules in relativistic field theories in 1952 \cite{Peierls1952-Commutation_laws}. Such rules come by applying an algorithmic procedure  to functions of the solutions of the equations of motion via the Green's function of the linearized equations of the motion. 
As testified by \textit{B. DeWitt}'s words in \cite{DeWitt2003-Global_approach_QFT}:
\begin{quote}
\textit{There exists an anomaly today in the pedagogy of physics. 
When expounding the fundamentals of quantum field theory physicists almost universally fail to apply the lessons that relativity theory taught them early in the twentieth century. 
Although they usually carry out their calculations in a covariant way, in deriving their calculational rules they seem unable to wean themselves from canonical methods and Hamiltonians, which are holdovers from the nineteenth century and are tied to the cumbersome $(3+1)$-dimensional baggage of conjugate momenta, bigger-than-physical Hilbert spaces, and constraints. 
There seems to be a feeling that only canonical methods are {\rm  safe}; only they guarantee unitarity. 
This is a pity because such a belief is wrong, and it makes the foundations of field theory unnecessarily cluttered. 
One of the unfortunate results of this belief is that physicists, over the years, have almost totally neglected the beautiful covariant replacement for the 
canonical Poisson bracket that Peierls invented in 1952.}
\end{quote}
Peierls' construction went largely unnoticed or neglected, with only a few exceptions.
Indeed, DeWitt himself was the one who developed Peierls' approach,  applying it to the  case of the equations for the gravitational field  (see \cite{DeWitt1960-Commutators_Quant_Grav}) and adopting it in \cite{DeWitt1965-Groups_Fields, DeWitt2003-Global_approach_QFT} as the basic structure upon which one constructs a coherent covariant description of more general field theories. 

\noindent Although providing a breakthrough in the formulations of field theories from a relativistic point of view, the above mentioned contributions lack a deep analysis of the very geometrical structure involved in the construction of the covariant brackets.

\noindent Indeed, a mathematical framework that  fits the geometrical description of covariant field theories is the multisymplectic formalism, which started to be developed few years later, with the 
seminal contributions given by  \cite{ Kijowski1973-Canonical_formalism_CFT, Kijowski-Szczyrba1976-canonical_field_theories, Tulczyjew-Kjowksi1979-Symplectic_Framework_Field_Theories, Garcia-Perez-Rendon1969-Symplectic_Quantized_FieldsI, Garcia-Perez-Rendon1969-Symplectic_Quantized_FieldsII}. These papers were conceived as a generalisation to field theory of the symplectic formalism so successful in a geometric description of classical mechanics. 
In the successive decades several contributions to the development of multisymplectic geometry appeared. Among them, the papers \cite{Krupka-Krupkova-Saunders2010-Cartan_Form_generalizations, Krupka-Krupkova-Saunders2012-CartanLepage_forms, Krupkova1997-Geometry_OVE} (and references therein) focus on the geometrical formulation of Lagrangian variational principles and on the generalizations of the definition of the Poincar\'e-Cartan form within higher order field theories, while \cite{Fatibene-Francaviglia2003-Natural_Gauge_Natural} (and references therein) aims to a geometrical formulation of General Relativity as a field theory. 
Other relevant contributions to the development of multisymplectic geometry are in the following, far from exhaustive, list \cite{Gold-Sternb1973-HamiltonCartan_formalism_CV, Zuckermann1987-Action_principles_Global, Gotay1991-Multisymplectic_Field-TheoryI_Hamiltonian, Gotay1991-Multisymplectic_Field_TheoryII_spacetime, Giac-Mang-Sard1997-New_Lagrangian_Hamiltonian, Carinena-Crampin-Ibort1991-Multisymplectic, Ibort-Spivak2017-Covariant_Hamiltonian_YangMills, Binz-Sniatycki-Fisher1988-Classical_Fields, Ech-Munoz-RomanRoy2000-Multisymplectic_Hamiltonian, Eche-Munoz-RomanRoy1996-Geometry_Lagrangian_Field_Theories, Roman-Roy2009-Multisymplectic_Lagrangian_Hamiltonian, Montiel-Molgado-Lopez2021-Geometric_field_theory_Gravity, Margalef-Villasenor2021-Covariant_phase_space_boundary}.

\noindent Within the setting of multisymplectic geometry, the construction of a covariant Poisson bracket for field theories has received a renewed interest, although under a  slightly different perspective, that of studying the very differential geometrical structures involved in the construction of the brackets introduced in the literature cited above (see
\cite{Marsd-Mont-Morr-Thomp1986-Covariant_Poisson_bracket, Forger-Romero2005-Covariant_poisson_brackets, Forger-Salles2015-Covariant_Poisson_brackets, Khavkine2014-Covariant_phase_space, Asorey-Ciaglia-DC-Ibort-Marmo2017-Covariant_Jacobi_brackets, Asorey-Ciaglia-DC-Ibort2017-Covariant_brackets, Ciaglia-DC-Ibort-Marmo-Schiav2020-Jacobi_Particles, Ciaglia-DC-Ibort-Marmo-Schiav2020-Jacobi_Fields, Gieres2021-Covariant_field_theories, Balmaz-Marrero-Martinez2022-Affine_bracket}).
Other contributions, more oriented to the problem of reduction within a covariant formulation of classical field theories, worth to be mentioned are \cite{Castrillon-Berbel2023-Lagrangian_reduction, Castrillon-Berbel2023-Poisson_reduction, Castrillon-Rodriguez2022-Gauge_reduction} (see also references therein).

\noindent The aim of this paper (which we conceive as the first of a series) is to merge some of the previous considerations and offer a unified account of the description of covariant brackets together with a consistent analysis of the geometric and global-analytic structures presented by such a bracket.
With respect to the existing literature, we proceed a step further along two directions.

\noindent 
We extend the construction, given in \cite{Forger-Romero2005-Covariant_poisson_brackets} for the scalar field theory, to the more general setting of gauge theories. Notice that, 
for the purposes of constructing a Poisson structure on the quotient of the solution space of the equations modulo gauge symmetries, this was already done in \cite{Khavkine2014-Covariant_phase_space}.
That construction makes indeed use of the possibility of performing a gauge fixing and thus, if one aims at providing a bracket given as a suitable Poisson bivector which is \textit{globally}  defined  on the solution space (prior to quotienting it modulo gauge symmetries), it can not be applied for those gauge theories for which the gauge fixing can not be global in the space of fields, i.e. for theories presenting the so called \textit{Gribov's ambiguities} (note that, for instance, all non-Abelian gauge theories fall  into this case). 
Therefore, for such a class of gauge theories the construction in \cite{Khavkine2014-Covariant_phase_space} still remains valid but only locally in the space of fields.
In this paper, we formulate such a condition to construct the Poisson bracket on the solution space as a precise geometrical condition, that is  as the flatness of a connection on a bundle suitably  associated to the gauge field equations. 
Even if such a condition excludes the same class of theories (those presenting Gribov's ambiguities), the formalism that we describe within this paper can be  extended to the realm of non-Abelian gauge theories by a consistent use of the coisotropic embedding theorem (following an idea introduced in \cite{Ciaglia-DC-Ibort-Marmo-Schiav-Zamp-2022-Symmetry}).
We aim to describe such an extension in the second part of this project 
while focussing in the following pages on the Abelian example given by free Electrodynamics (i.e. sourceless Maxwell equations)  in order to introduce the geometrical construction within the easiest case of a ${\rm U}(1)$ gauge theory.
In particular, in the present paper we show that the solution space of a gauge theory is a pre-symplectic manifold, and introduce a Poisson bracket on it in terms of a flat connection which arises as the horizontal complement to the vertical distribution given  by fibering such manifold along the kernel of the pre-symplectic form.  This approach indeed circumvents the problem of a quite difficult rigorous  analysis of the resulting quotient manifold, i.e. the basis of the bundle, as well as the problem of a clear physical interpretation of it.

\noindent 
Moreover, we avoid the usual approach of dealing with the space of smooth fields as a formal differential manifold on which a formal differential calculus \`a la Cartan is defined and we  deal, in all the examples we analyse,  with well defined Banach or Hilbert manifolds on which all the geometrical objects and the differential calculus are rigorously defined.

\vsp

\noindent This paper originates by noticing that, for  the partial differential equations\footnote{pde for short.} we consider, there exists  a bijection
\be
\Psi \;\; : \;\; \elag^\Sigma \to \elag \,.
\ee
between the solution space $\elag$ and the space of Cauchy data\footnote{We denote by $\Sigma$ the hypersurface of the base manifold (i.e. the space-time)  on which Cauchy data are given.} $\elag^\Sigma$ by virtue of the existence and uniqueness theorem for the (at least weak) solutions of the equations of motion arising from the variational principle (see for example  \cite{Ciaglia-DC-Ibort-Marmo-Schiav2020-Jacobi_Particles, Ciaglia-DC-Ibort-Marmo-Schiav2020-Jacobi_Fields} where this idea is already introduced). For the specific pde's that we consider, it turns out that  both spaces are  Banach manifolds with a well defined notion of differential calculus and that, with respect to these differential structures, the bijection $\Psi$ is a diffeomorphism.

\noindent Within  the multysimplectic formulation of first order Hamiltonian field theories a canonical $2$-form on the solution space naturally emerges from the variational principle. 
Let us mention that the opposite point of view, namely that of starting from a (set of) pde together with a pre-symplectic structure in order to reconstruct a variational principle, was taken in \cite{Grigoriev2022-Intrinsic_inverse_problem} in order to provide a solution of the inverse problem of the calculus of variations in suitable cases.
In our case, the $2$-form emerging from the variational principle can be written as  
$$\Omega \,=\, {\Psi^{-1}}^\star \tilde{\Omega},$$
where $\tilde{\Omega}$ is a suitable $2$-form on $\elag^\Sigma$.
Depending on the properties of $\elag^\Sigma$ and $\tilde{\Omega}$ we discuss cases (in increasing order of difficulty) in which the form  ${\Psi^{-1}}^\star \tilde{\Omega}$ is used to define a Poisson bracket on $\elag$:
\begin{itemize}
\item For point particle  mechanical systems it turns out  that $\elag^\Sigma$ is a finite-dimensional manifold and that $\tilde{\Omega}$ is a symplectic form.
Therefore, since $\Psi^{-1}$ is a diffeomorphism, the form ${\Psi^{-1}}^\star \tilde{\Omega}$ is symplectic as well and corresponds  to a Poisson bracket;
\item For the pde's describing field theories, one has that  $\elag^\Sigma$ is infinite-dimensional and $\tilde{\Omega}$ is proved to be the structure emerging from the pre-symplectic constraint algorithm developed  when dealing with singular Hamiltonian  systems, i.e. those presenting (following Dirac's theory \cite{Dirac1964-Lectures_QM}) \emph{constraints}.

\noindent When the equations presents only first class constraints, one has that $\tilde\Omega$ is symplectic. It then follows that  ${\Psi^{-1}}^\star \tilde{\Omega}$ is symplectic and gives rise to a Poisson bracket.
\item The last case we analyse is that  of pde with gauge symmetries for which  the pre-symplectic constraint algorithm ends up with a degenerate $2$-form and thus both $\tilde{\Omega}$ and ${\Psi^{-1}}^\star \tilde{\Omega}$ are pre-symplectic.
In such a case, a Poisson bracket can be defined in terms of a flat connection on a principal bundle associated to the gauge theory. 
\end{itemize}

\noindent The paper is organised as follows.

\noindent  Section \ref{Sec: Multisymplectic formulation of first order Hamiltonian field theories}  recalls the basic aspects of the multysimplectic formalism used to describe first order Hamiltonian field theories.

\noindent Section \ref{Sec: Schwinger-Weiss variational principle, the first fundamental formula and the solution space} deals with the formulation of the Schwinger-Weiss action principle within the multisymplectic framework and with the construction of the solution space of the theory.

\noindent In Section \ref{Sec: The canonical structure on the solution space} the existence of the canonical $2$-form on the solution space is exhibited, while in Section \ref{Sec: The symplectic case vs gauge theories} it is shown that such a form  is  symplectic for equations  presenting only first class constraints and pre-symplectic for equations with  gauge invariance. 

\noindent To conclude, Section \ref{Sec: Poisson brackets on the solution space} shows how to construct a Poisson bracket on the solution space related with the canonical $2$-form both in the symplectic case and within gauge theories.
The construction of the bracket in the three relevant examples cited above is presented in \ref{Subsec: Mechanical systems: the free particle on the line}, \ref{Subsec: The symplectic case: massive vector boson field} and \ref{Subsec: Gauge theories: Electrodynamics}.



\section{Multisymplectic formulation of first order Hamiltonian field theories}
\label{Sec: Multisymplectic formulation of first order Hamiltonian field theories}

\noindent 
A covariant Hamiltonian description of the pde associated to a classical field theory begins upon considering a  (smooth) manifold $\mathscr{M}$ -- possibly with a non empty boundary $\partial\m$ -- and a (smooth)  bundle 
\be
\label{1ale}
\pi \colon \mathbb{E} \to \mathscr{M}
\ee
with typical fiber a (smooth) manifold $\mathcal{E}$, whose sections $\phi$ model the fields of the theory. The manifold $\m$ is assumed to be $n=1+d$ dimensional. The only global aspect required on $\mathscr{M}$ is that it is orientable, so that  a volume $(d+1)$-form exists. Unless stated otherwise,  a local chart is denoted by  
 $(U_{\mathscr{M}},\, \psi_{U_{\mathscr{M}}})$, with 
 \be
 \label{a267.1}
 \psi_{U_\m}(m) = \left( x^0,..., x^d \right)
 \ee
  for $m\in U_{\mathscr{M}} \subseteq \mathscr{M}$ and the volume form is assumed to be  
\be\label{vom}
vol_{\mathscr{M}} \,=\, \mathrm{d}x^0 \wedge ... \wedge \mathrm{d}x^d.
\ee
Notice that this allows to consider point particle dynamics as the  $d=0$ case of a more general field theory. Moreover, in many examples coming from physics (and for all the examples we consider in this paper), the manifold $\mathscr{M}$ is a space-time, that is it carries a Lorentzian metric. For point particle dynamics one then usually writes $vol_\m=\dd t$.

\noindent 
Since the fields of the theory are modelled by sections of the fibration $\pi\colon\mathbb{E}\to\mathscr{M}$, they can be formulated in terms of suitable local maps from open sets in $\mathscr{M}$ to $\mathcal{E}$.  Such a manifold $\mathcal{E}$ generalises to field theory the role played by the configuration manifold (usually denoted by $Q$) in point mechanics, to which it reduces for the case $d=0$. This is the reason why sections are usually referred to as {\it local configurations} of the field. A general mathematical description of such a  set of local configurations is beyond the scope of this paper. We limit ourselves to assume that it is a smooth manifold, and prove this assumption right for the examples we describe in detail. 
We consider an adapted fibered chart on $\mathbb{E}$ given by $(U_{\mathbb{E}},\, \psi_{U_{\mathbb{E}}})$, with $$\psi_{U_{\mathbb{E}}}(\mathfrak{e}) = \left(x^0,...,x^d,\,u^1,..., u^r\,\right),$$ where $\dim \,\mathcal{E}\,=\,r$ and $\mathfrak{e}$ denotes a point in $U_{\mathbb{E}} \subseteq \mathbb{E}$. 
Local sections of the fibration $\pi\colon\mathbb{E}\to\mathscr{M}$ are then  denoted as
$$
\phi\,=\, (\,\phi^1(x),...,\, \phi^r(x)\,).
$$

\noindent 
Regarding the boundary, it is assumed to be  a $d$-dimensional smooth manifold (not necessarily connected), suitably embedded in $\mathscr{M}$, that is it is assumed that there exists a collar  
\be
\label{emb2} 
C_\epsilon \,=\, [0,\, \epsilon) \times \partial \mathscr{M}
\ee 
around it and an embedding $\mathfrak{i}_{C_\epsilon} \colon C_\epsilon \hookrightarrow \mathscr{M}$, such that  $i( \{ 0 \} \times  \partial \mathscr{M}) = \partial \mathscr{M}$, 
with
\be
\label{emb3}
\mathfrak{i}_{C_\epsilon}^\star \, vol_{\mathscr{M}} \,=\, \mathrm{d}x^0 \wedge vol_{\partial \mathscr{M}}
\ee 
where $vol_{\partial \mathscr{M}}$ is a volume form on $\partial \mathscr{M}$.

\noindent 
It is well known that, given a smooth manifold $Q$, its cotangent bundle $\mathbf{T}^\star Q$ has a canonical symplectic structure, which makes it the natural setting for a Hamiltonian description of a point particle dynamics whose configuration space is the base manifold $Q$.
The natural generalisation of the notion of phase space $\mathbf{T}^\star Q$  is that of {\it covariant phase space} $\mathcal{P}(\mathbb{E})$ associated to the bundle $\pi \; :\;\mathbb{E} \to \m$. 
Our aim is to recall the basics of its constructions, while referring to the literature cited in the introduction for a more extensive account.

\noindent The first step along this path comes by noticing that the manifold $\mathbf{T}^\star Q$ can be described as the dual bundle to the tangent bundle $\mathbf{T}Q$ corresponding to $Q$. 
The generalisation of the notion of tangent bundle manifold appears to be that of the first order jet bundle $\mathbf{J}^1\pi$ associated to $\pi\,:\,\mathbb{E}\,\to\,\m$.
This is the bundle over the base $\mathbb{E}$ whose typical fiber is given by considering, at each  point  $m\in\mathscr{M}$,  the first order jets (or first order germs) of sections of the fibration $\pi\colon\mathbb{E}\to\mathscr{M}$. These are equivalence classes of sections of $\pi\colon\mathbb{E}\to\mathscr{M}$  sharing the  first order Taylor expansion in a chart $(U_m, \psi_{U_m})$ around $m$. Sections $\phi_1, \phi_2\colon\mathscr{M}\to\mathcal{E}$ have the same first order jet when 
\begin{equation}
\label{1jcon}
\begin{split}
&\phi_1 \simeq \phi_2 \qquad \Leftrightarrow  \\ & \qquad \phi_1^a(m) = \phi_2^a(m)\,, \qquad \left.\frac{\partial (\phi_1^a \circ \psi_{U_m}^{-1})}{\partial x^{\mu}}\right|_{x=\psi_{U_m}(m)}= \left. \frac{\partial (\phi_2^a \circ \psi_{U_m}^{-1})}{\partial x^{\mu}}\right|_{x=\psi_{U_m}(m)}
\end{split}
\end{equation}  
for 
$\mu = 0,...,d\,,\;\;a=1,...,r\,.$ The set 
$\mathbf{J}^1 \pi$ turns to be  the total space of a fiber bundle structure over $\m$ and of another fiber bundle structure over  $\mathbb{E}$ (see \cite{Saunders1989-Jet_Bundles} page $97-101$), i.e. 
\be
\label{1jdia}
\begin{tikzcd}
\mathbf{J}^1 \pi \arrow[d, "\pi^1_0"] \arrow[dd, "\pi_1", bend right = 50, swap] \\
\mathbb{E} \arrow[d, "\pi"] \\
\mathscr{M} 
\end{tikzcd} \,.
\ee
The space $\mathbf{J}^1\pi$ is a $(n+r+nr)$-dimensional smooth manifold on which we consider the adapted fibered chart $(U_{\mathbf{J}^1 \pi},\, \psi_{U_{\mathbf{J}^1 \pi}})$, 
\be
\psi_{U_{\mathbf{J}^1 \pi}}(\mathfrak{e}_x) = (x^0,...,x^d,\,u^1,...,u^r,\, z^1_0,...,z^r_d) \,,
\ee
with $\mathfrak{e}_x$  a point in $U_{\mathbf{J}^1 \pi} \subseteq \mathbf{J}^1 \pi$. 

\noindent It is immediate to see that, if $\m\simeq \mathbb{R}$ and $\mathcal{E}$ is a manifold $Q$, so that one has $\pi\,:\,Q\,\times\,\mathbb{R}\,\to\,\mathbb{R}$,  it results that $\mathbf{J}^1\pi\simeq \mathbf{T}Q\times \mathbb{R}$.  
Such a formalism provides then, via the double fibration $$\mathbf{T}Q\times \mathbb{R}\quad\to\quad Q\times\mathbb{R}\quad\to\quad\mathbb{R},$$ the  natural geometric setting for a Lagrangian description of a  possibly time dependent point  dynamical system,   where histories of the system under investigation are modelled as those elements which emerge as critical points (with respect to a suitable class of variations) of a functional for first order jet prolongations of curves on $Q$, i.e. sections of the bundle $\pi\colon Q\times\mathbb{R}\to\mathbb{R}$. That the first order jet bundle formalism provides a natural geometric setting for a Lagrangian description of a suitable class of first order  field theories has been analysed at length in  \cite{Sniatycki1970-Geometric_Field_Theory, Binz-Sniatycki-Fisher1988-Classical_Fields, Eche-Munoz-RomanRoy1996-Geometry_Lagrangian_Field_Theories, Giac-Mang-Sard1997-New_Lagrangian_Hamiltonian, Krupkova1997-Geometry_OVE, Fatibene-Francaviglia2003-Natural_Gauge_Natural, Krupka2015-Variational_Geometry}), and to them we refer for more details. 
Along this path, it is natural to analyse 
 whether the dual bundle to $\bf{J}^1\pi$ is suitable for a Hamiltonian description of a class of first order field theories. 

\noindent The answer to such a question starts by noticing that  $\mathbf{J}^1 \pi$ is an affine bundle over $\mathbb{E}$. 
Thus, its dual has to be introduced in terms of affine spaces.  The typical fibre of the dual bundle comes as  the set of all affine maps over the fibres of $\pi^1_0\colon\mathbf{J}^1\pi\to\mathbb{E}$:
\begin{equation}
z^a_\mu \mapsto \rho_a^\mu \, z^a_\mu + \rho_0\,.
\end{equation}
Such a dual bundle (sometimes called \textit{extended dual}), denoted by $\mathrm{Aff}(\mathbf{J}^1\pi,\, \mathbb{R})$ or  by $\mathbf{J}^\dag \pi$, is a fibre bundle over both $\mathbb{E}$ and  $\mathscr{M}$. 
We consider on $\mathbf{J}^\dag \pi$ the adapted fibered chart $(U_{\mathbf{J}^\dag \pi},\, \psi_{U_{\mathbf{J}^\dag \pi}})$, $$\psi_{\mathbf{J}^\dag \pi}(\bar{\mathfrak{p}}) = (x^0,...,x^d,\,u^1,...,u^r,\, \rho_1^0,...,\rho_r^d,\, \rho_0),$$ with  $\bar{\mathfrak{p}}$  a point in  $U_{\mathbf{J}^\dag \pi} \subseteq \mathbf{J}^\dag \pi$.
The quotient of  the extended dual space with respect to the addition of constants provides the covariant phase space $\mathcal{P}(\mathbb{E})$ of the theory, which turns out to be again a fibre bundle both over $\mathbb{E}$ and over $\mathscr{M}$. 
As adapted fibered chart on $\mathcal{P}(\mathbb{E})$ we take $(U_{\mathcal{P}(\mathbb{E})},\, \psi_{U_{\mathcal{P}(\mathbb{E})}})$, 
\be
\label{28.7.1}
\psi_{U_{\mathcal{P}(\mathbb{E})}}(\mathfrak{p}) = (x^0,...,x^d,\,u^1,...,u^r,\, \rho_1^0,...,\rho_r^d \,),
\ee
 with $\mathfrak{p}$  a point in $U_{\mathcal{P}(\mathbb{E})} \subseteq \mathcal{P}(\mathbb{E})$.
The fibre bundles introduced so far fit into the diagram
\be
\begin{tikzcd}
\mathbf{J}^\dag \pi \arrow[rr, "\kappa", bend left = 40] \arrow[dr, "\tau_1", swap] \arrow[r, "\tau^1_0"]   & \mathbb{E} \arrow[d, "\pi"] & \mathcal{P}(\mathbb{E}) \arrow[l, "\delta^1_0", swap] \arrow[dl, "\delta_1"] \\
 & \m & 
\end{tikzcd} \,,
\ee
where all arrows represent (smooth, local) submersions and the map $\kappa\colon\mathbf{J}^\dagger\pi\,\to\,\mathcal{P}(\mathbb{E})$ is the projection associated to the quotient procedure described above.

\noindent An interesting interpretation of such dual bundles in terms of spaces of differential forms on $\mathbb{E}$ exists. 
Indeed, the fibration  $\tau_0^1\colon\mathbf{J}^\dag \pi\to\mathbb{E}$ is proven to be  isomorphic (under an isomorphism depending on the choice of a reference volume form on $\mathscr{M}$ that we denote by $vol_{\mathscr{M}}$) to the fibration  $$\nu\colon\Lambda^{d+1}_1(\mathbb{E})\to\mathbb{E},$$ where $\Lambda_k^m(\mathbb{E})$ is the subbundle $\Lambda^m(\mathbb{E})$ of $m$-forms on $\mathbb{E}$ which vanish when contracted along $k+1$ vector fields which are vertical with respect to the fibration $\pi\colon\mathbb{E}\to\mathscr{M}$\footnote{Notice that such forms are often referred to as $k$-semi-basic or $k$-horizontal on $\pi\colon\mathbb{E}\to\mathscr{M}$ in the literature on the subject.}.  
The elements in $\Lambda^{d+1}_1(\mathbb{E})$ can be locally written as 
\begin{equation} \label{Eq: canonical one-form extended dual}
\boldsymbol{w} \,=\, \rho^\mu_a \, \mathrm{d}u^a \wedge \mathrm{d}^{d} x_\mu + \rho_0 \, vol_{\mathscr{M}}\,,
\end{equation}
where $\mathrm{d}^{d} x_{\mu}\,=\, i_{\frac{\partial}{\partial x^\mu}}vol_{\mathscr{M}}$ and $vol_{\mathscr{M}}=\dd x^0\wedge\ldots\wedge\dd x^d$ is the volume introduced in \eqref{vom}.  
It is then possible to prove that sections of the fibration $\delta^1_0\,\colon\,\mathcal{P}(\mathbb{E})\,\to\,\mathbb{E}$ can be identified with the quotient $\Lambda^{d+1}_1(\mathbb{E}) / \Lambda^{d+1}_0(\mathbb{E})$, where elements in  $\Lambda^{d+1}_0(\mathbb{E})$ can be locally written as 
$\boldsymbol{w}_0 = \rho_0vol_{\mathscr{M}}$, so that one can locally represent such sections as
\be
\label{29.7.4}
\mathbf{w}\,=\, \rho^\mu_a \, \mathrm{d}u^a \wedge \mathrm{d}^{d} x_\mu.
\ee

\noindent The  extended dual space $\mathbf{J}^\dag \pi$ carries a canonical $(d+1)$-form, also called \textit{tautological form}, defined as $$\Theta_{\boldsymbol{w}}(X_1,..., X_{d+1}) \,=\, \boldsymbol{w}(\nu_\star X_1,..., \nu_\star X_{d+1}),$$ where $X_j\in\mathfrak{X}(\mathbf{J}^\dagger\pi)$ and $\boldsymbol{w}\in\mathbf{J}^\dagger\pi$, while  $\nu$ denotes the projection $\nu \;:\; \Lambda^{d+1}_1(\mathbb{E}) \to \mathbb{E}$ introduced above. Its local expression coincides with the local  expression \eqref{Eq: canonical one-form extended dual}.

\noindent It is now possible to introduce the notion of Hamiltonian  as a local function $H=H(\mathfrak{p})$  defined on the  covariant phase space $\mathcal{P}(\mathbb{E})$. This definition amounts to select a hypersurface on the extended dual space $\mathbf{J}^\dagger\pi$  which is transverse to the projection $\kappa$ (i.e. a section of the bundle $\kappa  \colon \mathbf{J}^\dag \pi \to \mathcal{P}(\mathbb{E})$), by means of
$$
\rho_0 = H(\mathfrak{p}) \, .
$$
Notice that, since $\rho_0$ is the local coordinate of the  inhomogeneous term of an affine map, the Hamiltonian turns to be defined up to an additive constant. The pullback of the tautological $(d+1)$-form on $\mathbf{J}^\dag \pi$ to the covariant phase space $\mathcal{P}(\mathbb{E})$ via (minus)\footnote{This sign is a matter of convention.} the Hamiltonian section determines a $(d+1)$-form on $\mathcal{P}(\mathbb{E})$:
\be \label{Eq: multisymplectic form}
\Theta_{\mathrm{H}} \,=\, \rho^\mu_a \, \mathrm{d}u^a \wedge \mathrm{d}^{d} x_\mu -H \, vol_{\mathscr{M}} \,.
\ee
Notice that in the case of particle dynamics, where $\m$ is an interval of the real line, the differential form above is the well known Poincar\'e-Cartan form:
\be
\Theta_H \,=\, p_j \dd q^j - H \dd t \,.
\ee
The $(n+1)$-form $\dd \Theta_{\mathrm{H}}$ is  closed and non-degenerate, and therefore defined as a multisymplectic (or $n$-plectic) form. 
Thus, the covariant phase space $\mathcal{P}(\mathbb{E})$ equipped with the $(n+1)$-form $\dd \Theta_{\mathrm{H}}$ is a multisymplectic manifold (hence the name of multisymplectic formalism).



\section{Schwinger-Weiss variational principle, the first fundamental formula and the solution space}
\label{Sec: Schwinger-Weiss variational principle, the first fundamental formula and the solution space}

\noindent In order to introduce the space of fields and the variational conditions in terms of which we describe the dynamical content of a theory within a multisymplectic formalism, we begin by noticing that a section  $\phi$ of the bundle $\pi\,\colon\, \mathbb{E}\, \to\, \mathscr{M}$ has a canonical lift to a section $\mathbf{j}^1\phi$ of the first order jet bundle fibration (see \eqref{1jdia}) $\pi_1\,\colon\,\mathbf{J}^1\pi\,\to\,\mathscr{M}$ (such a lift being given  in terms of the relations \eqref{1jcon}), but has no canonical lift to a section of the fibration  $\delta_1\colon\mathcal{P}(\mathbb{E})\to\mathscr{M}$. This  mimics -- and generalises to field theory -- what is well known to appear in classical mechanics, namely that, for a given a smooth curve on a configuration manifold $Q$, the velocity field along the curve is canonically defined while the values of the momenta are not. 
We indeed notice that  sections of the fibration  $\delta_1\,\colon\,\mathcal{P}(\mathbb{E})\,\to\,\mathscr{M}$ split into a section $\phi$ of the fibration  $\pi\,\colon\,\mathbb{E}\,\to\,\mathscr{M}$  and a section $P$ of the fibration $\delta^1_0\,\colon\,\mathcal{P}(\mathbb{E})\,\to\,\mathbb{E}$  in the sense of the following proposition.
\begin{proposition}
Given a section  $\chi$ for the fibration $\delta_1\,\colon\,\mathcal{P}(\mathbb{E})\,\to\,\mathscr{M}$, and given the section $\phi \,=\, \delta^1_0 \circ \chi$ for the fibration  $\pi\,\colon\,\mathbb{E}\,\to\,\mathscr{M}$, there exists a (not unique) section $P$ for the fibration  $\delta^1_0\,\colon\,\mathcal{P}(\mathbb{E})\,\to\,\mathbb{E}$ such that $\chi \,=\, P \circ \phi$.
\end{proposition}
\noindent  Since for the elements in the space $\Gamma(\delta_1)$ of (suitably regular) sections for the fibration $\delta_1\,\colon\,\mathcal{P}(\mathbb{E})\,\to\,\mathscr{M}$ the notion of momenta of a given configuration field is not uniquely defined, we consider in this paper the space $\Gamma(\pi) \times \Gamma(\delta^1_0)$ given by pairs $\chi\,=\,(\phi, P)$ where $\phi\in\Gamma(\pi)$ and $P\in\Gamma(\delta^1_0)$. Such sections have the expression  
$$
\phi^a\,=\,u^a \circ \phi, \qquad\qquad P^\mu_a\,=\,\rho^\mu_a \circ P 
$$
along the local coordinates adopted above. The reason for this choice is twofold. The first is that, as will be clear in the following lines, the evolution equations that come as the extrema of a variational problem in  $\Gamma(\pi) \times \Gamma(\delta^1_0)$ coincide with those arising from $\Gamma(\delta_1)$. The second is that for elements of $\Gamma(\pi) \times \Gamma(\delta^1_0)$ one can consider a class of constraints involving the space of momenta $P$ alone, and this will play an important role when dealing with non abelian Yang-Mills theories and the Palatini's formulation for the Einstein's equations for the gravitational field.

\noindent The dynamical content of the theory is encoded into an Action functional which is defined as follows upon a set $\tilde{\Gamma}(\pi) \times \tilde{\Gamma}(\pi^1_0) \subseteq\Gamma(\pi) \times \Gamma(\delta^1_0)$ of elements of $\Gamma(\pi) \times \Gamma(\delta^1_0)$ with suitable further regularity conditions\footnote{As it will be clear in the examples considered along the manuscript, such regularity conditions have to be imposed in order for the action functional, which is defined by means of an integral, to be well defined.}. 
\begin{definition}[\textit{Action functional}]
The action functional of a first order Hamiltonian field theory is the map $\mathscr{S}\,\colon\,\tilde{\Gamma}(\pi) \times \tilde{\Gamma}(\pi^1_0)\,  \to\,  \mathbb{R}$ given by 
\be \label{Eq: action functional field theory}
 \chi\quad\mapsto\quad \mathscr{S}_{\chi} =  \int_{\mathscr{M}} \chi^\star \Theta_{\mathrm{H}} \,,
\ee
whose coordinate expression turns out to be
\be
\mathscr{S}_\chi \,=\, \int_{\mathscr{M}} vol_{\mathscr{M}}\left[ \, P^\mu_a(x,\,\phi(x)) \partial_\mu \phi^a(x) - H(x, \phi(x), P(x, \phi(x))) \,\right] \,.
\ee
\end{definition}

\noindent Notice that the set $\tilde{\Gamma}(\pi) \times \tilde{\Gamma}(\pi^1_0)$ is not necessarily a Banach manifold\footnote{Notice that a general theory on equipping the space of smooth sections of a fibre bundle with a Banach norm,  even locally, does not exist. Consider for example the case of a trivial vector bundle over a non-compact manifold with fibre given by  $\mathbb{R}$. The space of smooth sections of the bundle is isomorphic to the space of smooth real valued functions on the base manifold, which is a 
Fr\'echet space on which one can not define a Banach norm.}, a requirement which allows for the definition of vector fields, differential forms and a general exterior Cartan calculus on the space of fields $\chi$. Although milder conditions than those of a Banach manifold structure   (see \cite{Michor-Kriegl1997-Convenient_setting}) would allow for the definition of an exterior Cartan calculus, we assume to equip the set   $\tilde{\Gamma}(\pi) \times \tilde{\Gamma}(\pi^1_0)$  with a Banach norm under which the action functional turns to be continuous so that it can be extended  by continuity to the completion $\fpe$.  
We refer to such a completion $\fpe$ as the set of \textit{dynamical fields} of the theory. Its elements, with a slight notational abuse, are again written as $\chi=(\phi, P)$; the extension of the action functional $\ac$ to $\fpe$ is as well denoted with the same symbol. Moreover, we define the set $\mathcal{F}_{\mathbb{E}}$ as the set of $\phi$ (the \textit{configuration fields}) from $\chi=(\phi,P)$. We shall develop the content of this paper coherently with this assumptions, and prove it right (with a suitable definition of a norm) for the examples we consider. This strategy allows also not to deal with general properties of Fr\'echet and locally convex spaces.

\noindent For the Banach manifold $\fpe$ one has a well defined notion (see \cite{Abraham-Marsden-Ratiu1980-Manifolds, Michor-Kriegl1997-Convenient_setting}) of tangent space $\mathbf{T}_\chi \fpe$ at each point $\chi$, whose elements are denoted by $\mathbb{X}_\chi$ and are given in terms of infinitesimal variations within $\fpe$ around $\chi$. It turns out that elements $\mathbb{X}_\chi$  can be identified with the set of $\delta_1$-vertical vector fields $X\,\in\,\mathfrak{X}(\pe)$ along the image of $\chi$ which, in order to preserve the split structure of $\fpe$, are $\delta^1_0$-projectable on $\mathbb{E}$, so that they can be locally written, in terms of the local chart \eqref{28.7.1}, as 
\be
\label{28.7.2}
X\,=\,X_{u^a}(x,u)\,\frac{\de}{\de u^a}\,+\,X_{\rho^\mu_a}(x,u,\rho)\frac{\de}{\de \rho^{\mu}_a}.
\ee
 The ``trajectories'' of the physical system under investigation are defined as the extrema of $\ac$. 
In particular they can be obtained as solutions of a set of pde by imposing a Schwinger-Weiss variational principle. 
It states that the first variation of $\mathscr{S}$ along any direction can only depend on boundary terms.
\begin{definition}[\textit{First variation of the action functional}]
Given the action functional $\mathscr{S}$ on $\fpe$ and a tangent vector $\mathbb{X}_\chi \in \mathbf{T}_\chi \mathcal{F}_{\mathcal{P}(\mathbb{E})}$, the first variation of $\mathscr{S}$ along the direction $\mathbb{X}_\chi$ is defined to be the Gateaux derivative of $\mathscr{S}$ along the direction $\mathbb{X}_\chi$.
\end{definition}
\noindent Denote by $F_s^X$  the flow of a vector field $X$ on  $\mathcal{P}(\mathbb{E})$ which is vertical with respect to the fibration $\delta_1\,\colon\,\mathcal{P}(\mathbb{E})\,\to\,\m$ and is defined in an open  neighborhood $U^{(\chi)}$ of the image of $\chi$ in $\mathcal{P}(\mathbb{E})$ such that its restriction to  $\chi$ coincides with $\mathbb{X}_\chi$, and denote by  $\chi_s$ the one parameter family of sections generated by the action of the flow $F^X_s$ upon $\chi=\chi_{s=0}$.  Recalling the contravariance of the pullback with respect to the composition of maps,  one has\footnote{By $\chi_s$ we denote the one parameter family of sections in terms of which the Gateaux derivative is defined.} 
\begin{align}
\label{Eq:gateaux derivative}
\delta_{\mathbb{X}_\chi} \mathscr{S}_\chi &\,=\, \frac{d}{ds} \,\mathscr{S}_{\chi_s} \, \Bigr|_{s=0} \nn \\ & \,=\, \frac{d}{ds} \, \int_{\m} \, \chi_s^\star \, \Theta_{\mathrm{H}} \, \Bigr|_{s=0} \nn \\ &\,=\, \frac{d}{ds}\, \int_{\m} \, \left(\, F^{X}_s \circ \chi  \,\right)^\star \, \Theta_{\mathrm{H}} \, \Bigr|_{s=0} \nn \\ &\,=\, \frac{d}{ds}\, \int_{\m} \, \chi^\star \circ {F^{X}_s}^\star \, \Theta_{\mathrm{H}} \, \Bigr|_{s=0} \,=\, \int_{\m} \, \chi^\star \left[ \, \mathfrak{L}_X \, \Theta_{\mathrm{H}} \,\right]    \,.
\end{align}
The previous expression, by means of the Cartan identity for the Lie derivative, gives rise to two terms:
\be \label{Eq: first variation field theory}
\delta_{\mathbb{X}_\chi} \mathscr{S}_\chi \,=\, \int_{\mathscr{M}} \chi^\star \left[\, i_{X} \mathrm{d} \Theta_H \,\right] + \int_{\partial \mathscr{M}} \chi_{_{\partial \mathscr{M}}} ^\star \left[\, i_{X} \Theta_H \,\right] \, ,
\ee
with 
\be
\label{emb1}
\mathfrak{i}_{_{\partial \mathscr{M}}}\,:\,\partial \mathscr{M}\,\hookrightarrow\,\mathscr{M}
\ee the embedding of the boundary, and 
\be
\label{reza}
\chi_{_{\partial \mathscr{M}}} = \chi \circ  \mathfrak{i}_{_{\partial \mathscr{M}}}
\ee  the restriction of the field $\chi$ to the boundary.  
It is worth noting that the pullback map $\chi^\star$ restricts the bracketed expression to the image of the section  $\chi$ (or its restriction to the boundary). 
Therefore, such a variation only depends on $\mathbb{X}_\chi$ and not on the chosen extension.
Then, a Schwinger-Weiss-like principle imposes the first term on the r.h.s. of the relation \eqref{Eq: first variation field theory} to vanish for all $X \in \mathfrak{X}(U^{(\chi)})$, where $U^{(\chi)}$ is an open neighborhood in $\mathcal{P}(\mathbb{E})$ of the image of $\chi$ and  $\mathfrak{X}(U^{(\chi)})$ can be identified with the set of vector fields on $U^{(\chi)}$ which are vertical along the fibration $\delta_1\,\colon\,\mathcal{P}(\mathbb{E})\,\to\,\mathscr{M}$ (and projectable along the fibration $\delta^1_0\,:\,\mathcal{P}(\mathbb{E})\,\to\,\mathbb{E}$). This condition is written as
\be
\chi^\star \left[\, i_X \dd \Theta_{\mathrm{H}} \,\right] \,=\, 0 \;\;\;\; \forall \,\, \mathfrak{X}(U^{(\chi)}) \,.
\ee
An interpretation of the terms appearing in \eqref{Eq: first variation field theory} in terms of differential forms on the infinite dimensional manifold $\mathcal{F}_{\mathcal{P}(\mathbb{E})}$ exists. Since  $\mathscr{S}\,:\,\mathcal{F}_{\mathcal{P}(\mathbb{E})}\,\to\,\mathbb{R}$, the term on the l.h.s. in \eqref{Eq: first variation field theory} represents  the ``component along $\mathbb{X}_\chi$'' of the differential of $\mathscr{S}$ at the point $\chi$, that is 
$$
\delta_{\mathbb{X}_\chi} \mathscr{S}_\chi \,=\,(i_{\mathbb{X}} \, \mathrm{d} \mathscr{S})\Bigr|_{\chi}, 
$$
where $\mathrm{d}$ denotes the differential on $\fpe$. Analogously,
the first term in the r.h.s. of \eqref{Eq: first variation field theory} can be read as the evaluation at the point $\chi\,\in\,\fpe$ of the contraction along the tangent vector $\mathbb{X}_\chi\,\in\,\mathbf{T}_\chi\mathcal{F}_{\mathcal{P}(\mathbb{E})}$ of a differential one-form $\mathbb{EL}\,\in\,\Lambda^1(\mathcal{F}_{\mathcal{P}(\mathbb{E})})$, that we call the Euler-Lagrange form of the given theory. 
\begin{definition}[\textit{Euler-Lagrange form}] \label{def:2.3}
Given an action functional $\mathscr{S}$, the first term appearing in the first fundamental formula \eqref{Eq: first variation field theory} amounts to contract the following \textit{Euler-Lagrange form} along the tangent vector $\mathbb{X}_\chi\,\in\,\mathbf{T}_\chi\mathcal{F}_{\mathcal{P}(\mathbb{E})}$ at the point $\chi$
\be \label{Eq: Euler-Lagrange form}
(i_{\mathbb{X}}(\mathbb{EL}))\bigr|_{\chi}\,=\,(\mathbb{EL}(\mathbb{X}))\bigr|_\chi \,=\, \int_{\mathscr{M}} \chi^\star \left[\, i_{X} \mathrm{d} \Theta_H \,\right] \,,
\ee
where $X\,\in\,\mathfrak{X}(U^{(\chi)})$  is an extension of $\mathbb{X}_\chi$.
\end{definition}
\noindent For the second term in the r.h.s. of the first fundamental formula \eqref{Eq: first variation field theory}  an analogous interpretation in terms of  differential forms on a suitable space of sections exists,  but requires more comments that we itemize as follows.
\begin{itemize}
\item The first step  is to consider again the relation \eqref{emb1}, i.e. the embedding of the boundary $\de\m$ into $\m$ as
$$
\mathfrak{i}_{_{\de\m}}\,\colon\,\de\m\,\hookrightarrow\,\m,
$$
so to have the fibrations given as in the following diagram 
\be
\label{diazamp}
\begin{tikzcd}
\mathfrak{i}^{\star}_{_{\partial\m}}(\mathcal{P}(\mathbb{E})) \arrow[d, "\tilde{\delta}^1_0" ] & 
\mathcal{P}(\mathbb{E})  \arrow[d, "\delta^1_0"] & \mathcal{P}(\mathfrak{i}^\star_{_{\de \m}}(\mathbb{E}))  \arrow[d, "(\delta_{_{\de\m}})^{1}_0"] \\
 \mathfrak{i}^\star_{_{\partial\m}}(\mathbb{E})  \arrow[d, "\pi_{_{\de\m}}"] & \mathbb{E} \arrow[d, "\pi"] & \mathfrak{i}^\star_{_{\partial\m}}(\mathbb{E})  
 \arrow[d, "\pi_{_{\de\m}}"] \\
   \de\m \arrow[r, hook, "\mathfrak{i}_{_{\de\m}}"] & \m & \de\m \arrow[l, hook, "\mathfrak{i}_{_{\de\m}}"] 
\end{tikzcd}\,. 
\ee
The left column represents the (double) pullback bundle over $\de\m$  (see \cite{Saunders1989-Jet_Bundles}) associated to the  phase space fibration represented by the central column: the typical fibre of the bundle 
$$
\pi_{_{\de \m}}\,\colon\,\mathfrak{i}^\star_{_{\partial\m}}(\mathbb{E})\,\to\,\de\m
$$ 
is diffeomorphic to the typical fibre of the bundle $\pi\,\colon\,\mathbb{E}\,\to\,\m$, and analogously the typical fibre for the bundle 
$$
\tilde\delta^1_0\,\colon\,\mathfrak{i}^{\star}_{_{\partial\m}}(\mathcal{P}(\mathbb{E})) \,\to\,\mathfrak{i}^\star_{_{\partial \m}}(\mathbb{E})
$$
is diffeomorphic to the typical fibre of the bundle $\delta^1_0\,\colon\,\mathcal{P}(\mathbb{E})\,\to\,\mathbb{E}$.  The right column represents the phase space fibration associated to the pullback bundle $\pi_{_{\de \m}}\,\colon\,\mathfrak{i}^\star_{_{\partial\m}}(\mathbb{E})\,\to\,\de\m$, in analogy to what  we described in the previous section \ref{Sec: Multisymplectic formulation of first order Hamiltonian field theories}.

\noindent Notice that, because of the global properties of the manifold $\m$ described in the previous section,  the  embedding $\mathfrak{i}_{_{\de\m}}\,\colon\,\de\m\,\hookrightarrow\,\m$ can be written as a slicing, that is the set $(x^0,x^1,\ldots, x^d)$  (see \eqref{a267.1}) provides a local chart for a collar in $\m$ close to the boundary $\partial\m$, and its coordinate expression comes upon  fixing the value of the coordinate $x^0$ (as in \eqref{emb2})
\be
\label{29.7.1}
(x^1, \dots,x^d)\qquad\mapsto\qquad(x^0=\bar{x}^0, x^1, \ldots, x^d)
\ee
so that the set  $(x^1, \ldots, x^d)$ is a local chart for the boundary and $x^0$ is transversal to it.  
It is  easy to observe that each $\chi\in\Gamma(\pi) \times \Gamma(\delta^1_0)$ locally defined on a  collar close to $\de\m$ provides an element 
$\chi_{_{\de\m}}\in \Gamma(\pi_{_{\de\m}})\times\Gamma(\tilde{\delta}^1_0)$ where $\tilde{\delta}_1\,=\,\pi_{_{\de\m}}\circ\tilde{\delta}^1_0\,\colon\,\mathfrak{i}^{\star}_{_{\partial\m}}(\mathcal{P}(\mathbb{E}))\,\to\,\de\m$, and for it we can write 
\be
\label{vesu1}
\chi_{_{\de\m}}\,=\,(\phi\mid_{_{\partial\mathscr{M}}},\,\, P \mid_{_{\partial\mathscr{M}}})
\ee
with 
\be
\begin{split}
&\phi\mid_{_{\partial\mathscr{M}}}\,=\,\phi\circ\mathfrak{i}_{_{\de\m}}\,=\,\mathfrak{i}^\star_{_{\de\m}}\phi, \\ 
&  P \mid_{_{\partial\mathscr{M}}}\,=\, \phi\circ\mathfrak{i}_{_{\de\m}}\,=\,\mathfrak{i}^\star_{_{\de\m}}P.
\end{split}
\label{vesu2}
\ee
It is then natural to define the set $\fpe^{\de\m}\subseteq\Gamma(\pi_{\de \m}) \times \Gamma(\delta^1_0)$ as the set of {\it restrictions to the boundary}  $\de\m$ of the set of dynamical fields $\fpe$ on $\m$, as well as the set $\mathcal{F}_{\mathbb{E}}^{\de\m}\subseteq\Gamma(\pi_{_{\de\m}})$ as the set of {\it restrictions to the boundary} $\de\m$ of the set of configuration fields $\mathcal{F}_{\mathbb{E}}$ on $\m$.

\medskip

\item The second step starts by noticing that, via  the embedding map as in \eqref{29.7.1}, one has  a volume form on $\de\m$ as 
$$
vol_{\m}\,=\,\dd x^0\wedge vol_{\de\m}:
$$
in terms of such a volume form on $\de\m$, as already described in the previous pages, the sections of the bundle 
\be
\label{29.7.2}
(\delta_{_{\de\m}})^1_0\,\colon\,\mathcal{P}(\mathfrak{i}_{_{\de\m}}^\star(\mathbb{E}))\,\to\,
\mathfrak{i}_{_{\de\m}}^\star(\mathbb{E})
\ee
 can be identified with the quotient $\Lambda^{d}_1(\mathfrak{i}_{_{\de\m}}^\star\mathbb{E})/\Lambda^d_0(\mathfrak{i}_{_{\de\m}}^\star\mathbb{E})$, whose elements can be written, with respect to the local action given in \eqref{29.7.1}, as
 \be
 \label{29.7.3}
 \tilde{\mathbf{w}}\,=\,\rho_{a}^j\dd u^a\wedge {i_{\de_j}}\dd x^1\wedge\ldots\wedge\dd x^d,
 \ee
  for $a=1,\dots,r$ and $j=1,\ldots,d$, where $\rho^j_a\,\in\,\mathcal{F}_{\mathbb{E}}^{\de\m}$. Moreover, if $\mathbf{w}\in\Lambda^{d+1}_1(\mathbb{E}) / \Lambda^{d+1}_0(\mathbb{E})$ as in \eqref{29.7.4}, one has
 \be
 \label{29.7.5}
 \mathfrak{i}^\star_{_{\de\m}}(\mathbf{w})\,=\,\rho_a^0\dd u^a\wedge vol_{_{\de\m}},
 \ee
 with again $\rho^0_a\,\in\,\mathcal{F}_{\mathbb{E}}^{\de\m}$. The  bundle on the basis $\mathfrak{i}^\star_{_{\de\m}}(\mathbb{E})$ with typical fibre diffeomorphic to $\mathfrak{i}_{_{\de\m}}^\star(\mathbf{w})$ turns to be diffeomorphic (see \cite{Ibort-Spivak2017-Covariant_Hamiltonian_YangMills}) to the cotangent bundle 
 $$
 \boldsymbol{\pi}_{_{\de\m}}\,\colon\,
 \mathbf{T}^\star\mathcal{F}_{\mathbb{E}}^{\de\m}\,\to\,\mathcal{F}_{\mathbb{E}}^{\de\m}.
 $$
 Both $\boldsymbol{\pi}_{_{\de\m}}$  
 {\it and} the bundle $(\delta_{_{\de\m}})^1_0$ in \eqref{29.7.2} share the same basis manifold and are subbundles of the fibration $\tilde\delta^1_0\,\colon\,\mathfrak{i}^{\star}_{_{\partial\m}}(\mathcal{P}(\mathbb{E})) \,\to\,\mathfrak{i}^\star_{_{\partial\m}}(\mathbb{E})$. For them one has (in terms of the notion of sum of vector bundles, see \cite{Saunders1989-Jet_Bundles})
\be
\label{29.7.6}
\tilde\delta^1_0\,=\,(\delta_{_{\de\m}})^1_0\,\oplus\,\boldsymbol{\pi}_{_{\de\m}}\,. 
 \ee
 Such a splitting allows for a splitting of the set of momenta $\{P^{\mu}_a|_{_{\de\m}}\}$ (with $\mu=0,1,\ldots,d)$ introduced in \eqref{vesu2} into a set $\{P^0_a|_{_{\de\m}}\}$ of momenta on the boundary $\de\m$ relative to the $x^0$ coordinate on $\m$ transversal to the boundary $\de\m$, and a complementary set of momenta  on the boundary $\{P^j_a|_{_{\de\m}}\}$ with $j=1,\ldots,d$. One can then recast the relations \eqref{vesu1}-\eqref{vesu2} in terms of the map
 $\Pi_{_{\partial \mathscr{M}}}\,:\,\mathcal{F}_{\mathcal{P}(\mathbb{E})} \,\to\, \mathcal{F}_{\mathcal{P}(\mathbb{E})}^{\partial \mathscr{M}}$ of the fields to the boundary\footnote{Notice that one can write  $\Pi_{_{\partial\m}}=i_{_{\partial\m}}^\star$ when acting on functions, in terms of the pull-back of the embedding map $i_{_{\partial\m}}\,\colon\,\partial\mathscr{M}\,\hookrightarrow\,\mathscr{M}$.} 
$$
\Pi_{_{\partial \mathscr{M}}}\,:\,(\phi , P) \,\mapsto\,  (\phi\mid_{_{\partial\mathscr{M} }}, P_a^\mu \mid_{_{\partial\mathscr{M}} } ) 
$$
so to write 
\be
\label{ale2}
(\phi\mid_{_{\partial\mathscr{M}}}, P_a^\mu \mid_{_{\partial\mathscr{M}}} ) \,=\, \left(\phi\mid_{_{\partial\mathscr{M}}}, P_a^0 \mid_{_{\partial\mathscr{M}}}, P_a^j \mid_{_{\partial\mathscr{M}}} \right)
\ee
and 
  $$
\mathcal{F}_{\mathcal{P}(\mathbb{E})}^{\partial \mathscr{M}} \,=\, \mathcal{F}_{\mathcal{P}(\mathbb{E}),\,0}^{\partial \mathscr{M}} \times \mathscr{B}^{\partial \mathscr{M}}.
$$
Upon denoting   
$$
\left(\phi\mid_{_{\partial\mathscr{M}}}, P_a^0 \mid_{_{\partial\mathscr{M}}} \right) \,=\, (\varphi^a,\, p_a) \in \mathcal{F}_{\mathcal{P}(\mathbb{E}),\,0}^{\partial \mathscr{M}}$$ and $$P_a^j \mid_{_{\partial\mathscr{M}} } \,=\, \beta_a^j \in \mathscr{B}^{\partial \mathscr{M}},
$$
one has the isomorphism 
$\lambda\,:\,\mathcal{F}_{\mathcal{P}(\mathbb{E}),\,0}^{\partial \mathscr{M}} \to \mathbf{T}^\star \mathcal{F}_{\mathbb{E}}^{\partial \mathscr{M}}$ (which is  continuous since it is  an identification map)  which we write as 
\be \label{Eq: isomorphism F0 T*FQ}
\lambda \;\;:\;\; \,(\varphi,\, p) \,\mapsto\, (\varphi,\, \varrho) \,
\ee
where we have denoted as $(\varphi, \rho)$ elements in $\mathbf{T}^\star \mathcal{F}_{\mathbb{E}}^{\partial \mathscr{M}}$ and a fibered chart on such a bundle as 
 $(U_{\mathbf{T}^\star \mathcal{F}_{\mathbb{E}}^{\partial \mathscr{M}}},\, \psi_{U_{\mathbf{T}^\star \mathcal{F}_{\mathbb{E}}^{\partial \mathscr{M}}}})$ with 
$$
\psi_{U_{\mathbf{T}^\star \mathcal{F}_{\mathbb{E}}^{\partial \mathscr{M}}}}(\varphi,\, \varrho) \,=\, (\varphi^1,...,\varphi^r,\, \varrho_1,...,\varrho_r).
$$
We can depict the various maps introduced above as  
\be \label{Eq: projections to boundary (fields)}
\begin{tikzcd}
	\mathcal{F}_{\mathcal{P}(\mathbb{E})} \arrow[r, "\Pi_{_{\partial \mathscr{M}}}"] & 
	\mathcal{F}_{\mathcal{P}(\mathbb{E})}^{\partial \mathscr{M}} \,\simeq\,\left(\mathcal{F}_{\mathcal{P}(\mathbb{E}),\,0}^{\partial \mathscr{M}} \times \mathscr{B}^{\partial \mathscr{M}}\right) \arrow[r, "\rho_{_{\partial \mathscr{M}}}"] & \mathcal{F}_{\mathcal{P}(\mathbb{E}),\,0}^{\partial \mathscr{M}} \arrow[r, "\lambda"] & \mathbf{T}^\star \mathcal{F}_{\mathbb{E}}^{\partial \mathscr{M}}
\end{tikzcd}.
\ee

\medskip

\item The third step consists in noticing that the cotangent bundle
$\mathbf{T}^\star \mathcal{F}_{\mathbb{E}}^{\partial \mathscr{M}}$ has a canonical non-degenerate $2$-form, that we denote by $\bar{\omega}^{\partial \mathscr{M}}$, whose action is 
\be
\label{omega1}
\bar{\omega}_{(\varphi,\, \varrho)}^{\partial \mathscr{M}}(\mathbb{X}_{(\varphi,\, \varrho)},\, \mathbb{Y}_{(\varphi,\, \varrho)}) \,=\, \int_{\partial \mathscr{M}}  vol_{\partial \mathscr{M}}\left(\,\mathbb{X}_{\varphi^a} \, \mathbb{Y}_{\varrho_a} - \mathbb{Y}_{\varphi^a} \, \mathbb{X}_{\varrho_a} \,\right)  \,,
\ee
where $\mathbb{X}_{(\varphi,\, \varrho)}$ and $\mathbb{Y}_{(\varphi,\, \varrho)}$ are elements in $\mathbf{T}_{(\varphi,\,\varrho)}(\mathbf{T}^\star \mathcal{F}_{\mathbb{E}}^{\partial \mathscr{M}})$, while  $\mathbb{X}_{\varphi^a}$, $\mathbb{X}_{\varrho_a}$ (resp. $\mathbb{Y}_{\varphi^a}$, $\mathbb{Y}_{\varrho_a}$) are their components with respect to the adopted chart.
It is immediate to see that $$\bar{\omega}_{(\varphi, \varrho)}^{\partial \mathscr{M}}\,=\,-\dd \bar{\vartheta}_{(\varphi, \varrho)}^{\partial \mathscr{M}},$$ with 
\be
\label{theta1}
\bar{\vartheta}_{(\varphi, \varrho)}^{\partial \mathscr{M}}(\mathbb{X}_{(\varphi,\,\varrho)}) \,=\, \int_{\partial \mathscr{M}}  vol_{\partial \mathscr{M}} \left(\, \varrho_a \, \mathbb{X}_{\varphi^a} \,\right)  \,.
\ee

\medskip

\end{itemize}

We can now directly compute  the second term on the r.h.s. of \eqref{Eq: first variation field theory}, which  reads
\be
\int_{{\partial \mathscr{M}}} \chi_{_{\partial \m}}^\star \left[\, i_{X} \Theta_H \,\right] \,=\,  \int_{\partial \mathscr{M}} vol_{\partial \mathscr{M}} \left(\, p_a\, X_{u^a} \,\right)\bigr|_{\chi_{_{\de \m}}}\,. 
\ee
Notice that, since tangent vectors $\mathbb{X}\in\mathfrak{X}(\fpe)$ are represented in terms of vector fields $X\in\mathfrak{X}(\mathcal{P}(\mathbb{E}))$ (as in \eqref{28.7.2}) which are $\delta^1$-vertical (and $\delta^1_0$-projectable), the terms $X_{u^a}\bigr|_{\chi_{_{\de\m}}}$ represent a vector field on $\de\m$ which is vertical along the $\tilde\delta^1$ pullback fibration defined  above on $\de\m$, and then give the $\mathbb{X}_{\varphi^a}$-components of a vector  $\mathbb{X}_{(\varphi, \rho)}$ on $\mathbf{T}_{(\varphi,\rho)} (\mathbf{T}^\star\mathcal{F}_{\mathbb{E}}^{\de\m})$. We can then write  
\be
\label{29.7.7}
\int_{\partial \mathscr{M}} vol_{\partial \mathscr{M}} \left(\, p_a\, X_{u^a} \,\right)\bigr|_{\chi_{_{\de \m}}} \,=\, \int_{\partial \mathscr{M}} vol_{\partial \mathscr{M}} \left(\, p_a\, \mathbb{X}_{\varphi^a} \,\right)  \,,
\ee

This shows that we can  interpret the term $\int_{\partial \mathscr{M}} \chi_{_{\partial \mathscr{M}}}^\star \left[\, i_{X} \Theta_H \,\right]$ appearing in the first fundamental formula as a differential form on $\mathcal{F}_{\mathcal{P}(\mathbb{E})}$, namely as the image, under the  the pull-back of (see \eqref{Eq: projections to boundary (fields)}) 
$$
\lambda \circ \rho_{_{\partial \mathscr{M}}} \circ \Pi_{_{\partial \mathscr{M}}}\,\colon\,\fpe\,\to\,\mathbf{T}^{\star}\mathcal{F}^{\de\m}_{\mathbb{E}},
$$ of the canonical form $\bar{\vartheta}^{\partial \mathscr{M}}$ on $\mathbf{T}^\star \mathcal{F}^{\partial \mathscr{M}}_{\mathbb{E}}$, whose action we can write as 

\be
\label{15ott1}
\begin{split} 
\int_{\partial \mathscr{M}} \chi_{_{\partial \mathscr{M}}}^\star \left[\, i_{X} \Theta_H \,\right] &\,=\, \left(\, \lambda \circ \rho_{_{\partial \mathscr{M}}} \circ \Pi_{_{\partial \mathscr{M}}} \,\right)^\star \bar{\vartheta}_\chi^{\partial \mathscr{M}}(\mathbb{X}_{\chi}) \nn \\ &\,=\, \Pi_{_{\partial \mathscr{M}}}^\star \, \rho_{_{\partial \mathscr{M}}}^\star \, \lambda^\star \bar{\vartheta}_\chi^{\partial \mathscr{M}}(\mathbb{X}_{\chi}) \nn \\  &\,=\,\Pi_{_{\partial \mathscr{M}}}^\star \, \rho_{_{\partial \mathscr{M}}}^\star \, \vartheta_\chi^{\partial \mathscr{M}}(\mathbb{X}_{\chi})\,=\, \Pi^\star_{_{\partial \mathscr{M}}} \alpha_\chi^{\partial \mathscr{M}}(\mathbb{X}_{\chi}) \, ,
\end{split}
\ee 
upon defining $\vartheta_\chi^{\partial \mathscr{M}}=\lambda^\star \bar{\vartheta}_\chi^{\partial \mathscr{M}}$ and $\alpha_\chi^{\partial \mathscr{M}}=\rho_{_{\partial \mathscr{M}}}^\star \, \vartheta_\chi^{\partial \mathscr{M}}$. Such analysis adds the diagram \eqref{Eq: projections to boundary (fields)} a new row, as follows
\be
\label{5.8.3}
\begin{tikzcd}
	\mathcal{F}_{\mathcal{P}(\mathbb{E})} \arrow[r, "\Pi_{_{\partial \mathscr{M}}}"] & 
	\mathcal{F}_{\mathcal{P}(\mathbb{E})}^{\partial \mathscr{M}}  \arrow[r, "\rho_{_{\partial \mathscr{M}}}"] & \mathcal{F}_{\mathcal{P}(\mathbb{E}),\,0}^{\partial \mathscr{M}} \arrow[r, "\lambda"] & \mathbf{T}^\star \mathcal{F}_{\mathbb{E}}^{\partial \mathscr{M}} \\ 
	\left\{\begin{array}{l}\Pi^\star_{_{\de\m}}\alpha^{\de\m}, \\ \Pi^\star_{_{\de\m}}\Omega^{\de\m}\end{array}\right\} & \left\{\begin{array}{l} \alpha^{\de\m}\,=\,\rho^\star_{_{\de\m}}\vartheta, \\ \Omega^{\de\m}\,=\,-\dd\alpha^{\de\m}\end{array}\right\} \arrow[l, "\Pi^\star_{_{\de\m}}"] &  \left\{\begin{array}{l} \vartheta^{\de\m}\,=\,\lambda^\star_{_{\de\m}}\bar\vartheta, \\ \omega^{\de\m}\,=\,-\dd\vartheta^{\de\m}\end{array}\right\} \arrow[l, "\rho^\star_{_{\de\m}}"] & 	 \left\{\begin{array}{l} \bar\vartheta^{\de\m}, \\ \bar\omega^{\de\m}\,=\,-\dd\bar\vartheta^{\de\m}\end{array}\right\}. \arrow[l, "\lambda^\star_{_{\de\m}}"] 
\end{tikzcd}
\ee

\noindent The variational formula \eqref{Eq: first variation field theory} can then be written in terms of differential forms on $\fpe$ as 
\be \label{Eq: first fundamental formula}
\mathrm{d} \mathscr{S}_{\chi} \,=\, \mathbb{EL}_\chi + (\Pi_{_{\partial \mathscr{M}}}^\star \alpha^{\partial \mathscr{M}})_\chi \,
\ee
so that  solutions of the equations of the motion can be seen as the zeroes of the Euler-Lagrange form:
\be
\mathbb{EL}_\chi \,=\, 0 \,,
\ee
a condition which is equivalent to
\be
\chi^\star \left[\, i_X \dd \Theta_{\mathrm{H}} \,\right] \,=\, 0 
\ee
for any $X\,\in\,\mathfrak{X}(U^{(\chi)})$ which represents an infinitesimal variation of the set of $\chi$'s, i.e. such that it is $\delta_1$-vertical and $\delta^1_0$-projectable. The local expression of such condition reads 
\be \label{Eq: covariant hamilton equations}
\frac{\partial \phi^a}{\partial x^\mu} - \frac{\partial H }{\partial \rho^\mu_a}\Biggr|_{\chi} \,=\, 0 \,, \qquad  \frac{\partial \rho^\mu_a}{\partial x^\mu} + \frac{\partial H}{\partial u^a}\Biggr|_\chi \,=\, 0 \,.
\ee
These are the  \textit{covariant Hamilton equations} of the theory, or the \textit{De Donder-Weyl equations}. 
We will equivalently refer to them as to the  \textit{equations of the motion}, or \textit{Euler-Lagrange equations}. 
It is then natural to set the following.
\begin{definition}[\textit{Solution space}]
Given an action functional $\mathscr{S}$ on a suitable space of fields $\fpe$, we define the following subset of $\mathcal{F}_{\mathcal{P}(\mathbb{E})}$
\be\label{eq:elag}
\elag_{\mathscr{M}} \,=\, \{\, \chi \,=\, (\phi,\, P) \;\; : \;\; \mathbb{EL}_{\chi} \,=\,0 \,\} \,=\, \left\{\,(\phi,\, P) \;\; : \;\; \chi^*(i_X \dd \Theta_{\mathrm{H}}) = 0 \;\;\; \forall \,\, X \in \mathfrak{X}(U^{(\chi)}) \right\} \,
\ee
to be the solutions space of the theory.
\end{definition}
\noindent We assume that the map  $\mathfrak{i}_{\elag_{\mathscr{M}}}\,\colon\,\elag_{\mathscr{M}}\,\hookrightarrow\,\fpe$ is a smooth embedding (and we shall prove such assumption right for the examples we shall analyse in the following pages). 
Before moving on to the next section we provide some explicit examples in order to clarify the formalism described so far. 
\begin{example}[\textit{Free point particle on the line}] \label{Ex: free particle variational principle}
The first example we consider is that of a point particle system seen as a field theory over a one dimensional manifold $\m$, which we assume to be an interval $\mathbb{I}\,=\,(a,b)\, \subset \mathbb{R}$ of the real line, equipped with a chart given by a single time coordinate $t$.
In particular we consider the dynamics of a free particle.
The bundle $\pi\,\colon\,\mathbb{E}\, \to\, \mathbb{I}$ has typical fibre $\mathcal{E} \,=\, \mathbb{R}$ where we denote by $q$ a (global) coordinate.
Therefore, configuration fields are real valued smooth functions $q \;\; : \;\; \mathbb{I} \to \mathbb{R}$, i.e. they are smooth curves with support on $\mathbb{I}\,\subset\,\mathbb{R}$ .
The covariant phase space is $$\pe \,=\, \mathbb{I} \times \mathbb{R}^2$$ with a set of (global) coordinates  denoted by $(t,\, q,\, p)$.
We notice that in this case the action functional is well defined on the whole set  $\Gamma(\pi) \times \Gamma(\delta^1_0)$ without assuming any additional regularity conditions, namely here $\tilde{\Gamma}(\pi) \times \tilde{\Gamma}(\delta^1_0) \,=\, \Gamma(\pi) \times \Gamma(\delta^1_0)$. 
The latter space can be given the structure of a Banach space.
Indeed, the 
norm 
\be
\Vert \, \chi \, \Vert \,=\, \mathrm{sup}_{t \in \mathbb{I}} \vert\, 
q(t) \,\vert + \mathrm{sup}_{t \in \mathbb{I}} \vert\, 
\dot{q}(t) \,\vert + \mathrm{sup}_{t \in \mathbb{I}} \vert\, 
p(t) \,\vert  \,
\ee
is well defined on it, and, as the following inequalities \footnote{Here we are denoting by $\chi\,=\,(q,\,p)$ and $\tilde{\chi} \,=\, (\tilde{q},\,\tilde{p})$ two elements of $\Gamma(\pi) \times \Gamma(\delta^1_0)$.}
\be
\begin{split}
\bigl|\,\ac_\chi - \ac_{\tilde{\chi}} \,\bigr|  \,\leq\, \int_a^b \dd t &\left\vert \, p(t) \dot{q}(t) - \tilde{p}(t) \dot{\tilde{q}}(t) - \frac{p(t)^2}{2m} + \frac{\tilde{p}(t)^2}{2m} \,\right\vert \,  \\ 
\,\leq\, \int_a^b  \dd t &\biggl[\, \left\vert \, 
(p(t) - \tilde{p}(t)) \dot{q}(t) \,\right\vert + \left\vert\, \tilde{p}(t) (\dot{q}(t) - \dot{\tilde{q}}(t)) \,\right\vert + \\
+ &\frac{1}{2m} \left\vert\, \tilde{p}(t) (p(t) - \tilde{p}(t)) \,\right\vert + \frac{1}{2m} \left\vert\, p(t) (p(t) - \tilde{p}(t)) \,\right\vert \,\biggr]  \\
\,\leq\, (b-a) &\biggl[\, \mathrm{sup}_{t \in \mathbb{I}} \vert\, \dot{q}(t) \,\vert \, \mathrm{sup}_{t \in \mathbb{I}} \vert\, p(t) - \tilde{p}(t) \,\vert + \mathrm{sup}_{t \in \mathbb{I}} \vert\, \tilde{p}(t) \,\vert \, \mathrm{sup}_{t \in \mathbb{I}} \vert\, \dot{q}(t) - \dot{\tilde{q}}(t) \,\vert + \\
+ &\frac{1}{2m} \mathrm{sup}_{t \in \mathbb{I}} \vert\, p(t) \,\vert \, \mathrm{sup}_{t \in \mathbb{I}} \vert\, p(t) - \tilde{p}(t) \,\vert + \frac{1}{2m} \mathrm{sup}_{t \in \mathbb{I}} \vert\, \tilde{p}(t) \,\vert \, \mathrm{sup}_{t \in \mathbb{I}} \vert\, p(t) - \tilde{p}(t) \,\vert \,\biggr] \,
\end{split}
\ee
show, the action functional is continuous with respect to the norm defined above.
Consequently, it can be extended by continuity to the completion $\overline{\Gamma(\pi) \times \Gamma(\delta^1_0)}^{\|\cdot\|}$, which turns to be  a Banach space and which we assume as   our space of dynamical fields $\fpe \,=\, \overline{\Gamma(\pi) \times \Gamma(\delta^1_0)}^{\Vert\,\cdot \,\Vert}$.
The Hamiltonian of the theory is a section $H\,:\,\pe \to \mathbf{J}^\dagger\pi$ which we write as 
\be
(t,\,q,\, p) \mapsto H(t,\,q,\,p) \,=\, \frac{p^2}{2m} \,,
\ee
and, thus, the canonical form $\Theta_H$ is
\be
\Theta_{H} \,=\, p \dd q - \frac{p^2}{2m} \dd t \,,
\ee
giving the multisymplectic form
\be
\dd \Theta_H \,=\, \dd p \wedge \dd q - \frac{p}{m} \dd p \wedge \dd t \,.
\ee

\noindent In order to compute the Euler-Lagrange form and to obtain the equations of the motion, let us compute the contraction of $\dd \Theta_H$ along a $\pi^1$-vertical vector field $X$, which we write as:
\be
i_X \dd \Theta_H \,=\, X_p \dd q - X_q \dd p - \frac{p}{m} X_p \dd t \,.
\ee
Its pull-back via $\chi$ reads
\be
\chi^\star \left[\,i_X \dd \Theta_H \,\right] \,=\, \left[\,X_p \left(\dot{q} - \frac{p}{m}\right) - X_q \dot{p} \,\right] \dd t  \,,
\ee
while the Euler-Lagrange form is 
\be
\mathbb{EL}_\chi(\mathbb{X}_\chi) \,=\, \int_{\mathbb{I}} \left[\,X_p \left(\dot{q} - \frac{p}{m}\right) - X_q \dot{p} \,\right] \dd t \,.  
\ee
Its zeroes satisfy the following set of Euler-Lagrange equations
\be\label{eqn: EL eqns free particle}
\begin{split}
\dot{q} \,=\, \frac{p}{m} , \\  \qquad \dot{p} \,=\, 0 \,, \end{split}
\ee
in complete accordance to  the usual geometrical description of the dynamics of a free  particle in classical mechanics.
\end{example}

\begin{example}[\textit{Free vector meson field theory}] \label{Ex: vector boson field variational principle}
As a further example, we consider  the set of pdes describing the evolution of the free (real) vector meson field. 
In this case, the space-time $\mathscr{M}$ is the Minkowski space-time, that we denote as   $\mathbb{M}^{1,3}$, i.e. $\mathbb{R}^4$ equipped with the Minkowski metric 
$$
\eta = \mathrm{d}x^0\otimes \mathrm{d}x^0 -\sum_{j=1}^3 \mathrm{d}x^j\otimes \mathrm{d}x^j, 
$$
in terms of a global chart coordinate system given by  $(\mathbb{M}^{1,3}\,, \psi_{\mathbb{M}^{1,3}})$, 
$$
\psi_{\mathbb{M}^{1,3}}(m)=(\,x^0,...,x^3\,)\,=\, x,
$$ for $m\,\in\,\mathbb{M}^{1,3}.$
The bundle $\pi\,\colon\,\mathbb{E}\, \to\, \mathbb{M}^{1,3}$ has typical fiber $\mathcal{E} \cong \mathbb{R}^r$, so the configuration fields can be identified with maps $\phi \colon \mathbb{M}^{1,3} \to \mathbb{R}^r$.
The covariant phase space results in  $$\mathcal{P}(\mathbb{E})=\mathbb{M}^{1,3}\times \left( \mathbb{R}\times\mathbb{R}^4 \right)^r,$$ where we chose the adapted fibered chart $(\mathcal{P}(\mathbb{E}),\, \psi_{\mathcal{P}(\mathbb{E})})$, $$\psi_{\mathcal{P}(\mathbb{E})}(\mathfrak{p})=\left(\,x^0,...,x^3,\, u^1,...,u^r,\, \rho_1^{0},...,\rho_3^r\,\right) \,=\, (x,\,u,\, \rho)$$ with $\mathfrak{p}\,\in\,\mathcal{P}(\mathbb{E})$. Notice that this chart is global since  the base manifold has a global chart and the bundle $\mathcal{P}(\mathbb{E})$ is trivial. 
The Hamiltonian of the theory, which is a global section $\mathrm{H}\,:\,\mathcal{P}(\mathbb{E})\,\to\,\mathbf{J}^\dagger\pi$ since the bundle is trivial, can be written as 
\begin{align}
\label{Eq: klein-gordon hamiltonian}
\mathrm{H}\quad:\quad &\mathcal{P}(\mathbb{E}) \to \mathbf{J}^\dagger \pi \nn \\  &\qquad \psi_{\mathcal{P}(\mathbb{E})}^{-1}\left(\,x, u,\, \rho\,\right) \,=\,\psi_{\mathbf{J}^\dagger \pi}^{-1}\left(\,x,\, u,\, \rho,\,  \rho_0 \,=\, \frac{1}{2}\left(\eta_{\mu \nu}\delta^{ab} \rho_a^\mu \rho_b^\nu + m^2 \delta_{ab}u^a u^b \right)\, \right)  \, ,
\end{align}
so that the canonical one-form on the covariant phase space reads
\be
\Theta_H \,=\, \rho_a^\mu \, \mathrm{d}u^a \wedge i_\mu \mathrm{d}^4x - \frac{1}{2} \left(\, \eta_{\mu \nu}\delta^{ab} \rho_a^\mu \rho_b^\nu + m^2 \delta_{ab}u^a u^b \,\right) \, \mathrm{d}^4 x \,.
\ee
The associated multisymplectic form is
\be
\mathrm{d}\Theta_H \,=\, \mathrm{d}\rho_a^\mu \wedge \mathrm{d}u^a \wedge i_\mu \mathrm{d}^4 x - \eta_{\mu \nu}\delta^{ab} \rho_a^\mu \mathrm{d}\rho_b^\nu \wedge \mathrm{d}^4x - m^2 \, \delta_{ab}u^a \mathrm{d}u^b \wedge \mathrm{d}^4 x \,
\ee
and the corresponding  action functional \eqref{Eq: action functional field theory} 
\be
\ac_\chi \,=\, \int_\m \chi^\star \Theta_H \,=\, \int_\m \left[\, P^\mu_a \de_\mu \phi^a -\frac{1}{2} \left(\, P^\mu_a P^a_\mu + m^2 \phi^a \phi_a \,\right) \,\right] \dd^4 x \,.
\ee
The integral above is well-defined if we consider pairs $\chi$  for which all the local representative functions $\phi^a$, $P^\mu_a$ and $\de_\mu \phi^a$ (for any $a$ and $\mu$)   are square-integrable with respect to the reference measure on $\m$, that is
$\ac$ is well defined on sections 
$$\left(\,\prod_{a} {\mathcal{H}^1(\mm, vol_{\mm})}^a \times \prod_{\mu,\, a} {\mathcal{L}^2(\mm, vol_{\mm})}^a_\mu\,\right) \cap \Gamma(\pi) \times \Gamma(\delta^1_0),$$
but for reasons that will be clear in Sect. \ref{Subsec: The symplectic case: massive vector boson field}, we  restrict ourselves to the subset $\tilde{\Gamma}(\pi) \times \tilde{\Gamma}(\pi^1_0)\subset \Gamma(\pi) \times \Gamma(\delta^1_0)$ given by
$$
 \tilde{\Gamma}(\pi) \times \tilde{\Gamma}(\pi^1_0)  \,=\, \Gamma(\pi) \times \Gamma(\delta^1_0) \cap \left(\, \prod_{a} {\mathcal{H}^{2}(\mm, vol_{\mm})}^a \times \prod_{\mu,\, a} {\mathcal{H}^{1}(\mm, vol_{\mm})}^a_\mu \,\right) ,
$$
given by elements in $\Gamma(\pi) \times \Gamma(\delta^1_0)$ for which the  norm (given in terms of Sobolev norms)
\be
\Vert \chi \Vert \,=\, \sum_a \Vert \phi^a \Vert_{\mathcal{H}^{2}} + \sum_{\mu,\, a}\Vert P^a_\mu \Vert_{\mathcal{H}^{1}} \,, 
\ee
is finite.
The functional $\ac$ is continuous in $\Vert \,\cdot \,\Vert_{\mathscr{C}}$. 
Indeed the following inequalities
\be
\begin{split}
\vert \ac_\chi - \ac_{\tilde{\chi}} \vert
\leq \int_\m &\left\vert\, P^\mu_a \de_\mu \phi^a - \tilde{P}^\mu_a \de_\mu \tilde{\phi}^a - P^\mu_a P_\mu^a + \tilde{P}^\mu_a \tilde{P}^a_\mu \,\right\vert vol_\m \, \\
\,\leq\, \int_\m &\biggl[\, \left\vert\, P^\mu_a (\de_\mu \phi^a - \de_\mu \tilde{\phi}^a)
 \,\right\vert + \left\vert\, (P^\mu_a - \tilde{P}^\mu_a) \de_\mu \tilde{\phi}^a
 \,\right\vert  \\
 &\qquad + \left\vert\, P^\mu_a (P^a_\mu - \tilde{P}^a_\mu)
 \,\right\vert + \left\vert\, \tilde{P}^\mu_a (P^a_\mu - \tilde{P}^a_\mu)
 \,\right\vert \,\biggr] vol_\m  \,\\
 \,\leq\, & \sum_{\mu, a}\Vert\,P^\mu_a \,\Vert_{\mathcal{L}^2} \, \Vert\, \de_\mu (\phi^a - \tilde{\phi}^a)  \,\Vert_{\mathcal{L}^2} + \sum_{\mu, a} \Vert\, 
P^\mu_a - \tilde{P}^\mu_a \,\Vert_{\mathcal{L}^2} \, \Vert\, \de_\mu \phi^a \,\Vert_{\mathcal{L}^2}  \\
 &\qquad+  \sum_{\mu, a} \Vert\, P^\mu_a \,\Vert_{\mathcal{L}^2} \, \Vert\, 
P^\mu_a - \tilde{P}^\mu_a \,\Vert_{\mathcal{L}^2} + \sum_{\mu, a} \Vert\, \tilde{P}^\mu_a \,\Vert_{\mathcal{L}^2} \, \Vert\, 
P^\mu_a - \tilde{P}^\mu_a \,\Vert_{\mathcal{L}^2} \,,
\end{split}
\ee 
hold, with $\chi\,=\, (\phi,\, P)$ and $\tilde{\chi} \,=\, (\tilde{\phi},\, \tilde{P})$.
One sees  that the right hand side approaches zero when $\chi$ approaches $\tilde{\chi}$ with respect to the norm $\|\cdot\|$, and this   proves the continuity of $\ac$ which can then be extended to the completion of $\tilde{\Gamma}(\pi) \times \tilde{\Gamma}(\pi^1_0)$ with respect to $\Vert \,\cdot \,\Vert$. Such a completion is the set we denote by $\fpe$ and which coincides (by definition, see  \cite{Adams1975}) with 
\be
\fpe \,=\, \prod_{a} {\mathcal{H}^{2}(\mm, vol_{\mm})}^a \times \prod_{\mu,\, a} {\mathcal{H}^{1}(\mm, vol_{\mm})}^a_\mu \,.
\ee

\noindent In order to compute the Euler-Lagrange form, we begin upon considering the contraction of $\dd \Theta_H$ along a $\delta_1$-vertical vector field $X$
\be
i_X \mathrm{d}\Theta_H \,=\, X_{\rho^\mu_a} \, \mathrm{d}u^a \wedge i_\mu \mathrm{d}^4 x - X_{u^a} \, \mathrm{d}\rho_a^\mu \wedge i_\mu \mathrm{d}^4 x - \eta_{\mu \nu}\delta^{ab} \rho_a^\mu \, X_{\rho^\nu_b} \, \mathrm{d}^4 x - m^2 \delta_{ab} u^a \, X_{u^b} \, \mathrm{d}^4x \,, 
\ee
whose  pull-back via $\chi$ is 
\be
\begin{split}
\chi^\star \left[ \, i_X \mathrm{d}\Theta_H \, \right] &\,=\, \left[ \, X_{\rho^\mu_a}\bigr|_{\chi} \, \partial_\mu \phi^a - X_{u^a}\bigr|_{\chi} \partial_\mu P_a^\mu - \eta_{\mu \nu} \delta^{ab} P_a^\mu \, X_{\rho^\nu_b}\bigr|_{\chi} - m^2\delta_{ab} \phi^a \, X_{u^b}\bigr|_{\chi} \, \right] \, \mathrm{d}^4 x \\
&\qquad =\, \left[ \, ({\mathbb{X}_\chi})_{\rho^\mu_a} \, \partial_\mu \phi^a - ({\mathbb{X}_\chi})_{u^a} \partial_\mu P_a^\mu - \eta_{\mu \nu} \delta^{ab} P_a^\mu \, ({\mathbb{X}_\chi})_{\rho^\nu_b} - m^2\delta_{ab} \phi^a \, ({\mathbb{X}_\chi})_{u^b} \, \right] \, \mathrm{d}^4 x \,,
\end{split}
\ee
so that  the Euler-Lagrange can be written as
\be
\mathbb{EL}_\chi(\mathbb{X}_\chi) \,=\, \int_{\mathbb{R}^4} \, \left[ \, ({\mathbb{X}_\chi})_{\rho^\mu_a} \, \partial_\mu \phi^a - ({\mathbb{X}_\chi})_{u^a} \partial_\mu P_a^\mu - \eta_{\mu \nu} \delta^{ab} P_a^\mu \, ({\mathbb{X}_\chi})_{\rho^\nu_b} - m^2\delta_{ab} \phi^a \, ({\mathbb{X}_\chi})_{u^b} \, \right] \, \mathrm{d}^4 x \,,
\ee
whose zeroes satisfies the following set of Euler-Lagrange equations
\be \label{Eq: euler-lagrange equations vector boson field}
 \begin{split} &\partial_\mu \phi^a \,=\, \eta_{\mu \nu} \delta^{ab} P_b^\nu \, \\
 &\partial_\mu P_a^\mu \,=\, -\delta_{ab} m^2 \phi^b \, .
\end{split} \ee
which implies that the configuration fields satisfy the Klein-Gordon equation.
\end{example}
\begin{example}[\textit{Free electrodynamics}] \label{Ex: electrodynamics variational principle}
Another example we consider is the pde describing the evolution of a sourceless electromagnetic field, i.e. the  Maxwell's equations.
Such equations come as Yang-Mills equations for a $U(1)$ gauge theory. 
For the sake of simplicity we  consider electrodynamics in Minkowski space time $(\mathbb{M}^{1,3}, \eta)$, as in the previous example.
In a gauge field theory fields are represented by connection one-forms on a principal bundle.
In the electromagnetic case the structure group is the Abelian group $G = U(1)$, whereas the principal bundle is $P=\mathbb{M}^{1,3}\times G $, which is a trivial bundle. 
Under these assumptions, electromagnetic fields are represented by Lie Algebra-valued differential forms on $\mathbb{M}^{1,3}$, with Lie Algebra given by $i \,\mathbb{R}$, so that they are sections of the fibre bundle 
$$
\pi \,\colon\, \mathbf{T}^*\mathbb{M}^{1,3}=\mathbb{E}\,\rightarrow\,\mathbb{M}^{1,3},
$$
whose first jet bundle is 
$$
\pi_1\,\colon\,\mathbf{J}^1\mathbb{E} = \mathbf{J}^1\mathbf{T}^*\mathbb{M}^{1,3} \,\rightarrow\,\mathbb{M}^{1,3}.
$$
The $\pi_1$-bundle is trivial over $\mathbb{M}^{1,3}$,  with typical fibre given by  $\mathbf{T}^\star_m \mathbb{M}^{1,3} \times \left(\, \bigotimes^2 \mathbf{T}^\star_m \mathbb{M}^{1,3} \,\right)$. 
It also results that the covariant phase space is a trivial bundle  
\be
\delta_1\,\colon\,\mathcal{P}(\mathbb{E})\,\rightarrow\,\mathbb{M}^{1,3},
\label{17set01}
\ee 
whose typical fibre is $\mathbf{T}^\star_m \mathbb{M}^{1,3} \times \left(\, \bigotimes^2 \mathbf{T}_m \mathbb{M}^{1,3} \,\right)$. 
Sections of the fibration $\delta^1_0\,\colon\,\mathcal{P}(\mathbb{E})\,\to\,\mathbb{E}$ are contravariant tensors of rank $2$. 
In particular, in order to correctly describe the properties of gauge fields within the multisymplectic framework, one has to consider the covariant phase space as the subbundle  (whose total space we still denote by $\mathcal{P}(\mathbb{E})$, with a slight abuse of notation) of the $\delta_1$-fibration on the basis $\mathbb{M}^{1,3}$ given by restricting its fibers to antisymmetric (contravariant) tensors. 
This choice is made to correctly reproduce Maxwell equations as we are about to describe. 
An adapted fibered chart (which turns out to be global) is written as $(\mathcal{P}(\mathbb{E}),\, \psi_{\mathcal{P}(\mathbb{E})})$, 
$$
\psi_{\mathcal{P}(\mathbb{E})}(\mathfrak{p}) \,=\, (x^0,..,x^3, u_{1},...,u_3, \rho^{00},..., \rho^{33})\,=\,(x, u, \rho),
$$
with $\mathfrak{p}$  a point on $\mathcal{P}(\mathbb{E})$. 
We denote sections of $\pi$ and $\delta^1_0$ by $\chi \,=\, (A_\mu,\, P^{\mu \nu})$, with $P^{\mu\nu}=-P^{\nu\mu}$.
The Hamiltonian of the theory is the section $\mathrm{H}\,:\,\mathcal{P}(\mathbb{E}) \to \mathbf{J}^\dagger \pi$ that we write as 
\be
 \psi_{\mathcal{P}(\mathbb{E})}^{-1} (x,\, u,\, \rho) \mapsto \psi_{\mathbf{J}^\dagger\pi}^{-1}\left(\, x,\, u,\, \rho, \, \rho_0 = \frac{1}{2}\rho^{\mu\nu}\rho^{\alpha \beta} \eta_{\mu\alpha}\eta_{\nu \beta} \right) \,,
\ee
reading the canonical $4$-form on the covariant phase space
\be
\Theta_H \,=\, \rho^{\mu \nu} \dd u_\nu \wedge i_\mu \dd^4 x - \frac{1}{2}\rho^{\mu\nu}\rho^{\alpha \beta} \eta_{\mu\alpha}\eta_{\nu \beta} \dd^4 x \,
\ee
and the multisymplectic form
\be
\dd \Theta_H \,=\, \dd \rho^{\mu \nu} \wedge \dd u_\nu \wedge i_\mu \dd^4 x - \rho^{\mu \nu} \eta_{\mu \alpha} \eta_{\nu \beta} \dd \rho^{\alpha \beta} \wedge \dd^4 x \,.
\ee

\noindent In analogy to the previous example, the action functional $\ac$ can be defined via \eqref{Eq: action functional field theory} on the space of suitably regular sections of $\pi$ and $\delta^1_0$ then  extended by continuity as
\be 
\fpe \,=\,  \prod_{\mu=0,...,3} {\mathcal{H}^{2}(\mathbb{M}^{1,3}, vol_{\mathbb{M}^{1,3}})}_\mu \times \prod_{\mu,\nu = 0,...,3} {\mathcal{H}^{1}(\mm, vol_{\mm})}^{\mu \nu} \,.
\ee
The contraction of the Euler-Lagrange form along  a tangent vector $\mathbb{X}_\chi$ turns out to be
\begin{equation}
\mathbb{EL}_{\chi}(\mathbb{X}_\chi) = \int_{\mathbb{M}^{1,3}} \left[ \, ({\mathbb{X}_\chi})_{\rho^{\mu\nu}} \, F_{\mu\nu} - ({\mathbb{X}_\chi})_{{u}_\mu} \partial_{\nu} P^{\mu \nu} + \frac{1}{2} P_{\mu \nu} ({\mathbb{X}_\chi})_{\rho^{\mu \nu}} \, \right] \, \mathrm{d}^4 x \,,
\end{equation} 
with $F_{\mu \nu} \,=\, \frac{1}{2}(\partial_\mu A_\nu - \partial_\nu A_\mu)$.
From the previous  expression one derives the following set of equations of motion
\begin{equation} \label{Eq: euler-lagrange equations electrodynamics} \begin{split} 
   & F_{\mu \nu} + \frac{1}{2}P_{\mu \nu} = 0\,, \\ & \partial_{\nu} P^{\mu \nu} = 0\,. \end{split} 
\end{equation}
from which one notices that the antisymmetry of $P^{\mu\nu}$ is compatible with the dynamics. 
\end{example}



\section{The canonical structure on the solution space}

\label{Sec: The canonical structure on the solution space}

\noindent We devote this section to comment further on the boundary term emerging from the variational principle since it gives rise to a canonical $2$-form on the solution space $\elag_{\m}$ which in turn allows for the definition of a canonical bracket on the solution space.

\noindent We begin by recalling that  $\m$ is a smooth orientable manifold with a boundary $\partial\m$, which we assume to have a finite number of connected components, that we denote by $\de \m^j$, for $j =1,...,k$. Each component $\partial\m^j$ inherits an orientation from the one on $\m$ and the embedding $i_{_{\partial\m}}\,\colon\,\partial\m\,\hookrightarrow\,\m$. 
We denote by $\mathfrak{o}^j$ ($j=1,\ldots,k$) the sign of each orientation. It is then possible to write
the boundary term emerging in the variational principle as
\be \label{Eq: structure boundary various components (mechanics)}
\begin{split}
\left({\Pi^\star_{_{\partial \m}} \alpha^{\partial \m}}\right)_\chi \,:\,\mathbb{X}_{\chi}\,\mapsto\,&   \int_{\partial \m} \chi^\star_{_{\partial \m}} \left[\, i_X \Theta_{\mathrm{H}} \,\right]  \\ &\qquad =\,  \sum_{j=1,...,k} \mathfrak{o}^j \int_{\de \m^j} \chi^\star_{_{\partial \m^j}} \left[\, i_X \Theta_{\mathrm{H}} \,\right] \\ & \qquad = 
\, \sum_{j=1,...,k} \left(\mathfrak{o}^j {\Pi^\star_{_{\partial \m^j}} \alpha^{\partial \m^j}}\right)_\chi \,\mathbb{X}_{\chi}\, \,,
\end{split} 
\ee
where $\alpha^{\de \m^j}$ denote the differential forms on the space of dynamical fields $\fpe^{\de \m^j}$ restricted to the component $\de \m^j$ of the boundary (see the diagram \eqref{5.8.3}).
Its differential reads
\be \label{Eq: structure at the boundary}
\begin{split}
-\dd \bigl( \,{\Pi^\star_{_{\partial \m}} \alpha^{\partial \m}} \, \bigr) \,=\, -\,{\Pi^\star_{_{\partial \m}} \bigl( \,\dd \alpha^{\partial \m}} \, \bigr) \,&=\,{\Pi^\star_{_{\partial \m}}\Omega^{\partial \m}} \\ & \, = -\dd \left(\sum_{j=1,...,k} \mathfrak{o}^j \Pi_{_{\de \m^j}}^\star \alpha^{\de \m^j}\right)\,=\, - \sum_{j=1,...,k} \mathfrak{o}^j \Pi_{_{\de \m^j}}^\star \dd \alpha^{\de \m^j} \\ & \qquad\qquad\qquad\qquad \qquad\qquad\qquad\,=\, \sum_{j=1,...,k} \mathfrak{o}^j \Pi_{_{\de \m^j}}^\star \Omega^{\de \m^j} \,,
\end{split}
\ee
where $\Omega^{\de \m^j}$ denote  differential forms on the space of sections restricted to $\de \m^j$.

\noindent We recall that, within the Cartan exterior calculus on any Banach manifold, one has 
\be
\dd \alpha (X,\, Y) \,=\, i_X \dd [\,\alpha (Y)\,] - i_Y \dd [\,\alpha (X)\,] - i_{[X,\,Y]} \alpha \,,
\ee
for a 1-form $\alpha$ and any pair of vector fields $X,Y$. Since $\fpe$ has a Banach (infinite dimensional) manifold structure, this expression makes sense on it. Given the vector fields $\mathbb{X}$ and $\mathbb{Y}$ on $\fpe$ we can write 
\be
\dd \bigl( \, {\Pi^\star_{_{\partial \m}} \alpha^{\partial \m}} \, \bigr) \left(\, \mathbb{X},\, \mathbb{Y} \,\right) \,=\, i_{\mathbb{X}} \dd \, \left[ {\Pi^\star_{_{\partial \m}} \alpha^{\partial \m}} \left(\, \mathbb{Y}\,\right) \right] - i_{\mathbb{Y}} \dd \, \left[ {\Pi^\star_{_{\partial \m}} \alpha^{\partial \m}} \left(\, \mathbb{X} \,\right) \right] - i_{\left[ \mathbb{X},\, \mathbb{Y} \right]}\left[ {\Pi^\star_{_{\partial \m}} \alpha^{\partial \m}}\right] \,.
\ee
Recalling relation \eqref{Eq: structure boundary various components (mechanics)} and \eqref{Eq:gateaux derivative},  it is  
\be
\begin{split}
i_{\mathbb{X}_\chi}\dd \int_{\de \m} \chi^\star_{\de \m} \left[ i_Y \Theta_{\mathrm{H}} \right] \,&=\, i_{\mathbb{X}_\chi}\dd \int_{\m} \chi^\star \left[ \dd i_Y \Theta_{\mathrm{H}} \right] \\ & \qquad =\, \int_{\m} \chi^\star \left[ \mathfrak{L}_X \dd i_Y \Theta_{\mathrm{H}} \right] \,=\, \int_{\m} \chi^\star \left[ \dd i_X \dd i_Y \Theta_{\mathrm{H}} \right] \,=\, \int_{\de \m} \chi^\star_{\de \m} \left[ i_X \dd i_Y \Theta_{\mathrm{H}} \right] \,.
\end{split}
\ee
Upon reinserting the dependence on  $\chi\,\in\,\fpe$, we can write 
\be
\begin{split}
\left({\Pi_{_{\partial \m}}^\star \Omega^{\partial \m}}\right)_{\chi} \left(\, \mathbb{X}_{\chi},\, \mathbb{Y}_\chi \,\right) &\,=\, -\dd\left( {\Pi^\star_{_{\partial \m}} \alpha^{\partial \m}}\right)_{\chi} \left(\, \mathbb{X}_{\chi},\, \mathbb{Y}_\chi \,\right)  \\
&\,=\, -i_{\mathbb{X}_{\chi}} \dd \,\left( \int_{\partial \m}  \chi_{_{\de \m}}^\star \left[\,i_Y \Theta_{\mathrm{H}}\,\right]\right)  + i_{\mathbb{Y}_\chi} \dd \,\left( \int_{\partial \m} \chi_{_{\de \m}}^\star \left[\,i_X \Theta_{\mathrm{H}}\,\right]\right)  + i_{\left[\,\mathbb{X}_\chi, \, \mathbb{Y}_\chi \, \right]} ({\Pi^\star_{_{\partial \m}} \alpha^{\partial \m}})
  \\ 
&\,=\, -\int_{\partial \m} \chi_{_{\de \m}}^\star \bigl[\, i_X \, \dd \, i_Y \Theta_{\mathrm{H}}  - i_Y \, \dd \, i_Y \Theta_{\mathrm{H}} - i_{[X,\,Y]} \, \Theta_{\mathrm{H}} \,\bigr] \,,
\end{split}
\ee
where $X$ and $Y$ are vector fields on $\mathcal{P}(\mathbb{E})$ which represent $\mathbb{X}_{\chi}$ and $\mathbb{Y}_\chi$ in an open neighborhood of the image of $\chi$.
A straightforward application of the identity $i_{[X,\, Y]} \,=\, \left[\, \mathfrak{L}_X , i_Y \,\right]$ leads to:
\be
\left({\Pi^\star_{_{\partial \m}} \Omega^{\partial \m}}\right)_\chi \left(\, \mathbb{X}_{\chi},\, \mathbb{Y}_\chi \,\right) \,=\,-\int_{\partial \m} \chi_{_{\de \m}}^\star \left[\, i_Y \, i_X  \,\dd \Theta_{\mathrm{H}} \,\right] \, =\, \int_{\partial \m} \chi_{_{\de \m}}^\star \left[\, i_X \, i_Y  \,\dd \Theta_{\mathrm{H}} \,\right] \,.
\ee
The above picture can be generalised. 
Let $\Sigma$ be a $(n-1)$ dimensional \textit{slice} in $\m$, that is a codimension-1 submanifold $\Sigma\,\hookrightarrow\,\m$ such that $\m\setminus\Sigma$ is the disjoint union of two components (regions) $\m_+$ and $\m_-$. Consider a domain $U\,\subset\,\m$ which encloses $\Sigma$ and is equipped with a local chart $(U,\, \psi_U)$, $\psi_U(m) \,=\, (x^0,\, x^j)_{j=1,...,d}$ such that the embedding 
$$
i_\Sigma\,:\,\Sigma\,\hookrightarrow\,\m
$$
is given upon fixing (see \eqref{29.7.1}) the $x^0$ coordinate (which one can then see as transversal to $\Sigma$, in analogy to what we already described after the definition \ref{def:2.3}) and such that the volume reads $$vol_\m \,=\, \dd x^0 \wedge vol_\Sigma.$$
It is now straightforward to define the fibrations on the basis  $\Sigma$ which generalise  those on $\de\m$ as represented in \eqref{diazamp}, as well as the maps and the differential forms as in \eqref{5.8.3}, that we write as  
\be \label{sei.1}
\begin{tikzcd}
	\mathcal{F}_{\mathcal{P}(\mathbb{E})} \arrow[r, "\Pi_{\Sigma}"] & 
	\mathcal{F}_{\mathcal{P}(\mathbb{E})}^{\Sigma} \,\simeq\,\left(\mathcal{F}_{\mathcal{P}(\mathbb{E}),\,0}^{\Sigma} \times \mathscr{B}^{\Sigma} \right)\arrow[r, "\rho_{\Sigma}"] & \mathcal{F}_{\mathcal{P}(\mathbb{E}),\,0}^{\Sigma}\arrow[r, "\lambda"] & \mathbf{T}^\star \mathcal{F}_{\mathbb{E}}^{\Sigma} \\ \left\{\begin{array}{l}\Pi^\star_{\Sigma}\alpha^{\Sigma}, \\ \Pi^\star_{_\Sigma}\Omega^{\Sigma}\end{array}\right\} & \left\{\begin{array}{l} \alpha^{\Sigma}\,=\,\rho^\star_{\Sigma}\vartheta, \\ \Omega^{\Sigma}\,=\,-\dd\alpha^{\Sigma}\end{array}\right\} \arrow[l, "\Pi^\star_{\Sigma}"] &  \left\{\begin{array}{l} \vartheta^{\Sigma}\,=\,\lambda^\star_\Sigma\bar\vartheta, \\ \omega^{\Sigma}\,=\,-\dd\vartheta^{\Sigma}\end{array}\right\} \arrow[l, "\rho^\star_{\Sigma}"] & 	 \left\{\begin{array}{l} \bar\vartheta^{\Sigma}, \\ \bar\omega^{\Sigma}\,=\,-\dd\bar\vartheta^{\Sigma}\end{array}\right\}. \arrow[l, "\lambda^\star_{\Sigma}"] 
\end{tikzcd}
\ee
Such a chain starts on the cotangent bundle $\mathbf{T}^\star \mathcal{F}_{\mathbb{E}}^{\Sigma}$ where  one has, in analogy to \eqref{omega1} and \eqref{theta1}, the canonical 1-form
\be
\label{theta1s}
\bar{\vartheta}_{(\varphi, \varrho)}^{\Sigma}(\mathbb{X}_{(\varphi,\,\varrho)}) \,=\, \int_{\Sigma} vol_{\Sigma} \left(\, \varrho_a \, \mathbb{X}_{\varphi^a} \,\right)  \,,
\ee
and the canonical 2-form 
\be
\label{21ott-ale4}
\bar{\omega}_{(\varphi, \varrho)}^{\Sigma}\,=\,-\dd \bar{\vartheta}_{(\varphi, \varrho)}^{\Sigma},
\ee 
with
\be
\label{omega1s}
\bar{\omega}_{(\varphi,\, \varrho)}^{\Sigma}(\mathbb{X}_{(\varphi,\, \varrho)},\, \mathbb{Y}_{(\varphi,\, \varrho)}) \,=\, \int_{\Sigma} vol_{\Sigma} \left(\,\mathbb{X}_{\varphi^a} \,\mathbb{Y}_{\varrho_a} - \mathbb{Y}_{\varphi^a} \, \mathbb{X}_{\varrho_a} \,\right) \,,
\ee
where $\mathbb{X}_{(\varphi,\, \varrho)}$ and $\mathbb{Y}_{(\varphi,\, \varrho)}$ are elements in $\mathbf{T}_{(\varphi,\,\varrho)}(\mathbf{T}^\star \mathcal{F}_{\mathbb{E}}^{\Sigma})$, while  $\mathbb{X}_{\varphi^a}$, $\mathbb{X}_{\varrho_a}$ (resp. $\mathbb{Y}_{\varphi^a}$, $\mathbb{Y}_{\varrho_a}$) are their components with respect to the adopted chart. One writes, in analogy to  \eqref{15ott1}, that 
\be
\label{21ott-ale1}
\alpha^{\Sigma}\,=(\lambda\circ\rho)^\star\bar{\vartheta}^{\Sigma}\,=\,\rho_\Sigma^\star \,( \lambda^\star \bar{\vartheta}^{\Sigma}) \,=\, \rho_\Sigma^\star \, \vartheta^\Sigma 
\ee
with $-\dd \vartheta^\Sigma \,=\, \omega^\Sigma \,$
and
\be
\label{21ott-ale2}
-\dd (\Pi_{\Sigma}^\star \alpha^{\Sigma}) \,=\, \Pi_{\Sigma}^\star \Omega^{\Sigma} \,=\, \Pi_{\Sigma}^\star \, (\rho_\Sigma^\star \, (\lambda^\star \bar{\omega}^{\Sigma})) \,=\, \Pi_{\Sigma}^\star \,( \rho_\Sigma^\star \, \omega^\Sigma) \,.
\ee


\noindent We can then write 
\be
\label{21ott-ale3}
{\Pi^\star_{\Sigma} \Omega^{\Sigma}}_\chi \left(\, \mathbb{X}_{\chi},\, \mathbb{Y}_\chi \,\right) \,=\, \int_{\Sigma} \chi_\Sigma^\star \left[\, i_X \, i_Y  \,\dd \Theta_{\mathrm{H}} \,\right] \,.
\ee
The aim of this section is to prove that the 2-form ${\Pi_\Sigma^\star \Omega^\Sigma}$ does not depend on the specific slice $\Sigma$ if evaluated at $\chi \in \elag_{\m}$, i.e. that ${\Pi_\Sigma^\star \Omega^\Sigma}$ is a canonical structure on the solution space of the considered variational problem. Our analysis begins by  proving the following claims.
\begin{proposition} \label{Prop: isotropic manifold}
$\elag_{\m}$ is an isotropic manifold for the $2$-form $\Pi^\star_{\partial \m}\Omega^{\partial \m}$, i.e. it vanishes on $\elag_{\m}$.
\begin{proof}
The action of the differential operator upon both sides of \eqref{Eq: first fundamental formula} gives, at each pont $\chi\,\in\,\fpe$
\be
\underbrace{\dd \, \dd}_{=0} \mathscr{S} \,=\, \dd (\mathbb{EL}) + \dd ({\Pi_{_{\partial \m}}^\star \alpha}^{\partial \m}) \,=\, \dd( \mathbb{EL}) - ({\Pi^\star_{_{\partial \m}} \Omega}^{\partial \m}) \,,
\ee
that is 
\be
\dd (\mathbb{EL})- {\Pi^\star_{_{\partial \m}}\Omega}^{\partial \m} \,=\, 0 \,.
\ee
Consider the pull-back of such a differential form to $\elag_{\m}$ via $\mathfrak{i}_{\elag_{\m}}$.
Since $\elag_{\m}$ is the space of zeroes of $\mathbb{EL}$ and the pull-back acts naturally with respect to the exterior differential $\dd$, it is  $$\mathfrak{i}_{\elag_{\m}}^\star(\dd \mathbb{EL}) \,=\, \dd( \mathfrak{i}_{\elag_{\m}}^\star \mathbb{EL}) \,=\, 0$$ and therefore
\be
\mathfrak{i}_{\elag_{\m}}^\star \,\left( {\Pi^\star_{_{\partial \m}}\Omega^{\partial \m}}\right)_\chi \,=\, 0 \,
\ee
\end{proof}
\end{proposition}
if $\chi \in \elag_{\m}$.
\noindent In terms of the  \eqref{Eq: structure at the boundary}, the previous proposition gives
\be \label{Eq: sum structure at the boundary}
\sum_{j=1,...,k} \left(\mathfrak{o}^j {{\Pi^\star_{_{\de \m^j}}\Omega}^{\de \m^j}}\right)_\chi \,=\, 0
\ee 
for $\chi \in \elag_{\m}$.

\begin{proposition} \label{luca-prvp}
Consider a submanifold $\tilde{\m}$ open into $\m$ and having the same dimension of $\m$, and 
denote  by $\elag_{\tilde{\m}}$ the solution space related to the variational principle formulated on $\tilde{\m}$. 
There exists a one-to-one correspondence between restrictions to $\tilde{\m}$ of elements in $\elag_{\m}$ and elements in $\elag_{\tilde{\m}}$.
\begin{proof}
Consider the embedding $\mathfrak{i}_{_{\tilde\m}}\,:\,\tilde{\m}\,\hookrightarrow\,\m$.
Denote by $\pi_{_{\tilde\m}}\,\colon\,\tilde{\mathbb{E}}\,\to\,\tilde\m$ the pull-back bundle corresponding to $\pi\,\colon\,\mathbb{E}\,\to\,\m$ via the embedding 
$\mathfrak{i}_{_{\tilde\m}}$. 
The covariant phase space $\mathcal{P}(\tilde{\mathbb{E}})$ constructed over  $\tilde{\mathbb{E}}$ is embedded into $\mathcal{P}(\mathbb{E})$ via $\tilde{\mathfrak{i}}_{\mathcal{P}} \,=\, \mathfrak{i}_{_{\tilde\m}} \times \mathbbm{1}_{\mathcal{E}_{\delta_1}}$, with $\mathcal{E}_{\delta_1}$ denoting  the fibre of the fibration $\delta_1\,:\,\mathcal{P}(\mathbb{E})\,\to\,\m$. We have then the double fibration 
\be
\label{1.8.1}
\begin{tikzcd}
\mathcal{P}(\tilde{\mathbb{E}}) \arrow[r, "(\delta_{\tilde\m})^1_0"] & \tilde{\mathbb{E}} \arrow[r, "\pi_{_{\tilde\m}}"] & \tilde\m.
\end{tikzcd}
\ee
which gives the analogue to the right column from the diagram  \eqref{diazamp}. Notice moreover that, for the embedding $\mathfrak{i}_{\tilde\m}$, the left and the right columns out of that diagram represent fibrations which are diffeomorphic, i.e. one has $\mathcal{P}(\tilde{\mathbb{E}})\simeq\mathfrak{i}^\star_{\tilde\m}(\mathcal{P}(\mathbb{E}))$.
In analogy to $\chi\,:\,\m\,\to\,\mathcal{P}(\mathbb{E})$, we denote by $\tilde{\chi}\,:\,\tilde\m\,\to\,\mathcal{P}(\tilde{\mathbb{E}})$ a pair of sections of $\pi_{\tilde{\m}}$ and $(\delta_{\tilde{\m}})^1_0$. 
For any $\chi$ there exists a $\tilde{\chi}$ (its restriction $\chi \bigr|_{\tilde{\m}}$) such that the diagram
\be
\begin{tikzcd}
\mathcal{P}(\tilde{\mathbb{E}}) \arrow[r, "\tilde{\mathfrak{i}}_{\mathcal{P}}", hookrightarrow] & \mathcal{P}(\mathbb{E}) \\
\tilde{\m} \arrow[u, "\tilde{\chi}", dashed] \arrow[r, "\tilde{\mathfrak{i}}", hookrightarrow] & \m. \arrow[u, "\chi", dashed] 
\end{tikzcd} 
\ee
commutes.
The variational principle formulated on $\mathcal{P}(\tilde{\mathbb{E}})$ gives the following equations of motion:
\be
\tilde{\chi}^\star \left(\, i_{\tilde{X}} \, \tilde{\mathfrak{i}}_{\mathcal{P}}^\star \dd \Theta_{\mathrm{H}} \,\right) \,=\, 0
\ee
for any $\tilde{X} \in\mathfrak{X}(U^{(\tilde\chi)})$ which is $\tilde\delta_1=(\delta_{\tilde\m})^1_0\circ\pi_{\tilde\m}$-vertical and $(\delta_{\tilde\m})^1_0$-projectable.
The left hand side of the previous equation can be rewritten as:
\be
\tilde{\chi}^\star \left(\, i_{\tilde{X}} \, \tilde{\mathfrak{i}}_{\mathcal{P}}^\star \dd \Theta_{\mathrm{H}} \,\right) \,=\, \tilde{\chi}^\star \left(\, \tilde{\mathfrak{i}}_{\mathcal{P}}^\star \, i_{X}  \dd \Theta_{\mathrm{H}} \,\right) \,=\, \left( \tilde{\mathfrak{i}}_{\mathcal{P}} \circ \tilde{\chi} \right)^\star \left(\, i_{X}  \dd \Theta_{\mathrm{H}} \,\right) \,=\, \left( \chi \circ \tilde{\mathfrak{i}} \right)^\star \left(\, i_X \dd \Theta_{\mathrm{H}} \,\right) \,=\, \tilde{\mathfrak{i}}^\star \left[\, \chi^\star  \left(\, i_X \dd \Theta_{\mathrm{H}} \,\right) \, \right] \,,
\ee
for any $X$ which is $\tilde{\mathfrak{i}}_{\mathcal{P}}$-related to $\tilde{X}$ and where $\chi$ is one of the sections of $\delta_1$ that restrict to $\tilde{\chi}$.
The previous chain of equalities clearly shows that if $\chi$ is a solution for the variational problem on $\m$, that is  if $\chi^\star  \left(\, i_X \Omega_{\mathrm{H}} \,\right) \,=\, 0$, then its restriction $\tilde{\chi}$ is a solution of the variational principle on $\tilde{\m}$.
Moreover, if we admit that the equations of motion do not give rise to biforcations, any $\tilde{\chi}$ can not come from two different $\chi$, say $\chi_1$ and $\chi_2$, since in that case $\chi_1$ and $\chi_2$ should coincide on the whole open region $\tilde{\mathscr{M}}$.
Consequently the correspondence between $\tilde{\chi}$ and $\chi$ is one-to-one.
\end{proof}
\end{proposition}
\noindent What we analysed above in this section allows to prove the following.
\begin{proposition}
The 2-form $\Pi_\Sigma^\star \Omega^\Sigma$ does not depend on the slice $\Sigma$ if evaluated on solutions of the equations of the motion.
\begin{proof}
Denote by $\mathscr{M}_{12}$ a submanifold embedded in $\m$  with an open interior fulfilling the conditions  of the previous Prop. \ref{luca-prvp}, such that its  boundary is given by two slices, that we denote by $\Sigma_1$ and $\Sigma_2$, whose orientations point towards the outside. 
A straightforward application of \eqref{Eq: sum structure at the boundary} and propositions \ref{Prop: isotropic manifold} with \ref{luca-prvp} gives:
\be \label{Eq: equivalence Sigma}
(\Pi_{\Sigma_1}^\star \Omega^{\Sigma^1})_\chi \,=\, (\Pi_{\Sigma_2}^\star \Omega^{\Sigma^2})_\chi \,,
\ee
if $\chi \in \elag_{\mathscr{M}_{12}}$.
Notice that any slice $\Sigma$ gives a set $S_\Sigma$ whose elements are all those slices such that the union with $\Sigma$ is the boundary of a region of $\m$ of the type of Prop. \ref{luca-prvp}.
The relation \eqref{Eq: equivalence Sigma} shows  that $\pssos$ does not vary on slices in $S_\Sigma$. 
The thesis follows upon noticing that for any couple of $\Sigma$'s belonging to different families, there always exists a third one belonging to both families.
\end{proof}
\end{proposition}

\section{Symplectic vs gauge theories}
\label{Sec: The symplectic case vs gauge theories}

\noindent The aim of this section is to show that the closed 2-form $\pssos$ on $\fpe$, which is  canonical in the sense that it does not depend on $\Sigma$ when restricted to the space of solutions $\elag_\m\subset\fpe$ of the variational problem, may be non degenerate (i.e. symplectic) or degenerate (i.e. pre-symplectic), depending on the specific field theory under analysis. 
We refer to the theories falling in the first case as \textit{symplectic theories}, and to theories falling into the second case as \textit{gauge theories}. 

\noindent We proceed as follows. We first prove that  the 2-form $\os$ on $\fpe^\Sigma$ is pre-symplectic.  Upon restricting the domain of the dynamics to a suitable collar $C^\Sigma_\epsilon\,\subseteq\,\m$ around $\Sigma$, it results that the corresponding space of   solutions $\elag_{C^\Sigma_\epsilon}$  is isomorphic 
 to the space of solutions of a pre-symplectic Hamiltonian system associated to  $\os$ on $\fpe^\Sigma$, whose integrability can be studied via the so called \textit{pre-symplectic constraint algorithm}  (see Appendix \ref{App: PCA}).
We conclude then this analysis by showing  that, on the restrictions to the collar $C^\Sigma_\epsilon$ of the fields $\fpe$,  the 2-form $\pssos$ corresponds to a 2-form  $\os_\infty$  which comes as the pull-back under a suitable smooth map of the 2-form $\os$ to the final stable manifold resulting from the pre-symplectic constraint algorithm.
It turns out that the latter may be symplectic or not and, consequently, $\pssos$ may be symplectic or not.
We start by considering the following
\begin{proposition}
The $2$-form $\omega^\Sigma$ is closed and non-degenerate (i.e. symplectic) on $\mathcal{F}^\Sigma_{\mathcal{P}(\mathbb{E}), \,0}$, while $\Omega^\Sigma$ is closed and degenerate (i.e. pre-symplectic) on $\mathcal{F}^\Sigma_{\mathcal{P}(\mathbb{E})}$.
\end{proposition}
\begin{proof}
Consider the diagram \eqref{sei.1}. It is sufficient to recall the definition \eqref{21ott-ale4} of the symplectic  2-form $\bar{\omega}^\Sigma$ on $\mathbf{T^\star}\mathcal{F}_{\mathbb{E}}^\Sigma$, together with the relations \eqref{21ott-ale1} and \eqref{21ott-ale2}. 
\end{proof}
\noindent Closely to what we already introduced in the previous section, consider a hypersurface (a slice) $\mathfrak{i}_{\Sigma}\,\colon\,\Sigma\,\hookrightarrow\,\m$ and a coordinate chart   $(U,\, \psi_U)$, $\psi_U(m) \,=\, (x^0,\, x^j)_{j=1,...,d}$ around $\Sigma$ such that one has the local representation
\be
\label{5.8.1}
\mathfrak{i}_{\Sigma}\,\colon\,(x^1, \ldots, x^d)\,\mapsto\,(x^0=x^0_\Sigma, x^1, \ldots, x^d)
\ee
for the embedding, and the local representation 
$$
C^\Sigma_\epsilon \,=\, [x^0_\Sigma,\, x^0_\Sigma + \epsilon) \times \Sigma \,=\, \mathbb{I}^\Sigma_\epsilon \times \Sigma
$$
for a collar $C^\Sigma_\epsilon$ around $\Sigma$, with $\mathbb{I}_\epsilon^{\Sigma}=[x^0_\Sigma, x^0_\Sigma+\epsilon)$ for $\epsilon >0$.  It is clear that a local chart on $C_\epsilon^\Sigma$ can be given as the product of a chart $(U_{\mathbb{I}^\Sigma_\epsilon},\, \psi_{U_{\mathbb{I}^\Sigma_\epsilon}})$ with $\psi_{U_{\mathbb{I}^\Sigma_\epsilon}}(r) \,=\, s$  on $ \mathbb{I}^\Sigma_\epsilon $, times a chart $(U_\Sigma,\, \psi_{U_\Sigma})$, $\psi_\Sigma(\mathfrak{x}) \,=\, (x^1,\ldots, x^d)$ on $\Sigma$, with $\mathfrak{x} \in U_\Sigma\subseteq \Sigma$. The volume form reads 
$$
vol_{C^\Sigma_\epsilon} \,=\, \dd x^0 \wedge vol_\Sigma.
$$
As in the proof of the proposition \ref{luca-prvp}, one has the embedding 
$$
\mathfrak{i}_{C^\Sigma_\epsilon}\,:\,C^\Sigma_\epsilon\hookrightarrow\m,
$$
the  pullback bundle $\pi^{\Sigma}_\epsilon\,\colon\,\mathbb{E}^\Sigma_\epsilon\,\to\,C^\Sigma_\epsilon$ 
 with fibers diffeomorphic to those  in  $\pi\,\colon\,\mathbb{E}\,\to\,\m$, and the corresponding phase space as in the diagram \eqref{1.8.1}. A  set of suitably regular pair of sections of such a phase space on the collar $C^\Sigma_\epsilon$ is what we denote by    
$\fpe^{C^\Sigma_\epsilon}$ (whose elements we then  consider as the restrictions to the collar $C^\Sigma_\epsilon$ of the dynamical fields on $\m$), the restriction of the action functional $\ac$ to   $\fpe^{C^\Sigma_\epsilon}$ is what we denote by $\ac^\epsilon$, namely (see \eqref{Eq: action functional field theory})
$$
\ac^\epsilon\,=\,\int_{C^\Sigma_\epsilon}\chi^\star\Theta_H
$$
where $\chi\in\fpe^{C^\Sigma_\epsilon}$ and $\Theta_H$ is the restriction of the Poincar\'e-Cartan form to the phase space bundle on $C^\Sigma_\epsilon$. 

The space $\fpe^{C^\Sigma_\epsilon}$ turns to be isomorphic to the space of curves  $\Gamma \left(\fpe^{\Sigma}\right)$ on $\fpe^\Sigma$. In order to describe such isomorphism, we begin upon  recalling that elements $\chi\in\fpe^{C^\Sigma_\epsilon}$ can be locally written as  the product of maps
$$
\chi\,=\,(\phi^a, P^0_a, P^j_a)
$$
on the collar domain $C^\Sigma_\epsilon$, then we  
consider a local chart $(U_{\fpe^\Sigma}, \psi_{U_{\fpe^\Sigma}})$ on $\fpe^\Sigma$ such that  the maps  
\be
\label{6.8.2}
\psi_{U_{\fpe^\Sigma}}(\chi_{_{\Sigma}}) \,=\, (\varphi^1,...,\varphi^r,\, p_1,...,p_r,\, \beta^1_1,...,\beta^d_r) \,=\, (\varphi^a, p_a, \beta_a^j). 
\ee
on $\Sigma$ provide a local representation if  $\chi_{_{\Sigma}}=\chi\big|_\Sigma$ denotes a point in $U_{\fpe^\Sigma} \subset \fpe^\Sigma$.  In terms of such  local representations, the isomorphism  
\be
\label{sei.2}
\varpi\,:\,\fpe^{C^\Sigma_\epsilon}\,\to\,\Gamma(\fpe^\Sigma)
\ee
reads, with  
$$
\gamma_{(s)} \,=\, \left(\, \varphi_{(s)},\, p_{(s)},\, \beta_{(s)} \,\right)\,\in\,\Gamma(\fpe^\Sigma,
$$
the expression 
\be
\label{20ott-ale3}
\begin{split}
\varphi^a_{(s)} \left(\psi^{-1}_{U_\Sigma} (x^k)\right) &= \phi^a \left(\psi_{U_\epsilon}^{-1}(s),\, \psi^{-1}_{U_\Sigma}(x^k) \right) \, , \\
{p_a}_{(s)} \left(\psi_{U_\Sigma}^{-1}(x^k) \right) &= P^0_a \left( \psi_{U_\epsilon}^{-1}(s),\, \psi_{U_\Sigma}^{-1}(x^k) \right) \, , \\ 
{\beta_a^j}_{(s)} \left( \psi_{U_\Sigma}^{-1}(x^k) \right) &= P_a^j \left( \psi_{U_\epsilon}^{-1}(s),\,  \psi_{U_\Sigma}^{-1}(x^k) \right) \, ,
\end{split}
\ee
with $s \in \mathbb{I}_\epsilon^\Sigma$ and $\mathfrak{x}\in U_\Sigma\subseteq\Sigma$. 
The map $\varpi$ amounts to identify the coordinate transversal to the collar with the evolution parameter describing the curve on $\fpe^\Sigma$, so it turns to be a diffeomorphism.
Let $\gamma_{(s)}$ be an element  in $\Gamma \left(\fpe^{\Sigma}\right)$:   
the pull-back induced by $\varpi$ upon the action functional $\ac^\epsilon$ to $\Gamma \left(\fpe^{\Sigma}\right)$ reads
\begin{equation} \label{Eq: action functional boundary}
\begin{split}
	\tilde{\mathscr{S}}^{\epsilon}_\gamma = \left({\varpi^{-1}}^\star \ac^\epsilon \right)_\gamma &= \int_{x^0_\Sigma}^{x^0_\Sigma + \epsilon} \mathrm{d}s\int_{\Sigma} vol_\Sigma \left( {p_a}_{(s)} \dot{\varphi}^a_{(s)} + {\beta^k_a}_{(s)} \partial_k \varphi_{(s)}^a - H(\gamma_s) \right)  \\ &=\int_{x^0_\Sigma}^{x^0_\Sigma + \epsilon}\dd s \,\,\mathscr{L}(\gamma_{(s)},\, \dot{\gamma}_{(s)}) \,,
\end{split}
\end{equation}
where the dot denotes the derivative with respect to $s$. The second line introduces  a Lagrangian function $\mathscr{L}\,\colon\,\mathbf{T}(\fpe^\Sigma)\,\rightarrow \,\mathbb{R}$ which, evaluated along the curve $t\gamma_{(s)} \,=\, (\gamma_{(s)},\, \dot{\gamma}_{(s)})$ given by the tangent lift of $\gamma_{(s)}$, reads
\begin{equation}
\label{sei.3}
\begin{split}
\mathscr{L}\left( \varphi, v_\varphi, p, v_p, \beta, v_\beta \right)\mid_{t \gamma_{(s)}} &= \int_{\Sigma}  vol_\Sigma \left( {p_a}_{(s)} \dot{\varphi}^a_{(s)} + {\beta_a^k}_{(s)}\partial_k \varphi^a -H(\gamma_s) \right)  \\
& = \langle p , \dot{\varphi} \rangle + \langle \beta , \mathrm{d}_{\Sigma}\varphi  \rangle - \int_{\Sigma}  vol_{\Sigma}\,H(\varphi_{(s)}, p_{(s)}, \beta_{(s)} ) \\ & = \langle p , \dot{\varphi} \rangle - \mathcal{H}(\gamma_{(s)})\,, 
\end{split}
\end{equation}
where  $\bigl(U_{\mathbf{T}(\mathcal{F}_{\mathcal{P}(\mathbb{E})}^\Sigma )},\, \psi_{U_{\mathbf{T}(\mathcal{F}_{\mathcal{P}(\mathbb{E})}^\Sigma)}}\bigr)$ is a local chart on $\mathbf{T}(\mathcal{F}_{\mathcal{P}(\mathbb{E})}^\Sigma)$ with $(\chi_\Sigma,\, \dot{\chi}_\Sigma) \in U_{\mathbf{T}(\mathcal{F}_{\mathcal{P}(\mathbb{E})}^\Sigma)} \subset \mathbf{T}(\mathcal{F}_{\mathcal{P}(\mathbb{E})}^\Sigma)$ and 
$$
\psi_{U_{\mathbf{T}(\mathcal{F}_{\mathcal{P}(\mathbb{E})}^\Sigma)}}((\chi_\Sigma,\, \dot{\chi}_\Sigma)) \,=\, (\varphi, {v_\varphi}, p, {v_p}, \beta, v_\beta). 
$$ 
In the previous formula \eqref{sei.3} the operator  $\dd_\Sigma$ denotes the differential over the smooth $(n-1)$-dimensional slice $\Sigma$, the symbol $\langle \cdot, \cdot \rangle$ denotes both the integration over $\Sigma$ and the contraction over the internal degrees of freedom of the fields, whereas the functional $\mathcal{H}$ is
\be
\label{5.8.4}
\mathcal{H}(\gamma_{(s)}) = \langle \beta , \mathrm{d}_{\Sigma}\varphi  \rangle - \int_{\Sigma}vol_\Sigma \,H(\gamma_{(s)}),. 
\ee

\begin{remark}
Notice  that the Lagrangian $\lag$ introduced above allows to study the solutions of Hamilton-De Donder equations (which is a system of first order PDEs) close to a slice $\Sigma$ in terms of the usual Lagrangian formalism over the tangent bundle on a configuration manifold which, as in this case, may be infinite dimensional. 
We indeed invite the reader not to confuse it with the Lagrangian form (i.e. a suitable $n$-form on the first order jet bundle of $\pi\,\colon\,\mathbb{E}\,\to\,\m$) in terms of which a multi-symplectic Lagrangian formulation of first order field theories can be given (see, for instance \cite{Krupka2015-Variational_Geometry} and references therein).
\end{remark}

It is a direct consequence of proposition \ref{luca-prvp} that there is a bijection between the set $\elag_{C_{\epsilon}^\Sigma}\subset\fpe^{C_\epsilon^\Sigma}$  of the extrema of the action functional $\ac^\epsilon$ and the set  containing the  restrictions for  $x^0\,\in\, [x^0_\Sigma,\, x^0_\Sigma + \epsilon)$ of the extrema of the action functional $\ac$. Such a bijection comes naturally in terms of the pullback map $\mathfrak{i}^*_{C^\Sigma_\epsilon}$ between sections of the phase spaces on $C^\Sigma_\epsilon$ and on $\m$. 
The above discussion shows moreover  that the map $\varpi$ provides a  bijection between the set  $\elag_{C^\Sigma_\epsilon}$ of  extrema of $\ac^\epsilon$  and the set $\widetilde{\elag}^\epsilon\subset\Gamma(\fpe^\Sigma)$  of extrema of ${\varpi^{-1}}^\star \ac^\epsilon \,=\, \tilde{\ac}^\epsilon$.
It is possible to prove (\cite{Ibort-Spivak2017-Covariant_Hamiltonian_YangMills} theorem $3.1$) the following result, whose proof consists of a direct computation that we report for the sake of completeness.
\begin{proposition} \label{Prop: pre-symplectic system boundary}
Extrema $\widetilde{\elag}^\epsilon$ of the functional $\tilde{\mathscr{S}}^\epsilon$ are the solutions of the pre-symplectic Hamiltonian system $\left(\, \mathcal{F}_{\mathcal{P}(\mathbb{E})}^\Sigma,\, \Omega^\Sigma,\, \mathcal{H} \,\right)$.
\begin{proof}
Given the local representation of the embedding \eqref{5.8.1}, a vector field $\mathbb{X}\in\mathfrak{X}(\fpe^\Sigma)$ can be represented (along the coordinate chart given in \eqref{6.8.2}) in terms of a vector field (compare this with the vector field represented in \eqref{28.7.2})
\be
\label{5.8.5}
\mathbb{X}\,=\,\mathbb{X}_{\varphi^a}(\varphi)\frac{\delta}{\delta \varphi^a}\,+\,\mathbb{X}_{p_a}(\varphi, p, \beta)\frac{\delta}{\delta p_a}\,+\,\mathbb{X}_{\beta_a^j}(\varphi, p, \beta)\frac{\delta}{\delta \beta_a^j},
\ee
on $\mathcal{P}(\mathbb{E})$, where $x=(x^1, \ldots, x^d)$ denotes the coordinates of a point on $\Sigma$.
Solutions of the pre-symplectic system above are the integral curves of the vector field $\mathbb{\Gamma} \in \mathfrak{X}(\fpe^\Sigma)$ satisfying
\be
\Omega^\Sigma (\mathbb{\Gamma},\, \mathbb{X} ) \,=\, \dd \mathcal{H}( \mathbb{X} )  
\ee
for any $\mathbb{X} \in \mathfrak{X}(\fpe^\Sigma)$.
From the  definition of $\Omega^\Sigma$ (see \eqref{omega1s}-\eqref{21ott-ale2}) and $\mathcal{H}$ in \eqref{5.8.4}, the previous equation gives
\be \label{Eq: pre-symplectic system}
\int_\Sigma vol_\Sigma\,\left(\, \mathbb{\Gamma}_{\varphi^a} \, \mathbb{X}_{p_a} - \mathbb{\Gamma}_{p_a} \, \mathbb{X}_{\varphi^a} \,\right)  \,=\, \int_\Sigma vol_\Sigma  \left(\, \frac{\delta \mathcal{H}}{\delta {\varphi^a}} \, \mathbb{X}_{\varphi^a} + \frac{\delta \mathcal{H}}{\delta p_a}\, \mathbb{X}_{p_a} + \frac{\delta \mathcal{H}}{\delta \beta^j_a} \, \mathbb{X}_{\beta^{j_a}} \,\right)  \\ 
\ee
with respect to the components, as in  \eqref{5.8.5}, of  the vector fields $\mathbb{\Gamma}$ and $\mathbb{X}$.
It results in 
\be \label{Eq: pre-symplectic system components}
\begin{split}
\mathbb{\Gamma}_{\varphi^a} &\,=\, \frac{d \varphi^a_{(s)}}{ds} \,=\, \frac{\delta \mathcal{H}}{\delta p_a} \,=\, - \frac{\partial H}{\partial \rho^0_a}(\varphi_{(s)},\, p_{(s)},\, \beta_{(s)}) \,,\\
{\mathbb{\Gamma}_{p_a}} &\,=\, \frac{d {p_a}_{(s)}}{d s}  \,=\,  -\frac{\delta \mathcal{H}}{\delta \varphi^a} \,=\, \frac{\partial {\beta_a^k}_{(s)}}{\partial x^k} + \frac{\partial H}{\partial u^a}(\varphi_{(s)},\, p_{(s)},\, \beta_{(s)}) \,, \\
0 &\,=\, \frac{\delta \mathcal{H}}{\delta \beta^j_a} \,=\, \frac{\partial \varphi^a}{\partial x^j} - \frac{\partial H}{\partial \rho^j_a}(\varphi_{(s)},\, p_{(s)},\, \beta_{(s)}) \,.
\end{split}
\ee
Upon  identifying the coordinate the parameter $s$ with $x^0$,  these equations formally coincide with the covariant Hamilton equations. Solutions of \eqref{Eq: pre-symplectic system components} are then  the extrema of $\tilde{\ac}^\epsilon$.
\end{proof}
\end{proposition}
\noindent We have then that extrema of $\ac^\epsilon$, i.e.  elements of the solution space restricted to the collar $C^\Sigma_\epsilon$, are in one-to-one correspondence (via $\varpi$) to solutions of the pre-symplectic system $(\fpe^\Sigma,\, \os,\, \mathcal{H})$.

\begin{remark} \label{Rem: globally hyperbolic space-times}
Notice  that  the analysis above is local, since all the results we proved  are valid within a collar $C_\epsilon^\Sigma$ whose width transversal to $\Sigma$ is $\epsilon>0$. 
The possibility of globalising it strongly depends on the particular space-time under consideration.
In particular in all the previous results the space $\elag_{C_\epsilon^\Sigma}$ can be replaced by  the whole $\elag_{\mathscr{M}}$ when $\mathscr{M}$ is homeomorphic to $\Sigma \times \mathbb{R}$ for a suitable codimension-one  $\Sigma$.
This is indeed the definition of a \textit{globally hyperbolic space-time}. All examples we consider fall into this case. 
\end{remark}

\noindent  The relation \eqref{Eq: pre-symplectic system} does not uniquely determine the vector field $\mathbb{\Gamma}$ since $\Omega^\Sigma$ has a non-trivial kernel, so one has to interpret it as a condition providing the components of a suitable family of vector fields on a suitable subset in $\fpe^\Sigma$. A way  to rigorously (from the geometrical point of view) deal with such a problem is described by  the so called \textit{pre-symplectic constraint algorithm}  (see \cite{Gotay-Nester-Hinds1978-DiracBergmann_constraints} where it was  introduced, and see the appendix \ref{App: PCA} where we  recall its formulation).
\noindent The final step of such algorithm provides in particular, for the pre-symplectic system $(\mathcal{F}^{\Sigma}_{\mathcal{P}(\mathbb{E})}, \Omega^{\Sigma},  \mathcal{H})$,  a suitable immersed subset (that this is a submanifold embedded into $\mathcal{F}^{\Sigma}_{\mathcal{P}(\mathbb{E})}$ is assumed along the development of the theory, and has indeed to be analysed case by case) 
\be
\label{sette.1}
\mathfrak{i}_\infty\,:\,\mathcal{F}_\infty\,\hookrightarrow\, \mathcal{F}^{\Sigma}_{\mathcal{P}(\mathbb{E})}
\ee
where the implicit  equation \eqref{Eq: pre-symplectic system} can be written in terms of (a family) of vector fields\footnote{With respect to the terminology from the Dirac's analysis on system with singular Hamiltonian, and more generally on implicit differential equations, the space $\mathcal{F}_\infty$ is usually referred to as the \textit{constraint} manifold.}. 


\noindent Two instances may  occur at this stage.
The first occurs when 
\be
\label{sette.2}
\Omega^\Sigma_\infty \,=\, \mathfrak{i}_\infty^\star \os
\ee
is non-degenerate.
We refer to theories falling into this case as \textit{symplectic} theories.
The second occurs when  $(\mathcal{F}_\infty, \os_\infty)$ is still a pre-symplectic manifold.
We refer to theories falling into this case as   \textit{gauge} theories.

\noindent  Since $\mathcal{F}_{\infty}$ is in general an infinite dimensional Banach manifold, we say $(\mathcal{F}_\infty, \os_\infty)$ to be {\it weakly} symplectic if $\os_\infty$ is non-degenerate, and to be {\it strongly} symplectic if 
$$
\mathbf{T}_{m_\infty}\mathcal{F}_\infty \simeq \mathbf{T}^\star_{m_\infty}\mathcal{F}_\infty
$$
for any $m_\infty \in \mathcal{F}_\infty$.
From now on, within the symplectic case, we assume  $(\mathcal{F}_\infty, \os_{\infty})$ to be strongly symplectic so that for any function on $\mathcal{F}_\infty$, the corresponding Hamiltonian vector field with respect to $\os_\infty$ is both uniquely and globally (that is on the whole $\mathcal{F}_\infty$) defined. Such assumption is indeed verified for the examples we already introduced and we will further develop in the following sections, since in these cases it turns out that the set
 $\mathcal{F}_\infty$ is a Hilbert manifold for which the isomorphism above is canonically given by the Hermitian structure of the model Hilbert space.
 
\noindent Under the assumption that $(\mathcal{F}_\infty,\os_\infty)$ is strongly symplectic, the solution space $\elag_{C^\Sigma_\epsilon}$ for the problem on the collar $C_\epsilon^\Sigma$ does not coincide with the whole space of fields restricted to $\Sigma$, but with the elements in $\mathcal{F}_\infty$, since the  solutions of the original pre-symplectic system $(\mathcal{F}^{\Sigma}_{\mathcal{P}(\mathbb{E})}, \Omega^{\Sigma},  \mathcal{H})$ are given upon immersing (under the action of the map $\mathfrak{i}_\infty$ defined in \eqref{sette.1})  into $\fpe^\Sigma$  the integral curves of the vector field $\mathbb{\Gamma}_\infty$ on $\mathcal{F}_\infty$ defined, with  $\mathcal{H}_\infty \,=\, \mathfrak{i}_\infty^\star \, \mathcal{H}$,  by 
\be \label{Eq: pre-symplectic system minfty}
i_{\mathbb{\Gamma}_\infty}\Omega_\infty^\Sigma  \,=\, \dd \mathcal{H}_\infty. 
\ee
Since $\Omega_\infty^\Sigma$ is strongly symplectic, for each point $m_\infty \in \mathcal{F}_\infty$ there exists a unique such integral curve and therefore a unique solution in $\widetilde{\elag}^\epsilon$, so that $\mathcal{F}_\infty$ can be identified with the set of allowed Cauchy data for the problem.  
Such a  unique solution associated to the Cauchy datum $m_\infty$ allows to define a map $\Psi\,:\,\mathcal{F}_\infty\,\to\,\widetilde{\elag}^\epsilon$ as  
\be
\label{20ott-ale1}
\gamma_{(s)}\,=\,\Psi(m_\infty)\,=\,\mathfrak{i}_\infty \circ F_s^{\mathbb{\Gamma}_\infty} (m_\infty)\,\in\,\widetilde{\elag}^\epsilon\,\subset\,\Gamma(\fpe^\Sigma)
\ee 
where $F^{\mathbb{\Gamma}_\infty}$ is the flow generated by $\mathbb{\Gamma}_\infty$ with evolution parameter $s\in[0,\epsilon)$, and to have 
\be
\label{20ott-ale2}
\chi_{_{m_\infty}} \,=\,\varpi^{-1}\circ\Psi(m_\infty)\,\in\,\elag_{C^\Sigma_\epsilon}\,\subset\,\fpe^{C^\Sigma_\epsilon}
\ee
in terms of the identification $\varpi$ defined in \eqref{20ott-ale3}, where the subscript $m_\infty$ in $\chi_{_{m_\infty}}$ stresses that $\chi_{_{m_\infty}}$ is the solution associated to the Cauchy datum $m_\infty$. 

\noindent The map $\Psi$ turns to be bijective and suitably regular.
It is injective by virtue of the existence and uniqueness theorem for the solutions of the equations corresponding to the vector field $\mathbb{\Gamma}_\infty$ that is uniquely determined by the relation \eqref{Eq: pre-symplectic system minfty} on $\mathcal{F}_\infty$ since $\os_\infty$ is strongly symplectic. The non degeneracy of the 2-form $\os_\infty$ allows, within the formalism of the pre-symplectic constraint algorithm (whose details we recall in the appendix \ref{App: PCA}) also to prove that the range of $\Psi$ coincides with the whole  $\widetilde{\elag}^\epsilon$.  Since the embedding \eqref{sette.1} is assumed to be smooth, the regularity of the map $\Psi$ coincides with the regularity of the flow $F_s^{\mathbb{\Gamma}_\infty}$, which depends on the regularity of the vector field $\mathbb{\Gamma}_\infty$.

\noindent 
It is clear that the inverse of $\Psi$  maps then any curve $\gamma_{(s)}\in\widetilde\elag^\epsilon$ into its Cauchy datum. We represent its action as 
\be 
\Psi^{-1}(\gamma_{(s)}) \,=\, {F_s^{\mathbb{\Gamma}_\infty}}^{-1} \left( \mathfrak{i}^{-1}_\infty \gamma_{(s)}\right):
\ee 
the map $s\in[0,\epsilon)\,\mapsto\,\mathfrak{i}^{-1}_\infty\gamma_{(s)}$  denotes  the unique curve on $\mathcal{F}_\infty$ whose image under the embedding $\mathfrak{i}_\infty$ gives the curve $\gamma_{(s)}\in\widetilde{\elag}^\epsilon$ (that is the curve $\mathfrak{i}^{-1}_\infty\gamma_{(s)}$ is an integral curve of the Hamiltonian vector field $\mathbb{\Gamma}_\infty$ on $\mathcal{F}_\infty$ given in \eqref{Eq: pre-symplectic system minfty} since $\os_\infty$ is symplectic); for any given $s\in[0,\epsilon)$, the action of  ${F_s^{\mathbb{\Gamma}_\infty}}^{-1}$ evolves the point $\mathfrak{i}^{-1}_\infty\gamma_{(s)}$ along the curve backward to the initial $s=0$ point $m_{\infty}\in\mathcal{F}_\infty$. 

\noindent The map whose action is defined in \eqref{20ott-ale2} provides   a one-to-one smooth correspondence between elements  $m_\infty \in \mathcal{F}_\infty$ and elements in  $\elag_{C^\Sigma_\epsilon}$, that is the space of solutions of the variational problem restricted to the collar\footnote{Eventually, if regularity conditions mentioned in remark \ref{Rem: globally hyperbolic space-times} are met, these solutions turn to be global, that is defined on the whole $\m$.  } $C^\Sigma_\epsilon$. Such correspondence coincides with the restriction to $\elag_{C_\epsilon^\Sigma}$ of the action of the map $\Pi_\Sigma\,:\,\mathcal{F}^{C^\epsilon_\Sigma}_{\mathcal{P}(\mathbb{E})}\,\to\,\mathcal{F}^\Sigma_{\mathcal{P}(\mathbb{E})}$ 
defined in  \eqref{sei.1}, so that \textit{for such restriction} we can write 
\be
\label{21ott-ale6}
\Pi_{\Sigma}\,=\,\mathfrak{i}_\infty \circ \Psi^{-1} \circ \varpi,
\ee 
along with the diagram
\be
\label{dia1}
\begin{tikzcd}
\widetilde{\elag}^\epsilon \arrow[r, "\varpi^{-1}"] & \elag_{C^\Sigma_\epsilon} \arrow[dl, "\Pi_\Sigma", swap] \\
\mathcal{F}_{\mathcal{P}(\mathbb{E})}^\Sigma 
 \\
\mathcal{F}_\infty \arrow[u, "\mathfrak{i}_\infty", hookrightarrow] \arrow[uu, bend left = 80, "\Psi", dashed] 
\end{tikzcd}
\ee
We can then write, for the restriction to $\elag_{C^\Sigma_\epsilon}$ of the 2-form $\Pi_\Sigma\Omega^\Sigma$ (a restriction that is canonical, i.e. it does not depend on $\Sigma$, as  analysed in the previous section, see \eqref{21ott-ale3}),  the relation 
\be
\label{sette.3}
\Pi_\Sigma^\star \Omega^\Sigma \,=\, \left(\, \mathfrak{i}_\infty \circ \Psi^{-1} \circ \varpi \,\right)^\star \Omega^\Sigma \,=\, \left(\, \Psi^{-1} \circ \varpi  \,\right)^\star \,  \Omega_\infty^\Sigma.
\ee
This shows that the restriction to $\elag_{C^\Sigma_\epsilon}$ of the 2-form  $\Pi_\Sigma^\star \Omega^\Sigma$ is the pull-back via $\Psi^{-1} \circ \varpi$ of the strongly symplectic structure $ \Omega_\infty^\Sigma$ on $\mathcal{F}_\infty$.
Since $\Psi$ and $\varpi$ are diffeomorphisms (and then the restriction of $\Pi_\Sigma$ to $\elag_\m^\epsilon$ written in \eqref{21ott-ale6} is a diffeomorphism as well), restricting $\Pi_\Sigma^\star \Omega^\Sigma$ to the solution space provides a strongly symplectic 2-form, corresponding via a diffeomorphism to the strongly symplectic 2-form $\os_\infty$ on $\mathcal{F}_\infty$.

\noindent When $(\mathcal{F}_\infty, \Omega^\Sigma_\infty)$ is pre-symplectic, the manifold  $\mathcal{F}_\infty$ is no longer diffeomorphic to  $\widetilde{\elag}^{\epsilon}$, since the vector field $\mathbb{\Gamma}_\infty$ (and then the map $\Psi$) implicitly defined in \eqref{Eq: pre-symplectic system minfty} is no longer uniquely determined.
Notice that, under such conditions, if the 2-form $\Omega_\infty^\Sigma$ has constant rank, then the quotient space given by 
\be 
\label{22ott-ale1}
\pi_\infty\,:\,\mathcal{F}_\infty\,\to\,\mathcal{R}_\infty\,=\,\mathcal{F}_\infty/\mathrm{ker}\,\Omega^\Sigma_\infty
\ee
has a manifold structure. On such a quotient space the 2-form $\tilde{\os}_\infty$ with 
$$
\Omega_\infty^\Sigma \,=\, \pi_\infty^\star \tilde{\Omega}^\Sigma_\infty
$$
is well defined so that $(\mathcal{R}_\infty, \tilde\Omega^\Sigma_\infty)$ turns to be a symplectic manifold, on which a Hamiltonian  $\tilde{\mathcal{H}}_\infty $ with  
$$
\mathcal{H}_\infty\,=\, \pi_\infty^\star \tilde{\mathcal{H}}_\infty
$$
exists. 
Along the lines described in the previous symplectic case, we write in such a pre-symplectic case
\be
\label{dia2}
\begin{tikzcd}
\widetilde{\elag}^\epsilon \arrow[r, "\varpi^{-1}"] & \elag_{C^\Sigma_\epsilon } \arrow[dl, "\Pi_\Sigma", swap] \\
\mathcal{F}^\Sigma_{\mathcal{P}(\mathbb{E})} 
\\
\mathcal{F}_\infty \arrow[d, "\pi_\infty"] \arrow[u, "\mathfrak{i}_\infty", hookrightarrow]  \\
\mathcal{R}_\infty. 
\end{tikzcd}
\ee
Elements of the quotient manifold $\mathcal{R}_\infty$ are equivalence classes of fields, and explicitly writing differential forms on such a quotient is rather involved. For this reason, we mimic  a procedure described in  \cite{Dubrovin-Giord-Marmo-Sim1993-Poisson_presymplectic}, which allows to consider suitable differential forms on $\mathcal{F}_\infty$. It is as follows.

\noindent Since the 2-form $\Omega^\Sigma_{\infty}$ on $\mathcal{F}_\infty$ has a non trivial kernel, Hamiltonian vector fields associated to elements in $\mathscr{F}(\mathcal{F}_\infty)$, i.e. functions on $\mathcal{F}_\infty$, are not uniquely determined.
Such a correspondence can be made unique via introducing a \textit{connection} for the fibration $\pi_{\infty}$ in \eqref{22ott-ale1}, that is an operator
$$
\mathcal{P} \,\colon\,\mathfrak{X}({\mathcal{F}}_{\infty}) \,\rightarrow \, \mathfrak{X}({\mathcal{F}}_{\infty})
$$
which is (left) linear with respect to the (left) $\mathscr{F}(\mathcal{F}_\infty)$-module structure of the set of vector fields on $\mathcal{F}_\infty$, idempotent (of order 2) and with a range given by the  (integrable) distribution $\mathrm{Ker}\,\Omega^\Sigma_{\infty}$. 
 The distribution $\mathrm{ker}\,\mathcal{P}$ defines the  so called \textit{horizontal} vector fields corresponding to the given connection. 
Among the vector fields $\mathbb{\Gamma_\infty}$ satisfying
\begin{equation}
\label{22ott-ale4}
i_{\mathbb{\Gamma}_\infty}\Omega^\Sigma_{\infty}\, = \,\mathrm{d}\mathcal{H}_{\infty},
\end{equation}
a connection allows to uniquely fix the vector field $\mathbb{\Gamma_\infty}$ satisfying
\begin{equation}\label{projection compatibility}
 \mathcal{P}(\mathbb{\Gamma_\infty}) = 0 \,,
\end{equation}  
namely the one which is horizontal with respect to the chosen connection. When the connection is flat, it turns out to be integrable (see, for instance \cite{Dubrovin-Giord-Marmo-Sim1993-Poisson_presymplectic}), i.e. the distribution of horizontal vector fields it generates is completely integrable.
The integral manifolds of such a distribution give rise to a foliation of $\mathcal{F}_\infty$ whose leaves are transversal to the fibers of $\pi_\infty$, and then each of these leaves amounts to a section $\sigma_{\mathcal{P}}\,:\,\mathcal{R}_{\infty}\,\to\,\mathcal{F}_\infty$ of the fibration $\pi_\infty$.
It turns out in this case that the 2-form $\os_\infty$ is strongly symplectic along the leaves of $\mathcal{P}$, i.e. it defines a bijection between the set of horizontal tangent vectors associated to the connection $\mathcal{P}$ and its dual. 

\begin{remark}
We notice that in this paper we consider examples of abelian gauge theories (namely electrodynamics) for which the flatness condition is valid. We are not considering  the same class of examples for which the construction given in \cite{Khavkine2014-Covariant_phase_space} can not be performed globally on the space of fields. The formalism we describe in this paper can indeed be extended to the cases of non-abelian gauge theories for which such a  flatness condition is not valid. This is the topic of the second part of our work, based on a suitable application of the coisotropic embedding theorem along the lines of \cite{Ciaglia-DC-Ibort-Marmo-Schiav-Zamp-2022-Symmetry}, where a momentum map within the pre-symplectic setting has been studied. 
\end{remark}
\noindent In analogy to the previous case of a symplectic theory (see \eqref{20ott-ale1}), one can define a map $\Psi_\mathcal{P}$
\be
\label{22ott-ale3}
\sigma_{\mathcal{P}}(\mathcal{R}_\infty)\,\ni\,m_\infty\quad\mapsto\quad \mathfrak{i}_\infty\,\circ\,F_s^{\mathbb{\Gamma}_\infty}(m_\infty)\,\in\,\widetilde{\elag}^\epsilon, 
\ee 
where $F_s^{\mathbb{\Gamma}_\infty}$ is the flow generated by $\mathbb{\Gamma}_{\infty}$, which is the unique vector field in $\mathfrak{X}(\sigma_{\mathcal{P}}(\mathcal{R}_\infty))$ that satisfies the conditions \eqref{22ott-ale4} and \eqref{projection compatibility}. 
Notice that,  since we are considering initial conditions on the image of a specific section $\sigma_{\mathcal{P}}$ rather than on the whole $\mathcal{F}_\infty$, the map $\Psi_{\mathcal{P}}$ may not range over the whole $\widetilde{\elag}^\epsilon$. 
The range $\mathrm{Im}(\Psi_{\mathcal{P}})$ (that we denote as ${\widetilde{\elag}}^\epsilon_{\sigma_{\mathcal{P}}}$) turns to be diffeomorphic (via $\varpi$) to a subset of $\elag_{C^\Sigma_\epsilon}$ that we denote as $\elag_{C^\Sigma_\epsilon}^{\sigma_\mathcal{P}}$. 
Such a set  represents the solution space of the variational problem associated to the chosen connection $\mathcal{P}$. This amounts to a \textit{gauge choice}.  The analogue of the diagram \eqref{dia1} is
\be
\label{dia3}
\begin{tikzcd}
{\widetilde{\elag}}^\epsilon_{\sigma_{\mathcal{P}}}  \arrow[r, "\varpi^{-1}"] & \elag_{C^\Sigma_\epsilon}^{\sigma_\mathcal{P}}    \arrow[dl, "\Pi_\Sigma", swap] \\
\mathcal{F}^\Sigma_{\mathcal{P}(\mathbb{E})} 
\\
\sigma_{\mathcal{P}}(\mathcal{R}_\infty) \subset \mathcal{F}_\infty \arrow[uu, "\Psi_{\mathcal{P}}", bend left = 80, dashed] \arrow[d, "\pi_\infty"] \arrow[u, "\mathfrak{i}_\infty", hookrightarrow]  \\
\mathcal{R}_\infty \arrow[u, bend left=30, dashed, "\sigma_{\mathcal{P}}"].   & 
\end{tikzcd}
\ee



\section{Poisson brackets on the solution space}
\label{Sec: Poisson brackets on the solution space}

\noindent 
In the previous sections we introduced a (strongly) symplectic structure on the solution space of  (a suitable class of) first order Hamiltonian field theories, both within what we defined as a symplectic theory and within what we defined  as a gauge theory. 
For a symplectic theory,  the  corresponding strongly symplectic structure has been  defined on the whole solution space, whereas for a gauge theory it has been  defined on a suitable subset of it,  corresponding to a specific gauge fixing. Such a gauge fixing is geometrically formulated as the choice of a suitable flat connection. 
 The Poisson bracket straightforwardly comes from the strongly symplectic 2-form introduced above, in analogy to the case of a symplectic structure on a finite dimensional smooth manifold. 
In this section we explicitly write such a Poisson tensor for the class of examples we already considered. 
The starting point can be either the 2-form $\pssos$ on $\elag_{C_\epsilon^\Sigma}$ or the 2-form $\os_\infty$ on $\mathcal{F}_\infty$. Since they are diffeomorphic, the corresponding brackets turn to be equivalent\footnote{Notice that  the choice of $\os_\infty$ is more convenient from the computational point of view, since it is a $2$-form on the manifold $\mathcal{F}_\infty$, while $\pssos$ is a $2$-form on the restriction  $\elag_{C_\epsilon^\Sigma}\subset\fpe$ of space of solutions of the field equations.}.

\noindent We begin by considering a symplectic theory. 
Let $F\,,G\in\,\mathscr{F}(\elag_{C_\epsilon^\Sigma})$. The Hamiltonian vector field $\mathbb{X}_F\in\mathfrak{X}(\elag_{C_\epsilon^\Sigma})$ corresponding to $F$ with respect to the strongly symplectic structure $\pssos$ in \eqref{sette.3} is the one satisfying 
\be \label{Eq: hamiltonian vector field solution space}
i_{\mathbb{X}_F}(\pssos) \,=\, \dd F \,,
\ee
so the Poisson bracket can be defined by 
\be
\{G, F\}\,=\, \pssos(\mathbb{X}_G,\, \mathbb{X}_F) \,=\, \mathbb{X}_F(G) \,.
\ee
Given the element $F\,\in\,\mathscr{F}(\elag_{C^\Sigma_\epsilon})$, consider its pull-back to  $\mathscr{F}({\mathcal{F}_\infty})$ given by 
$$
f \,=\, (\varpi^{-1} \circ \Psi)^\star F.
$$ 
The Hamiltonian vector field $\mathbb{X}_f\,\in\,\mathfrak{X}(\mathcal{F}_\infty)$ associated to $f$ with respect to $\os_\infty$ is the one satisfying
\be \label{Eq: hamiltonian vector field cauchy data}
i_{\mathbb{X}_f}\os_\infty\,=\, \dd f,
\ee
so that, for any $g\,\in\,\mathscr{F}(\mathbb{F}_\infty)$, the corresponding  Poisson bracket reads 
\be
\{g,\, f \}\,=\, \os_\infty(\mathbb{X}_g,\, \mathbb{X}_f) \,=\, \mathbb{X}_f(g) \,.
\ee
The relation between these two brackets can be analysed as follows.
The equation \eqref{Eq: hamiltonian vector field solution space} can be written, recalling that one has  $\pssos \,=\, \left(\Psi^{-1} \circ \varpi\right)^\star \os_\infty$ and $F \,=\, \left(\Psi^{-1} \circ \varpi \right)^\star f$, as
\be \label{Eq: relation hamiltonian vector field solutions and cauchy data}
\os_\infty\left(\, \left(\Psi^{-1} \circ \varpi\right)_\star \mathbb{X}_F,\, \left(\Psi^{-1} \circ \varpi \right)_\star \mathbb{Y} \,\right) \,=\, \dd f \left(\left(\Psi^{-1} \circ \varpi \right)_\star \mathbb{Y} \right) \;\;\; \forall \,\, \mathbb{Y} \in \mathfrak{X}(\elag_{C_\epsilon^\Sigma}) \,
\ee
since the map $\Psi^{-1}\circ\varpi$ is a diffeomorphism, and then any vector field on $\mathcal{F}_\infty$ can be written in terms of a vector field $\mathbb{Y}\in\,\mathfrak{X}(\elag_{C_\epsilon^\Sigma})$.
The comparison between  \eqref{Eq: hamiltonian vector field cauchy data} and \eqref{Eq: relation hamiltonian vector field solutions and cauchy data} gives the following relation between $\mathbb{X}_F$ and $\mathbb{X}_f$, that is
\be
\mathbb{X}_F \,=\, \left(\varpi^{-1} \circ \Psi \right)_\star \mathbb{X}_f \,=\, \left( \varpi^{-1} \circ \mathfrak{i}_\infty \circ F_s^{\mathbb{\Gamma}_\infty} \right)_\star \mathbb{X}_f \,=\, \varpi^{-1}_\star \circ  {\mathfrak{i}_\infty}_\star \circ {F_s^{\mathbb{\Gamma}_\infty}}_\star \, \mathbb{X}_f\,.
\ee
Recalling the definition of push-forward of a vector field under the action of a diffeomorphism on a manifold and the definition of the tangent lift of a vector field (see \cite[Sec. $2.1$]{Morandi-Ferrario-LoVecchio-Marmo-Rubano1990-Inverse_Problem}), we can write
$${F_s^{\mathbb{\Gamma}_\infty}}_\star \,=\, {F_s^{T\mathbb{\Gamma}_\infty}},$$
where 
 $F_s^{T\mathbb{\Gamma}_\infty}$ denotes the flow of the tangent lift to $\mathbf{T}\mathcal{F}_\infty$ of the vector field $\mathbb{\Gamma}_\infty$ on $\mathcal{F}_\infty$. 
It turns out that the flow  $F_s^{T\mathbb{\Gamma}_\infty}$ acts fiberwise linearly upon sections of the tangent bundle $\mathbb{T}\mathcal{F}_\infty$, i.e. linearly upon the components of $\mathbb{X}_f$, and such a linear action comes as a first order linearisation of the flow generated by $\mathbb{\Gamma}_\infty$ on $\mathcal{F}_\infty$.  Notice that such a linearisation recalls the formulation developped by Peierls (see \cite{Peierls1952-Commutation_laws}) where the commutation laws among fields were introduced in terms of a linearised evolution.

This means that $\mathbb{X}_F$ is reconstructed from $\mathbb{X}_f$ along the following steps
\begin{itemize}
\item linearise the flow generated by  $\mathbb{\Gamma}$ and let it act upon the element $\mathbb{X}_f$;
\item let ${\mathfrak{i}_\infty}_\star$ act upon such a solution. This amounts to (as we will see in the examples) impose constraint relations coming as  the tangent lift of the constraint equations emerging from the pre-symplectic constraint algorithm;
\item let $\varpi^{-1}_\star$ act. This amounts, again, to identify the evolution parameter $s$ from the above flow  with the coordinate transversal to the hypersurface $\Sigma$.
\end{itemize}

\noindent This analysis allows to write the following relations, 
\be \label{Eq: pull-back bracket}
\begin{split}
\{G,\, F\}_\chi \,&=\, \left(\pssos\right)_\chi(\mathbb{X}_G,\, \mathbb{X}_F) \\ &=\, \left[\left(  \Psi^{-1} \circ \varpi \right)^\star \os_\infty \right]_\chi (\mathbb{X}_G,\, \mathbb{X}_F) \\
&=\, \left[{\os_\infty}\right]_{\ps (\chi)} \left(\left( \Psi^{-1} \circ \varpi \right)_\star \,\mathbb{X}_G,\, \left(  \Psi^{-1} \circ \varpi \right)_\star \, \mathbb{X}_F \right)  \\ 
&=\,\left[{\os_\infty}\right]_{\ps(\chi)}(\mathbb{X}_g,\, \mathbb{X}_f) \\ &=\, \{g,\, f\}_{\ps(\chi)} \,=\, \pss \{g,\,f\}_\chi  \,.
\end{split}
\ee
Before presenting such Poisson structures within the theories we have introduced in the previous sections, we notice  that the above construction can be performed also in the case of a gauge theory, as we will see in the example in section \ref{Subsec: Gauge theories: Electrodynamics}. The Hamiltonian vector field $\mathbb{X}_f$ associated to $f\in\mathscr{F}(\mathcal{F}_\infty)$  with respect to $\os_\infty$ comes upon the further condition  to be horizontal with respect to a suitable flat connection, that we write as  $\mathcal{P}(\mathbb{X}_f) \,=\, 0$.



\subsection{Mechanical systems: the free particle on the line}
\label{Subsec: Mechanical systems: the free particle on the line}

Here we give the explicit construction of the Poisson bracket introduced in the previous sections on the solution space of the equations of the motion of the free particle moving on the line, that we described in the example \ref{Ex: free particle variational principle}. In particular, we write down the action of the bracket upon the following functions of the solutions:
\be
\label{28ott-ale1}
F \,=\, q(t_1) \,, \qquad G \,=\, q(t_2) \,, 
\ee
giving the position of the particle at the two particular time values $t_1$ and $t_2$.

\noindent Any slice $\Sigma\hookrightarrow\mathbb{I}$ of the open interval $\mathbb{I}=(a,b)\subset\mathbb{R}$ is identified by selecting a $\bar{t}\in\mathbb{I}$, so that the space  
 $\fpe^\Sigma$ is given by trajectories $\chi(t) \,=\, (q(t),\, p(t))$ restricted to $\bar{t}$, i.e. $\bar{\chi} \,=\, \chi\bigr|_{\bar{t}}$. Since such a $\Sigma$ is a single point space,  the measure  $vol_\Sigma$ is singular so that we write the canonical  structure $\pssos$ on the solution space as
\be
\left[\pssos (\mathbb{X},\, \mathbb{Y})\right]\big|_\chi \,=\, \int_\Sigma vol_\Sigma \left( \mathbb{X}_q \mathbb{Y}_p - \mathbb{X}_p \mathbb{Y}_q \right) \,=\, \left( \mathbb{X}_q \mathbb{Y}_p - \mathbb{X}_p \mathbb{Y}_q \right)_{t = \bar{t}} \,, 
\ee
and then the 2-form $\os$ reads
\be
\left[\os(\mathbb{X},\, \mathbb{Y})\right]\big|_{\bar\chi} \,=\, \left(\, \mathbb{X}_{\bar{q}} \mathbb{Y}_{\bar{p}} - \mathbb{X}_{\bar{p}} \mathbb{Y}_{\bar{q}} \,\right) \,, 
\ee
which is indeed the symplectic structure of the cotangent bundle  $\mathbf{T^\star}\mathbb{R}$.
In this example the set $\fpe^\Sigma$ is already a (finite-dimensional) symplectic manifold, so that there is no need to apply the pre-symplectic constraint algorithm: we have $\mathcal{F}_\infty\,=\,\fpe$ with coordinates $(\bar q, \bar p)$. 

\noindent The solutions of the equations of the motion written in example \ref{Ex: free particle variational principle} for the Cauchy datum $(\bar{q},\, \bar{p})$ at $\bar{t}$ are
\be
\label{otto.1}
\begin{split}
&q(t) \,=\, \bar{q} + \frac{\bar{p}}{m} (t - \bar{t}) \\ & p(t) \,=\, \bar{p} \end{split}
\ee
so that one has, for $f \,=\, (\varpi^{-1} \circ \Psi)^\star F$ and $g \,=\, (\varpi^{-1} \circ \Psi)^\star G$ from the functions \eqref{28ott-ale1}  in $\mathscr{F}(\elag_{C^\Sigma_{\epsilon}})$
\be
\begin{split}
&f \,=\, \bar{q} + \frac{\bar{p}}{m}(t_1 - \bar{t}), \\ & g \,=\, \bar{q} + \frac{\bar{p}}{m}(t_2 - \bar{t}). \end{split} 
\ee
The Poisson bracket between $f$ and $g$ is readily computed using the structure $\os$ written above.
It gives 
\be
\{g,\, f\}_{(\bar{q},\, \bar{p})} \,=\, \os(\mathbb{X}_g,\, \mathbb{X}_f) \,=\, \mathbb{X}_f (g),  
\ee
where the vector field $\mathbb{X}_f$ is defined by 
\be
i_{\mathbb{X}_f}\os \,=\, \dd f \,,
\ee
so that 
\be
\mathbb{X}_f \,=\, \frac{t_1 - \bar{t}}{m} \frac{\de}{\de \bar{q}} - \frac{\de}{\de \bar{p}} \,.
\ee
One has then, for the Poisson bracket, the following expression, which is a  function of $(\bar{q},\, \bar{p})$ depending on the parameters $t_1$ and $t_2$
\be
\label{otto.2}
\{g,\, f\}_{(\bar{q},\,\bar{p})} \,=\, \frac{t_1 - t_2}{m} \,. 
\ee

\noindent In order to compute the bracket between the original functions $F$ and $G$, one needs to determine the Hamiltonian vector field $\mathbb{X}_F$  defined in \eqref{Eq: hamiltonian vector field solution space} upon linearizing the flow defined in \eqref{otto.1}. Since such dynamics is linear, one has 
\be
\begin{split}
&{\mathbb{X}_F}_q \,=\, {\mathbb{X}_f}_{\bar{q}} + \frac{{\mathbb{X}_f}_{\bar{p}}}{m}(t - \bar{t}), \\  &{\mathbb{X}_F}_p \,=\, {\mathbb{X}_f}_{\bar{p}} \, \end{split}
\ee
and then 
\be
\{G,\, F\}_\chi \,=\, \left[\pssos(\mathbb{X}_F,\, \mathbb{X}_G)\right]\Big|_\chi \,=\, \mathbb{X}_F(G) \,.
\ee
This reads 
\be
\{G,\, F\}_\chi \,=\, \frac{t_1 - t_2}{m}
\ee
(as a function of $\chi$) for the given $F,G$ in $\elag_{C^\Sigma_\epsilon}$. 
Notice  that the bracket between $F$ and $G$ is the pull-back, via the restriction map to $\bar{t}$ as it comes from \eqref{Eq: pull-back bracket},  of the bracket between $f$ and $g$ written in \eqref{otto.2}.

\subsection{A symplectic theory: the massive vector boson field}
\label{Subsec: The symplectic case: massive vector boson field}

\noindent Here we give the explicit construction of the bracket between the functions on the solution space of equations \eqref{Eq: euler-lagrange equations vector boson field} given by 
\be
F \,=\, \phi^{a_1}\left(\psi^{-1}_{U_{\mathbb{M}^{1,3}}}(x^0_1, \underline{x}_1)\right),  \qquad G \,=\, P^{0}_{a_2}\left(\psi^{-1}_{U_{\mathbb{M}^{1,3}}}(x^0_2, \underline{x}_2)\right) \,,
\ee
that is $F$ amounts to evaluate  the field configuration $\phi^{a_1}$ at the given point $\psi^{-1}_{U_{\mathbb{M}^{1,3}}}(x^0_1, \underline{x}_1)$, while $G$ amounts to evaluate  the ($0$-component of the) momentum -- solutions of the equations of the motion -- at the given point  $\psi^{-1}_{U_{\mathbb{M}^{1,3}}}(x^0_2, \underline{x}_2)$ in $\mathbb{M}^{1,3}$.
Consider the setting of example \ref{Ex: vector boson field variational principle}, and fix the hypersurface 
$$
\Sigma \,=\,  \{x^0 \,=\, x^0_\Sigma \} \,=\, \mathbb{R}^3 \subset \mathbb{M}^{1, 3}, 
$$
which gives a collar  $C_\epsilon^\Sigma\,=\,[x^0_\Sigma, x^0_\Sigma + \epsilon) \times \mathbb{R}^3$ (notice that  $\epsilon$ may possibly tend to  $+\infty$) and the space of curves $\gamma_{(s)}$ on $\fpe^\Sigma$ is defined for the parameter $s\,\in\,[x^0_\Sigma, \epsilon)$\footnote{Notice that the analysis can be performed for  the collar $(x^0_\Sigma - \epsilon, x^0_\Sigma] \times \mathbb{R}^3$. This would amount that curves $\gamma_{(s)}$ on $\mathcal{F}_{\mathcal{P}(\mathbb{E})}^\Sigma$ can be defined for the value of the  parameter $s\,\in(-\epsilon, x^0_\Sigma]$.}.
We recall that, provided the elements in $\fpe$ have the regularity described in the example \ref{Ex: vector boson field variational principle}, the set $\fpe^\Sigma$ turns out to be a Hilbert space. 
Indeed one has
\be
\fpe \,=\, \prod_a {\mathcal{H}^{2}(\mm, vol_{\mm})}^a \times \prod_{\mu, a} {\mathcal{H}^{1}(\mm, vol_{\mm})}^\mu_a \,:
\ee
the  \textit{trace theorem}
for the restriction of fields from $\fpe$ to $\fpe^\Sigma$, from the {\it trace theorem} (see \cite{Lions1990-vol2}) one has  
\be
\fpe^\Sigma \,=\, \prod_a {\mathcal{H}^{\frac{3}{2}}(\Sigma, vol_{\Sigma})}^a \times \prod_{\mu, a} {\mathcal{H}^{\frac{1}{2}}(\Sigma, vol_{\Sigma})}^\mu_a \,,
\ee
which is  a Hilbert space whose elements are denoted, in analogy to what we already described, as 
\be
\chi_\Sigma \,=\, (\varphi, p, \beta) \,.
\ee
Following the procedure outlined in the previous sections, we consider the pre-symplectic system given by $$(\mathcal{F}^\Sigma_{\mathcal{P}(\mathbb{E})},\Omega^\Sigma, \mathcal{H})$$ where the 2-form  
$\Omega^\Sigma \,=\, \rho_\Sigma^\star \, \nu^\star \, \bar{\omega}^\Sigma$ is written as (see \eqref{omega1s})
\be
\Omega^\Sigma(\mathbb{X},\, \mathbb{Y}) \,=\, \int_\Sigma  vol_\Sigma \left( \mathbb{X}_{\varphi^a} \mathbb{Y}_{p_a} - \mathbb{X}_{p_a} \mathbb{Y}_{\varphi^a} \right) 
\ee
with $\mathbb{X},\, \mathbb{Y}$  vector fields on $\mathcal{F}_{\mathcal{P}(\mathbb{E})}^\Sigma$.  Evaluating the Hamiltonian  along the curve $\gamma_{(s)}$ results in, taking into account the expression \eqref{Eq: klein-gordon hamiltonian},
\be
\mathcal{H}(\gamma_{(s)}) \,=\, \int_\Sigma vol_\Sigma\left[\, {\beta^j_a}_{(s)} \partial_j{\varphi^a}_{(s)} - \frac{1}{2}\left( \eta_{jk} \delta^{ab} {\beta^j_a}_{(s)} {\beta^k_b}_{(s)} + \delta^{ab} {p_a}_{(s)} {p_b}_{(s)} + m^2 \delta_{ab} \varphi_{(s)}^a \varphi_{(s)}^b \right) \,\right] \,.
\ee
The 2-form $\Omega^\Sigma$ is clearly degenerate, its kernel being spanned by all elements in $\mathfrak{X}(\mathcal{F}^\Sigma_{\mathcal{P}(\mathbb{E})})$ with null components along any basis for the subspace $\mathfrak{X}(\mathcal{F}^\Sigma_{\mathcal{P}(\mathbb{E}),0})$, i.e. those elements that we can write along the basis ${\mathbb{X}_{\beta^j_a}}$, that we collectively denote as $\mathbb{X}_\beta$. The space originating from the first step of the pre-symplectic constraint algorithm (see the appendix \ref{App: PCA}) is 
\be
\mathfrak{i}_1 (\mathcal{F}_1) \,=\, \left\{\, \chi_\Sigma \in \mathcal{F}_{\mathcal{P}(\mathbb{E})}^\Sigma \;\; : \;\; \dd \mathcal{H}(\mathbb{X}_\beta) \,=\, 0 \right\} \,.
\ee
A direct computation shows that
\be \label{Eq: constraints section klein-gordon}
\mathfrak{i}_1 (\mathcal{F}_1) \,=\, \left\{\, \chi_\Sigma \in \mathcal{F}_{\mathcal{P}(\mathbb{E})}^\Sigma \;\; : \;\; {\beta^j_a} \,=\, \delta_{ab} \eta^{jk} \partial_k \varphi^b \,\right\} \,.
\ee
This shows that $\mathcal{M}_1$ coincides with $\mathcal{F}^\Sigma_{\mathcal{P}(\mathbb{E}),\,0}$, i.e. the space of $\varphi$'s and $p$'s, whose immersion into $\mathcal{F}_{\mathcal{P}(\mathbb{E})}^\Sigma$ is given by the previous relation.
Pulling-back to $\mathcal{F}_1$ the 2-form  $\Omega^\Sigma$ and the Hamiltonian  $\mathcal{H}$ gives 
\be
\begin{split}
\mathfrak{i}_1^\star \Omega^\Sigma (\mathbb{X},\, \mathbb{Y}) \,&=\, \Omega_1^\Sigma (\mathbb{X},\, \mathbb{Y}) \,=\, \int_\Sigma vol_\Sigma \left(\, \mathbb{X}_{\varphi^a} {\mathbb{Y}_{p_a}} - {\mathbb{X}_{p_a}} \mathbb{Y}_{\varphi^a} \,\right) \,, \\
\mathfrak{i}_1^\star \mathcal{H} \,=\, \mathcal{H}_1 \,&=\, \frac{1}{2} \int_\Sigma vol_\Sigma \left(\, \delta_{ab} \eta^{jk} \partial_j \varphi^a \partial_k \varphi^b - \delta^{ab}p_a p_b - m^2 \delta_{ab} \varphi^a \varphi^b \,\right)  \,. 
\end{split}
\ee
The 2-form $\Omega_1^\Sigma$ is  symplectic,  the pre-symplectic constraint algorithm stabilizes, with  
$$
\mathcal{F}_\infty \,=\, \mathcal{F}_1 \,=\, \mathcal{F}^\Sigma_{\mathcal{P}(\mathbb{E}),\,0}
$$
and $\mathfrak{i}_\infty \,=\, \mathfrak{i}_1$. Such embedding turns to be differentiable with respect to the differential structure of $\mathcal{F}_\infty$ inherited from the one on  $\fpe^{\Sigma}$ and represents the section $\sigma$ of the projection (see \eqref{sei.1}) $\rho_\Sigma$ given by \eqref{Eq: constraints section klein-gordon}, as in 
\be
\begin{tikzcd}
\widetilde{\elag}^\epsilon \arrow[r, "\varpi^{-1}"] & \elag_{C^\Sigma_\epsilon}  \arrow[dl, "\Pi_\Sigma", swap] \\
\mathcal{F}_{\mathcal{P}(\mathbb{E})}^\Sigma \arrow[d, "\rho_\Sigma", shift right=1, swap] & \\ \mathcal{F}_\infty \simeq \mathcal{F}^\Sigma_{\mathcal{P}(\mathbb{E}),\, 0} \arrow[u, "\mathfrak{i}_\infty \,=\, \sigma", hookrightarrow, shift right=1, swap] \arrow[uu, bend left = 80, "\Psi", dashed] &
\end{tikzcd}
\ee



\noindent The pre-symplectic system reduced to $\mathcal{F}_\infty$ gives rise to the following equations,
\be 
\Omega_\infty^\Sigma (\mathbb{\Gamma}_{(\varphi,p)},\, \mathbb{Y}_{(\varphi,p)}) \,=\, \dd \mathcal{H}_\infty(\mathbb{Y}_{(\varphi,p)}) \,,
\ee 
that we write as 
\be
\begin{cases}
\frac{d \varphi^a}{ds} \,=\, - \delta^{ab} p_a \\
\frac{d p_a}{ds} \,=\, \delta_{ab} \eta^{jk} \partial_j \partial_k \varphi^b + m^2 \delta_{ab} \varphi^b 
\end{cases}
\ee
whose solutions are the integral curves of $\mathbb{\Gamma}$.
A direct computation shows that solutions of the previous equations are:
\be
\begin{split}
\varphi^a_{(s)} (\underline{x}) &\,=\, \int_{\mathbb{R}^3 \times \mathbb{R}^3} \dd^3 k \, \dd^3\underline{x}'\left( \varphi_\Sigma^a(\underline{x}') \cos{[\omega_k (s-x^0_\Sigma)]} - \delta^{ab} {p_b}^\Sigma(\underline{x}') \frac{\sin{[\omega_k (s- x^0_\Sigma)]}}{\omega_k} \right) e^{i \underline{k} \cdot (\underline{x}- \underline{x}')} \,, \\
{p_a}_s(\underline{x}) &\,=\, \int_{\mathbb{R}^3 \times \mathbb{R}^3} \dd^3 k \, \dd^3 \underline{x}' \, \left( \omega_k \varphi_\Sigma^a(\underline{x}') \sin{[\omega_k (s- x^0_\Sigma)]} + \delta^{ab} {p_b}^\Sigma \cos{[\omega_k (s-x^0_\Sigma)]} \right) e^{i \underline{k} \cdot (\underline{x}-\underline{x}')} \,,
\end{split}
\ee
where $\varphi^a_\Sigma$ and ${p_a}^\Sigma$ are the Cauchy data of the equations on the hypersurface $\Sigma$, $\underline{x}$, $\underline{x}'$ and $\underline{k}$ denote points in three copies of  $\mathbb{R}^3$, while $$\omega_k \,=\, \sqrt{|\underline{k}|^2 + m^2}$$ gives the dispersion relation corresponding to the equation. 
One has
\be
\begin{split}
f \,&=\, (\varpi^{-1} \circ \Psi)^\star F \\ \,&=\, \int_{\mathbb{R}^3 \times \mathbb{R}^3} \dd^3 k\,\dd^3 \underline{x}'\left( \varphi^{a_1}_\Sigma(\underline{x}') \cos{[\omega_k (x^0_1 - x^0_\Sigma)]} - \delta^{a_1 b} {p_b}^\Sigma(\underline{x}') \frac{\sin{[\omega_k (x^0_1 - x^0_\Sigma)]}}{\omega_k} \right) e^{i \underline{k} \cdot (\underline{x}_1 - \underline{x}')} \,.\end{split}
\ee 
The corresponding  Hamiltonian vector field with respect to the symplectic structure $\Omega_\infty^\Sigma$ has the following components
\be
\begin{split}
{\mathbb{X}_f}_{\varphi^a} (\underline{x}) &\,=\, - \frac{\delta f}{\delta p^\Sigma_a}(\underline{x}) \,=\,  \int_{\mathbb{R}^3}  \dd^3k \delta^{a a_1} \frac{\sin{[\omega_k(x^0_1 - x^0_\Sigma)]}}{\omega_k} \, e^{i \underline{k} \cdot (\underline{x}_1-\underline{x})}  \,,  \\
{{\mathbb{X}_f}_p}_a(\underline{x}) &\,=\,  \frac{\delta f}{\delta \varphi^a_\Sigma}(\underline{x}) \,=\,  \int_{\mathbb{R}^3} \dd^3k  \delta_{a a_1} \cos{[\omega_k (x^0_1 - x^0_\Sigma)]} \, e^{i \underline{k} \cdot (\underline{x}_1 - \underline{x})}  \,.
\end{split},
\ee
so the bracket between $f$ and $g$ with respect to $\os_\infty$ is the (constant) function of $\varphi_\Sigma$ and $p^\Sigma$:
\be \label{Eq: bracket cauchy data klein-gordon}
\{g,\, f\} \,=\, \mathbb{X}_f (g) \,=\, \int_{\mathbb{R}^3} \dd^3k \,\, \mathrm{cos}[ \omega_k(x^0_1 - x^0_2)] e^{i \underline{k} \cdot (\underline{x}_2 - \underline{x}_1) } \,.
\ee
\noindent 
One has that for the Hamiltonian vector field $\mathbb{X}_F$ with respect to $\Pi^\star_\Sigma \Omega^\Sigma$ it is $$\mathbb{X}_F\,=\,(\varpi \circ \Psi)_\star\mathbb{X}_f$$
and integral curves of $T\mathbb{\Gamma}$ are easily seen to be elements $(\varphi^a, p_a, \mathbb{X}_\varphi^a, {\mathbb{X}_p}_a)$ where $(\varphi^a, p_a)$ are solutions of the previous equations and $( \mathbb{X}_\varphi^a, {\mathbb{X}_p}_a)$ are solutions of the linearized problem
\be
\begin{cases}
\frac{d \mathbb{X}_\varphi^a}{ds} \,=\, - \delta^{ab} {\mathbb{X}_p}_a \\
\frac{d {\mathbb{X}_p}_a}{ds} \,=\, \delta_{ab} \eta^{jk} \partial_j \partial_k \mathbb{X}_\varphi^b + m^2 \delta_{ab} \mathbb{X}_\varphi^b, 
\end{cases} \,
\ee
which is identical to the original equations.
The solutions are then 
\be
\begin{split}
({\mathbb{X}_{\varphi^a}})_{(s)}(\underline{x}) &\,=\, \int_{\mathbb{R}^3}\dd^3k \, \delta^{aa_1} \frac{\sin{[\omega_k (x^0_1 - s)]}}{\omega_k} \, e^{i \underline{k} \cdot (\underline{x}- \underline{x}_1)}  \,, \\
({\mathbb{X}_{p_a}})_{(s)}(\underline{x}) &\,=\, \int_{\mathbb{R}^3} \dd^3k  \, \delta_{aa_1} \cos{[\omega_k (x^0_1 - s)]} \, e^{i \underline{k} \cdot (\underline{x} - \underline{x}_1)} \,.
\end{split}
\ee
Their immersion into $\mathbf{T}\mathcal{F}^\Sigma_{\mathcal{P}(\mathbb{E})}$ gives the curve
\be
\begin{split}
({\mathbb{X}_{\varphi^a}})_{(s)}(\underline{x}) &\,=\, \int_{\mathbb{R}^3} \dd^3k\,\, \delta^{aa_1} \frac{[\sin{\omega_k (x^0_1 - s)]}}{\omega_k} \, e^{i \underline{k} \cdot (\underline{x}- \underline{x}_1)} , \\
({{\mathbb{X}_{p_a}}})_{(s)}(\underline{x}) &\,=\, \int_{\mathbb{R}^3} \dd^3 k\,\, \delta_{aa_1} \cos{[\omega_k (x^0_1 - s)]} \, e^{i \underline{k} \cdot (\underline{x} - \underline{x}_1)}  \,,\\
({{\mathbb{X}_{\beta^{j}_a}}})_{(s)}(\underline{x}) &\,=\,  \int_{\mathbb{R}^3} \dd^3k\,\, \eta^{jl} \, i k_l \delta^{aa_1} \frac{\sin{[\omega_k (x^0_1 - s)]}}{\omega_k} \, e^{i \underline{k} \cdot (\underline{x}- \underline{x}_1)}  \,.
\end{split}
\ee
The action of $\varpi_\star$ is given by the identification of the parameter $s$ with the coordinate $x^0$ in the (global) chart chosen on $\mathbb{M}^{1, 3}$.
This gives the vector field $\mathbb{X}_F$ whose components read
\be
\begin{split}
({\mathbb{X}_F})_{\phi^a}(x^0, \underline{x}) &\,=\, \int_{\mathbb{R}^3} \dd^3k\,\, \delta^{aa_1} \frac{\sin{[\omega_k (x^0_1 - x^0)]}}{\omega_k} \, e^{i \underline{k} \cdot (\underline{x}- \underline{x}_1)}  \\
({{\mathbb{X}_F}})_{p_a}(x^0, \underline{x}) &\,=\, \int_{\mathbb{R}^3} \dd^3k \,\,\delta_{aa_1} \cos{[\omega_k (x^0_1 - x^0)]} \, e^{i \underline{k} \cdot (\underline{x} - \underline{x}_1)}  \,, \\
({{\mathbb{X}_F}})_{\beta^j_a}(x^0, \underline{x}) &\,=\,  \int_{\mathbb{R}^3} \dd^3k \,\,\eta^{jl} \, i k_l \delta^{aa_1} \frac{\sin{[\omega_k (x^0_1 - x^0)]}}{\omega_k} \, e^{i \underline{k} \cdot (\underline{x}- \underline{x}_1)}  \,.
\end{split}
\ee
Recalling that $G \,=\, P^0_a(x^0_2, x^j_2)$,  the bracket between $F$ and $G$ is computed to be
\be
\begin{split}
\{G, F\}_\chi \,&=\, \mathbb{X}_F(G) \\ &=\, \int_{\mathbb{M}^{1, 3}} \dd x^0\dd^3\underline{x}\,\,{{\mathbb{X}_F}_p}_a(x^0, \underline{x}) \underbrace{\frac{\delta G}{\delta P^0_a}}_{\delta(x^0 - x^0_2) \, \delta(\underline{x}- \underline{x}_2)}  \\ &=\, \int_{\mathbb{R}^3} \dd^3k\,\, \cos{\omega_k (x^0_1 - x^0_2)} e^{i \underline{k} \cdot (\underline{x}_2 - \underline{x}_1)}  \,. \end{split}
	\ee

\noindent Notice that, as we expect from the general theory we described, this result coincides with the pull-back to $\elag_\m$ via $\ps$ of the function $\{g,f\}$ in \eqref{Eq: bracket cauchy data klein-gordon}: being a  constant function on the space of Cauchy data, they  share the same expression.



\subsection{Gauge theories: Electrodynamics}
\label{Subsec: Gauge theories: Electrodynamics}

Consider the example \ref{Ex: electrodynamics variational principle} where Maxwell equations have been studied  within a multisymplectic Hamiltonian formalism. Let 
$\left( \mathbb{M}^{1,3}, \eta \right)$ denote the Minkowski spacetime and let $\Sigma$ be the surface $x^{0}= x^0_\Sigma$, so that $\mathbb{M}^{1,3}\simeq \Sigma \times \mathbb{R}$. 
A (global) chart on $\Sigma$ is denoted by $(U_{\Sigma}, \psi_{U_\Sigma})$, $\psi_{U_\Sigma}(\mathfrak{x}) \,=\, \underline{x}$, with $\mathfrak{x} \in U_\Sigma \subset \Sigma$, whereas a chart on $\mathbb{R}$ is denoted by $(U_{\mathbb{R}}, \psi_{U_\mathbb{R}})$, $\psi_{U_\mathbb{R}}(r) \,=\, x^0$, with $r \in U_{\mathbb{R}} \subset \mathbb{R}$.
The space $\widetilde{\elag}^\epsilon$ is made by curves $\gamma \,\colon\, \mathbb{R}\,\rightarrow \, \fpe^\Sigma$. 
The components of these curves are denoted as follows
\be
\begin{split}
& {a_0}_{(s)}(\psi_{U_\Sigma}^{-1}(\underline{x})) = A_0(\psi_{U_{\mathbb{R}}}^{-1}(x^0)=s, \psi_{U_\Sigma}^{-1}(\underline{x})),\\ 
& {a_j}_{(s)}(\psi_{U_\Sigma}^{-1}(\underline{x}))=A_j(\psi_{U_{\mathbb{R}}}^{-1}(x^0)=s,\psi_{U_\Sigma}^{-1}(\underline{x})), \\
& {p^k}_{(s)}(\psi_{U_\Sigma}^{-1}(\underline{x}))=P^{k0}(\psi_{U_{\mathbb{R}}}^{-1}(x^0)=s, \psi_{U_\Sigma}^{-1}(\underline{x})),\\ 
&\beta^{jk}_{(s)}(\psi_{U_\Sigma}^{-1}(\underline{x}))=P^{jk}(\psi_{U_{\mathbb{R}}}^{-1}(x^0)=s,\psi_{U_\Sigma}^{-1}(\underline{x}))\,,
\end{split}
\ee
with  $j,k=1,2,3$. 
In analogy to the previous analysis, the space $\fpe^\Sigma$ reads
\be
\fpe^\Sigma \,=\, \prod_\mu {\mathcal{H}^{\frac{3}{2}}(\Sigma, vol_\Sigma)}_\mu \times \prod_{\mu, \nu} {\mathcal{H}^{\frac{1}{2}}(\Sigma, vol_\Sigma)}^{\mu \nu} \,.
\ee
\noindent The Hamiltonian introduced in example \ref{Ex: electrodynamics variational principle} gives rise to the following Lagrangian on $\mathbf{T}( \fpe^\Sigma )$ which, evaluated along (the tangent lift of) a curve $\gamma_{(s)}$ (in the next formulas we avoid explicitly denoting the dependence on the parameter $s$ for the sake of notational simplicity) gives 
\begin{equation}
\mathscr{L}\bigr|_{t \gamma} = \int_{\Sigma}vol_\Sigma \left[ p^k\left( \dot{a}_{k} - \partial_k a_{0}\right) - \frac{1}{2} \beta^{kj} \left( \partial_k a_{j} - \partial_j a_{k} \right) - \frac{1}{4} \delta_{jl}\delta_{km}\beta^{jk}\beta^{lm} - \frac{1}{2} \delta_{jk} p^k p^{j} \right] \,,
\end{equation}
from which one has the  Hamiltonian functional
\begin{equation}
\mathcal{H}(\gamma) = \int_{\Sigma}vol_\Sigma \left[ p^k \partial_k a_{0} + \frac{1}{2} \beta^{kj} \left( \partial_k a_{j} - \partial_j a_{k} \right) + \frac{1}{4} \delta_{jl}\delta_{km}\beta^{jk}\beta^{lm} + \frac{1}{2} \delta_{jk} p^k p^{j} \right] .
\end{equation}
The 2-form $\Omega^\Sigma$ on  $\fpe^\Sigma$
\be
\Omega^\Sigma (\mathbb{X}, \mathbb{Y}) \,=\, \int_\Sigma vol_\Sigma \left(\, {\mathbb{X}_{a_k}} \mathbb{Y}_{p^k} - \mathbb{X}_{p^k} {\mathbb{Y}_{a_k}} \,\right) \,
\ee
is clearly pre-symplectic, its kernel is spanned,  at each point $\chi_\Sigma$, by tangent vectors having only components $\mathbb{X}_{a_0}$ or $\mathbb{X}_{\beta^{jk}}$.
The first step of the pre-symplectic algorithm gives the manifold
\be
\mathfrak{i}_1(\mathcal{F}_1) \,=\, \left\{\, \chi_\Sigma \in \fpe^\Sigma \;\; :\;\; i_{\mathbb{X}_{\chi_\Sigma}} \dd \mathcal{H} \,=\, 0 \,\,\, \forall \,\, \mathbb{X}_{\chi_\Sigma} \in \mathbf{T}_{\chi_\Sigma} \mathcal{F}_1^\perp \,\right\}\,,
\ee
where $\mathbf{T}_{\chi_\Sigma} \mathcal{F}_1^\perp\,=\,\mathrm{ker}\,\os$ at $\chi_\Sigma$.
The latter condition gives the constraints
\be \label{Eq:constraints electrodynamics}
\begin{split} &\de_k p^k \,=\, 0 \,, \\ & \beta^{jk} \,=\, -2 \delta^{jl} \delta^{km} F_{lm} \,. \end{split}
\ee
The second constraint can be eliminated by expressing the $F_{lm}$ components in terms of the $a_k$'s (notice that expressing the $\beta^{js}$ elements in terms of the $a^k$'s amounts to describe the magnetic field in terms of the vector potential), while the first line constraint is non-holonomic  (it represents the Gauss' law) and  can not be eliminated.
Therefore, the manifold $\mathcal{F}_1$ results in 
\be 
\mathcal{F}_1 \,=\, \left\{\, (a_k, p^k) \in \fpe^\Sigma \;\;\; : \;\;\; \partial_j p^j \,=\, 0 \,\right\}
\ee
whose immersion into $\fpe^\Sigma$ is given by the second of \eqref{Eq:constraints electrodynamics} and upon fixing any arbitrary $a_0$.
The manifold $\mathcal{F}_1$ is seen to be also the final manifold of the pre-symplectic constraint algorithm, that is  $$\mathcal{F}_\infty \,=\, \mathcal{F}_1.$$
Pulling-back  $\Omega^\Sigma$ and $\mathcal{H}$ to the final manifold under $\mathfrak{i}_1 \,=\, \mathfrak{i}_\infty\,:\,\mathcal{F}_\infty\,\hookrightarrow\,\fpe^\Sigma$ reads
\be
\begin{split}
\mathfrak{i}_1^\star \Omega^\Sigma(\mathbb{X}, \, \mathbb{Y}) \,=\, \Omega^\Sigma_1(\mathbb{X},\, \mathbb{Y}) \,&=\, \int_\Sigma vol_\Sigma \left(\, {\mathbb{X}_{a_k}} \mathbb{Y}_{p^k} - \mathbb{X}_{p^k} {\mathbb{Y}_{a_k}} \,\right)  \,,\\
\mathfrak{i}_1^\star \mathcal{H}({\gamma_1}) \,=\, \mathcal{H}_1({\gamma_1}) \,&=\,  \int_{\Sigma} vol_\Sigma \left[ p^k \partial_k a_{0} + \frac{1}{4}\delta^{jl}\delta^{km}{F_{jk}} {F_{lm}} + \frac{1}{2} \delta_{jk} p^k p^{j} \right]  \\
\,&=\, \int_\Sigma vol_\Sigma \left[\, \frac{1}{4}\delta^{jl}\delta^{km}{F_{jk}} {F_{lm}} + \frac{1}{2} \delta_{jk} p^k p^{j} \,\right]  \,,
\end{split}
\ee
where $\gamma_1$ is a curve (whose dependence on $s$ we have again preferred not to write explicitly) on $\mathcal{F}_1$ and the latter equality is the result of an integration by parts.
The 2-form $\Omega^\Sigma_1$ is again pre-symplectic.
Indeed, tangent vectors with a ${\mathbb{X}_{a_k}}$ given by  ${\mathbb{X}_{a_k}} \,=\, \partial_k \zeta$ (for a function $\zeta$) lie in the kernel of $\Omega^\Sigma_1$ (at each point $(a, p)$) because of the Gauss' constraint.
As in the previous examples, the immersion map is easily seen to be differentiable with respect to the differential structure of $\mathcal{F}_\infty$.
The canonical equations on  $\mathcal{F}_\infty$ can be written as
\be \label{Eq: equations of the motion Electrodynamics on M1}
\begin{split}
p^k \,&=\, \delta^{kj}  \frac{d a_j}{ds} \,,  \\  \frac{d p^k}{ds} \,&=\, 	\delta^{km} \delta^{jl}\partial_j \partial_l a_m - \delta^{kj}\delta^{lm}  \partial_j \partial_l a_m  \,.  
\end{split}
\ee	
From  the initial choice $$\fpe \,=\,  \prod_{\mu=0,...,3} {\mathcal{H}^2(\mathbb{M}^{1,3}, vol_{\mathbb{M}^{1,3}})}_\mu \times \prod_{\mu,\nu = 0,...,3} {\mathcal{H}^1(\mm, vol_{\mm})}^{\mu \nu}, $$ follows that  $\mathcal{F}_\infty$ turns to be
\be
\mathcal{M}_\infty \,=\, \prod_{k}{\mathcal{H}^{\frac{3}{2}}(\Sigma,\, vol_\Sigma)}_k \times \prod_{k}{\mathcal{H}^{\frac{1}{2}}(\mathrm{div}0;\, \Sigma, vol_\Sigma) }^{k},
\ee
where $\mathcal{H}^{\frac{1}{2}}(\mathrm{div}0; \Sigma, vol_\Sigma)$ denotes the space of divergenceless $\mathcal{H}^{\frac{1}{2}}$-functions on $\Sigma$ (notice that the operator $\mathrm{div}$ denotes  the adjoint to the operator   $\mathrm{grad}\,:\,\mathcal{H}^{\frac{3}{2}}(\Sigma, vol_\Sigma)\,\to\,\mathcal{H}^{\frac{1}{2}}(\Sigma, vol_\Sigma)$). 
The latter space turns out to be a Hilbert space, since one proves that it is  a closed subspace of the space of square integrable functions. 
\begin{proposition}
The gradient operator $\mathrm{grad}\,:\,\mathcal{H}^{\frac{3}{2}}(\Sigma, vol_\Sigma)\,\to\,\mathcal{H}^{\frac{1}{2}}(\Sigma, vol_\Sigma)$ is closed.
\begin{proof}
Consider a sequence $\{f_n\}_{n \in \mathbb{N}}$ with $f_n\in\mathcal{H}^{\frac{3}{2}}(\Sigma, vol_\Sigma)$ such  that  
$$
\lim_{n\to\infty}\,f_n\,=\,f
$$
with $f\,\in\, \mathcal{H}^{\frac{3}{2}}(\Sigma, vol_\Sigma)$.
The following inequalities
\be
\begin{split}
\sum_j \Vert \de_j(f_n - f) \Vert_{\mathcal{H}^{\frac{1}{2}}} \,&=\, \sum_{j, k} \int_{\mathbb{R}^3} \dd^3k\, \vert \tilde{f}_n - \tilde{f} \vert^2 k_j \, k_k \, \vert k \vert \,  \\ &\leq \, 9 \int_{\mathbb{R}^3} \dd^3k\, \vert \tilde{f}_n - \tilde{f} \vert^2 \, \vert k \vert^3  \, =\, 9 \Vert f_n - f \Vert_{\mathcal{H}^{\frac{3}{2}}} \end{split}\,
\ee
($\tilde{f}$ denoting the Fourier transform of $f$) hold, thus proving that the operator  $\mathrm{grad}$ maps closed sets into closed sets.
\end{proof}
\end{proposition}
\noindent The  set $\mathcal{H}^{\frac{1}{2}}(\mathrm{div}0, \Sigma, vol_\Sigma)$ (which is the kernel of the $\mathrm{div}$ operator) is therefore  -- recall  the closed range theorem --  a closed subspace in $\mathcal{H}^{\frac{1}{2}}(\Sigma, vol_\Sigma)$, whose complement is the range of the action of the $\mathrm{grad}$ operator into  $\mathcal{H}^{\frac{1}{2}}(\Sigma, vol_\Sigma)$, that is the following decomposition 
\be
\prod_k {\mathcal{H}^{\frac{1}{2}}(\Sigma, vol_\Sigma)}^k \,=\, \prod_k {\mathcal{H}^{\frac{1}{2}}(\mathrm{div}0, \Sigma, vol_\Sigma)}^k \oplus \mathrm{grad}\mathcal{H}^{\frac{3}{2}}(\Sigma, vol_\Sigma) \,.
\ee
into closed subspaces -- and, thus, \textit{Hilbert} spaces -- exists. 
Notice that a similar proof shows that the following decomposition
\be
\mathcal{H}^{\frac{3}{2}}(\Sigma,\,vol_\Sigma) \,=\, \mathcal{H}^{\frac{3}{2}}(\mathrm{div}0,\, \Sigma,\,vol_\Sigma) \oplus \mathrm{grad}\mathcal{H}^{\frac{5}{2}}(\Sigma,\, vol_\Sigma) \,
\ee
holds, where the first term is the space of $\mathcal{H}^{\frac{3}{2}}$ functions with zero divergence and the second term is the range under the action of  the gradient operator upon  the set of $\mathcal{H}^{\frac{5}{2}}$ functions.
According to this splitting, the space  $\mathcal{F}_\infty$ can be written as 
\be
\mathcal{F}_\infty \,=\, \prod_k{\mathcal{H}^{\frac{3}{2}}(\mathrm{div}0,\, \Sigma,\, vol_\Sigma)}_k \times \mathrm{grad}\mathcal{H}^{\frac{5}{2}}(\Sigma,\, vol_\Sigma) \times \prod_k {\mathcal{H}^{\frac{1}{2}}(\mathrm{div}0,\, \Sigma,\, vol_\Sigma)}^k \,,
\ee
where all the summands are closed (and, thus, Hilbert) subspaces of $\mathcal{H}^{\frac{3}{2}}(\Sigma, vol_\Sigma)$ and $\mathcal{H}^{\frac{1}{2}}(\Sigma, vol_\Sigma)$ respectively.
Upon denoting as $m_\infty \,=\, (\tilde{a}_k,\, \de_k \phi,\, p^k)$ an element in $\mathcal{F}_\infty$, one sees that  the tangent space to $\mathcal{M}_\infty$ decomposes as
\be
\begin{split}
\mathbf{T}_{m_\infty}\mathcal{F}_\infty \,&=\, \mathbf{T}_{\tilde{a}} \prod_k{\mathcal{H}^{\frac{3}{2}}(\mathrm{div}0,\, \Sigma, vol_\Sigma)}_k \oplus \mathbf{T}_{\de \phi} \mathrm{grad}\mathcal{H}^{\frac{5}{2}}(\Sigma,\, vol_\Sigma) \oplus \mathbf{T}_{p} \prod_k {\mathcal{H}^{\frac{1}{2}}(\mathrm{div}0,\, \Sigma,\, vol_\Sigma)}^k  \\
\,&=\,  \prod_k {\mathcal{H}^{\frac{3}{2}}(\mathrm{div}0,\, \Sigma,\, vol_\Sigma)}^k \oplus \mathrm{grad}\mathcal{H}^{\frac{5}{2}}(\Sigma,\, vol_\Sigma) \oplus \prod_k{\mathcal{H}^{\frac{1}{2}}(\mathrm{div}0,\, \Sigma,\, vol_\Sigma)}^k \,, 
\end{split}
\ee
since all the summand spaces  are Hilbert spaces and then isomorphic to their tangent spaces.
Such a  decomposition of $\mathcal{F}_\infty$ and of its tangent space $\mathbf{T}_{m_\infty}\mathcal{F}$ is particularly suitable for our purposes because for the  second of its terms one has 
$$
\mathbf{T}_{\de \phi}\mathrm{grad}\mathcal{H}^{\frac{5}{2}} \,=\, \mathrm{ker}\,{\os_\infty}_{m_\infty}
$$
that is it gives the generators of the so called  gauge transformations. Moreover, notice that writing down such a decomposition amounts to fix a complementary space in  $\mathbf{T}_{m_\infty}\mathcal{F}_\infty$ to  $\mathrm{ker}\, \os_\infty$. 
In particular, the previous splitting amounts to fix a particular connection -- called the \textit{Coulomb connection} --  on the bundle $\mathcal{F}_\infty \to \mathcal{F}_\infty / \mathrm{ker}\, \os_\infty$ (see  \cite{Carinena-Ibort1985-Canonical_ghosts}). This connection is described by the idempotent $\mathcal{P} \;\; : \;\; \mathfrak{X}(\mathcal{F}_\infty) \to \mathfrak{X}^H(\mathcal{F}_\infty)$ whose action is 
\be 
\mathcal{P}(\mathbb{X}) \,=\, \tilde{\mathbb{X}} \,,
\ee
with  $\tilde{\mathbb{X}}$ having components
\be
\begin{split}
\tilde{\mathbb{X}}_{\tilde{a}_k} \,&=\, 0 \,, \\ 
\tilde{\mathbb{X}}_{{\de \phi}_k} \,&=\, \partial_k \int_{\Sigma} vol_\Sigma^{\underline{y}}\,\, G_{\Delta}(\underline{y}, \underline{x}) \delta^{jl}\partial_l {\mathbb{X}_{\de \phi}}_j(\underline{y})  \,, \\
\tilde{\mathbb{X}}_{p^k} \,&=\, 0 \,,
\end{split}
\ee
in terms of the Green's function $G_{\Delta}(\underline{y}, \underline{x})$ of the Laplacian on $\Sigma$. 
This connection is the one we adopt in order  to compute Hamiltonian vector fields in this pre-symplectic case, following the general theory presented in the previous sections.
	
\noindent Consider an element  $f\,\in\,\mathscr{F}(\mathcal{F}_\infty)$.
Its Hamiltonian vector field with respect to $\Omega_\infty^\Sigma$ is the one satisfying the conditions
\be
\begin{cases}
\Omega^\Sigma_\infty (\mathbb{X}_f,\, \cdot \,) \,=\, \dd f (\,\cdot \,) \\
\mathcal{P}(\mathbb{X}_f)\,=\, 0. 
\end{cases}
\ee
From the first condition we see that the vector field $\mathbb{X}_f$ has the components
\be
\begin{split} ({\mathbb{X}_f})_{\tilde{a}_k} + (\mathbb{X}_f)_{\de \phi_k} \,&=\, \frac{\delta f}{\delta p^k} , \\  
(\mathbb{X}_f)_{p^k} \,&=\, - \frac{\delta f}{\delta a_k} \,; \end{split} 
\ee
it is also  readily proven that the second equation  fixes $(\mathbb{X}_f)_{\de \phi_k}$ to vanish.
This fixes the bracket between  $g,f\,\in\,\mathscr{F}(\mathcal{F}_\infty)$ to be
\be
\{\, g, \, f \,\} \,=\, \os_\infty(\mathbb{X}_g,\, \mathbb{X}_f) \,=\, \int_\Sigma vol_\Sigma \left(\,\frac{\delta f}{\delta p^k} \frac{\delta g}{\delta \tilde{a}_k} - \frac{\delta f}{\delta \tilde{a}_k} \frac{\delta g}{\delta p^k} \,\right)  \,. 
\ee
	
\noindent Let us consider two elements $G,F \,\in\,\mathscr{F}(\elag_{\mathbb{M}^{1,3}})$. 
Following the procedure described in this section, the Hamiltonian vector field $\mathbb{X}_F$ associated to $F$ with respect to the canonical structure $\Pi_\Sigma^\star \Omega^\Sigma$ comes upon linearizing the flow generated by the dynamics   $(\Omega^\Sigma, \mathcal{H})$ (with the gauge condition given by the chosen connection) with Cauchy datum given by the Hamiltonian vector field $\mathbb{X}_f$ associated to  $f \,=\, (\varpi^{-1} \circ \Psi)^\star F$ with respect to $\Omega^\Sigma_\infty$ and the connection $\mathcal{P}$.
Then the bracket between $G$ and $F$ is given by the action of the vector field $\mathbb{X}_F$ upon  $G$.
\noindent As a specific example, we consider the following two functions
\be
\begin{split}
F_{k_1} \,=\, A_{k_1}(x^0_1,\, \underline{x}_1) \,=\, \int_{\mathbb{M}^{1,3}} vol_{\mathbb{M}^{1,3}}\,  A_{k_1}(x^0,\, \underline{x}) \delta(x^0-x^0_1) \delta(\underline{x} - \underline{x}_1)  \,, \\
G_{k_2} \,=\, A_{k_2}(x^0_2,\, \underline{x}_2) \,=\, \int_{\mathbb{M}^{1,3}} vol_{\mathbb{M}^{1,3}}\,  A_{k_2}(x^0,\, \underline{x}) \delta(x^0-x^0_2) \delta(\underline{x} - \underline{x}_2)  \,,
\end{split}
\ee
namely those providing the value of (the $k_j$-component of) the configuration field, solution of the equations of the motion, at the point $\psi_{U_{\mathbb{M}^{1,3}}}^{-1}(x^0_j, \underline{x}_j)$.
In order to follow the path outlined above, the first step is to have $f=(\varpi^{-1} \circ \Psi)^\star F$  on the space of Cauchy data. 
Such a space of the  Cauchy data is a single  leaf of the foliation induced by $\mathcal{P}$ on the manifold $\mathcal{F}_\infty$ given at the end of the pre-symplectic constraint algorithm, since any element in $\mathcal{F}_\infty$ belongs to one single leaf $\sigma_{\mathcal{P}}$ of $\mathcal{P}$ and for any point in $\sigma_{\mathcal{P}}$ there exists a unique solution to the relation
\be
i_{\mathbb{\Gamma}}\Omega^\Sigma_\infty \,=\, \dd \mathcal{H}   \,,
\ee
satisfying the
\be
\mathcal{P}(\mathbb{\Gamma})=0 \,
\ee
condition, which ensures the flow of $\mathbb{\Gamma}$ to entirely lie on the same leaf $\sigma_{\mathcal{P}}$.
In order to compute $f$ we need to write $A_{k_1}$ as the image under $\varpi$ of a solution of the equations of the motion for $a_{k_1}$ and explicitly show the dependence on the initial conditions.
Recalling  \eqref{Eq: equations of the motion Electrodynamics on M1}, the element $a_{k}$ is the solution of (recall also that we do not write the eplicit dependence on the parameter $s$)
\be
\frac{d^2 a_{k}}{ds^2} \,=\, \delta^{jl}\partial_j \partial_l a_{k} - \delta^{jl} \partial_k \partial_j a_{l}\,, 
\ee
which is unique (for any fixed initial condition $a^\Sigma_k(\underline{x})$, $p^k_\Sigma(\underline{x})$) after imposing the condition $	\mathcal{P}(\mathbb{\Gamma})=0$.
Indeed, the latter condition imposes $\delta^{jk}\partial_j a_k \,=\,0$, that is, the so called \textit{Coulomb's gauge}, since it is a projector onto the subspace $\prod_k{\mathcal{H}^{\frac{3}{2}}(\mathrm{div}0, \Sigma, vol_\Sigma)}_k \times \prod_k{\mathcal{H}^{\frac{1}{2}}(\mathrm{div}0, \Sigma, vol_\Sigma)}^k$. 
We are thus led to the equation
\be \label{Eq: free wave equation a_k}
\frac{d^2 a_{k}}{ds^2} - \delta^{jl}\partial_j \partial_l a_{k} \,=\,0 \,,
\ee
which is clearly the homogeneous wave equation for $a_{k}$. Since an existence and uniqueness theorem for the solutions of such equation  exists within the considered space of fields, we write 
\be \label{Eq: solution wave equation}
\begin{split}
a_{k, s}(\underline{x}) \,=\, \frac{1}{4 \pi} \biggl[\, &\int_{\Sigma}vol_\Sigma^{\underline{y}}\, \left(\, a^\Sigma_{k}(\underline{y}) + |s - x^0_\Sigma| \delta_{kj}p^j_\Sigma(\underline{y}) \,\right) \tilde{G}_{\underline{x}, \, |s-x^0_\Sigma|}(\underline{y})  \,  \\ 
&+ \int_{\Sigma}vol_\Sigma^{\underline{y}}\, |s - x^0_\Sigma| a^\Sigma_k(\underline{y}) \frac{\partial}{\partial s}\tilde{G}_{\underline{x}, (s- x^0_\Sigma)}(\underline{y}) (\Theta(s-x^0_\Sigma)-\Theta(x^0_\Sigma-s))   \biggr] \,,
\end{split}
\ee
where  $\tilde{G}_{\underline{x}, a}(\underline{y})$ is  the characteristic function of the surface of the sphere centered at  $\underline{x}$ with radius $a$ and $\Theta$ is the Heaviside function. 
Give such a solution \eqref{Eq: solution wave equation}, the pull-back of $F$ to the space of Cauchy data is readily computed to be
\be
\begin{split}
f \,=\, \frac{1}{4 \pi} \biggl[\, &\int_{\Sigma} vol_\Sigma^{\underline{y}}\,  \left(\, a^\Sigma_{k_1}(\underline{y}) + |x^0_1 - x^0_\Sigma| \delta_{kj}p^j_\Sigma(\underline{y}) \,\right) \tilde{G}_{\underline{x}_1, \, |x^0_1-x^0_\Sigma|}(\underline{y})  \,  \\ 
 & + \int_{\Sigma}vol_\Sigma^{\underline{y}}\, \,|x^0_1 - x^0_\Sigma| a^\Sigma_{k_1}(\underline{y}) \frac{\partial}{\partial s}\tilde{G}_{\underline{x}_1, (s- x^0_\Sigma)}(\underline{y})\bigr|_{s=x^0_1} (\Theta(x^0_1-x^0_\Sigma)-\Theta(x^0_\Sigma-x^0_1))  \biggr] \,.
\end{split} 
\ee
Taking into account this expression for $f$ and the analogous for $g$, a direct computation gives the bracket between the functions $g$ and $f$ on the space of Cauchy data
\be \label{Eq: bracket cauchy data electrodynamics}
\begin{split}
\{g,\, f\}_{(a, p)} \,&=\, \int_{\Sigma}  \left(\, \frac{\delta f}{\delta p^k}\frac{\delta g}{\delta \tilde{a}_k} - \frac{\delta f}{\delta \tilde{a}_k}\frac{\delta g}{\delta p^k} \,\right)  \\
&=\, \frac{\delta_{k_1 k_2}}{16 \pi^2} \int_{\Sigma} vol_\Sigma^{\underline{y}}\, \biggl\{ \left(|x^0_2 - x^0_\Sigma| - |x^0_1 - x^0_\Sigma| \right) \tilde{G}_{\underline{x}_1, |x^0_1 - x^0_\Sigma|}(\underline{y})\tilde{G}_{\underline{x}_2, |x^0_2 - x^0_\Sigma|}(\underline{y})  \\  
&\qquad + |x^0_1 - x^0_\Sigma||x^0_2 - x^0_\Sigma| \biggl[\, \frac{\partial}{\partial s}\tilde{G}_{\underline{x}_1, (s-x^0_\Sigma)}(\underline{y})\bigr|_{s=x^0_1}\left(\Theta(x^0_1 - x^0_\Sigma) - \Theta(x^0_\Sigma - x^0_1)\right) \tilde{G}_{\underline{x}_2, |x^0_2-x^0_\Sigma|}(\underline{y})  \\ 
&\qquad\qquad - \frac{\partial}{\partial s}\tilde{G}_{\underline{x}_2, (s-x^0_\Sigma)}(\underline{y})\bigr|_{s=x^0_2}\left(\Theta(x^0_2 - x^0_\Sigma) - \Theta(x^0_\Sigma - x^0_2)\right) \tilde{G}_{\underline{x}_1, |x^0_1-x^0_\Sigma|}(\underline{y})  \,\biggr] \,\biggr\} \,.
\end{split}
\ee
In order to evaluate the action of the bracket upon $F$ and $G$ we need to compute the Hamiltonian vector field $\mathbb{X}_f$ associated to $f$ with respect to $\Omega^\Sigma_\infty$ and $\mathcal{P}$  and then to reconstruct the Hamiltonian vector field $\mathbb{X}_F$ upon  linearizing the flow generated by the  equations of the motions. 
Since in this case the equation  \eqref{Eq: free wave equation a_k} is linear, it results that $\mathbb{X}_F$ is the vector field with components
\be
\begin{split}
({\mathbb{X}}_F)_{a_k}(\underline{x}, x^0) &\,=\, \frac{\delta_{k_1 k}}{16 \pi^2} \int_\Sigma vol_\Sigma^{\underline{y}}\, \biggl\{\, (|x^0 - x^0_\Sigma| - |x^0_1 - x^0_\Sigma| ) \tilde{G}_{\underline{x}_1 , |x^0_1 - x^0_\Sigma|}(\underline{y}) \tilde{G}_{\underline{x}_1 , |x^0 - x^0_\Sigma|}(\underline{y})  \\ 
& \qquad + |x^0_1 - x^0_\Sigma||x^0 - x^0_\Sigma| \biggl[\, \frac{\partial}{\partial s}\tilde{G}_{\underline{x}_1, (s-x^0_\Sigma)}(\underline{y})\bigr|_{s=x^0_1}\left(\Theta(x^0_1 - x^0_\Sigma) - \Theta(x^0_\Sigma - x^0_1)\right) \tilde{G}_{\underline{x}, |x^0-x^0_\Sigma|}(\underline{y})  \\
&\qquad\qquad - \frac{\partial}{\partial x^0}\tilde{G}_{\underline{x}, (x^0-x^0_\Sigma)}(\underline{y})\left(\Theta(x^0 - x^0_\Sigma) - \Theta(x^0_\Sigma - x^0) \right) \tilde{G}_{\underline{x}_1, |x^0_1-x^0_\Sigma|}(\underline{y}) \biggr] \,\biggr\} \,,
\end{split}
\ee
\be
(\mathbb{X}_F)_{p^k}(\underline{x}, x^0) \,=\, \delta^{kj} \frac{\partial}{\partial x^0} ({\mathbb{X}_F})_{a_j}  \,,
\ee
and
\be
(\mathbb{X}_F)_{\beta^{jk}} \,=\, \delta^{kl}\delta^{jm}(\partial_l (\mathbb{X}_F)_{a_m} - \partial_m (\mathbb{X}_F)_{a_l}) \,.
\ee
Such analysis gives
\be
\{G,\, F\}_\chi \,=\, \int_{\mathbb{M}^{1, 3}} \dd x^0 vol_\Sigma^{\underline{x}} \, \biggl[(\mathbb{X}_F)_{a_k} \underbrace{\frac{\delta G}{\delta A_k}(\underline{x}, x^0)}_{= \delta(x^0 - x^0_2) \delta(\underline{x}-\underline{x}_2)} + (\mathbb{X}_F)_{p^k} \underbrace{\frac{\delta G}{\delta p^k}}_{=0} \biggr] 
\ee
which is coherent with \eqref{Eq: bracket cauchy data electrodynamics}.

\section*{Conclusions}

\noindent To conclude, we showed how the solution space of a class of first order Hamiltonian field theories can be equipped with a Poisson manifold structure.
In particular we saw that such a Poisson manifold structure comes from a pre-symplectic structure that naturally emerges from the Schwinger-Weiss variational principle formulated within the multisymplectic formulation of first order Hamiltonian field theories.
\noindent 
The structure emerging from the variational principle is a 2-form $\Pi_{\Sigma}^\star \Omega^\Sigma$,  where $\Pi_{\Sigma}$ is the restriction of the fields to a suitable hypersurface $\Sigma$ in the domain where the fields are defined.
The 2-form $\Omega^\Sigma$ on $\fpe^\Sigma$ from which it comes  is related to  a 2-form $\os_\infty$ on the final manifold $\mathcal{F}_\infty$ of the pre-symplectic constraint algorithm, which allows to study constrained theories (using Dirac's words).
It comes out that $\os_\infty \,=\, \mathfrak{i}_\infty^\star \os$,  where $\mathfrak{i}_\infty$ is the immersion of the final manifold into $\fpe^\Sigma$.
We proved that, in theories without gauge symmetries, $\os_\infty$ is symplectic and there exists a diffeomorphism $\Psi$  associating to  each Cauchy datum in $\mathcal{F}_\infty$ a unique solution of the equations of motion.
In particular $\Psi$ is such that $\mathfrak{i}_\infty \circ \Psi^{-1} \,=\, \Pi_\Sigma$.
This allowed to prove that  $\pssos={\Psi^{-1}}^\star \os_\infty$ on $\elag_\m$. 
Being $\Omega^\Sigma_\infty$ symplectic and $\Psi$ a diffeomorphism, $\Pi_\Sigma^\star \Omega^\Sigma$ turns out to be symplectic as well and determines a Poisson bracket.

Within abelian gauge theories we slightly modified the construction.
Indeed, the 2-form $\Omega^\Sigma_\infty$ on the final submanifold $\mathcal{F}_\infty$ is still pre-symplectic. Nonetheless it has been possible to prove  that it is symplectic \textit{along}  the leaves of a flat connection on the bundle $\mathcal{F}_\infty / \mathrm{ker} \, \Omega^\Sigma_\infty$ that is canonically defined for any gauge theory\footnote{It is indeed flat only for those gauge theories which do not exhibit the so-called Gribov anomalies.}, in the sense that it is non-degenerate if contracted only along tangent vectors being horizontal with respect to the connection chosen.
\noindent From the physical point of view the choice of such a connection amounts to a gauge fixing which, in the case of the connection we chose in electrodynamics, is the Coulomb gauge.
We saw that the restriction of $\Omega^\Sigma_\infty$ to the  leaves of such a connection is symplectic and then (along with the same construction used above)  so is the 2-form $\Pi_\Sigma^\star \Omega^\Sigma \bigr|_{\sigma_{\mathcal{P}}}$ ($\sigma_{\mathcal{P}}$ representing a leaf of the foliation).
The case where a flat connection can not be chosen because of intrinsic obstructions of the theory is exemplified by non-Abelian gauge theories.
We will deal with them in the second part of this series 
where we will construct a Poisson bracket  within the framework developed in the present paper in terms of a suitable coisotropic embedding.
Such a construction will show that the emergence of additional degrees of freedom, interpreted as \textit{ghost fields}, seems to be necessary in order for properly defining a Poisson structure.

\appendix

\section{The pre-symplectic constraint algorithm} \label{App: PCA}

\noindent For the sake of completeness, we report in this appendix how the pre-symplectic constraint algorithm used in Sect. \ref{Sec: The symplectic case vs gauge theories} and \ref{Sec: Poisson brackets on the solution space} is defined.
We refer to \cite{Gotay-Nester-Hinds1978-DiracBergmann_constraints} for a more extensive discussion.

\noindent Let us consider a pre-symplectic Hamiltonian system, that is a triple $(\mathcal{M},\, \omega,\, \mathcal{H})$ where $\mathcal{M}$ is a Banach manifold, $\omega$ a closed $2$-form on it and $\mathcal{H}$ a real-valued function on $\mathcal{M}$.
Denoting by $\flat$ the map between $\mathbf{T}\mathcal{M}$ and $\mathbf{T}^\star \mathcal{M}$ induced by $\omega$,
\be
\flat \;\; : \;\; \mathbf{T}\mathcal{M} \to \mathbf{T}^\star \mathcal{M} \;\; : \;\; (m, \, X_m) \mapsto \left(m,\, \omega_m(X_m, \, \cdot \,)\right) \,,
\ee
we say that $\mathcal{M}$ is \textit{strongly symplectic} if $\flat$ is an isomorphism, that it is \textit{weakly symplectic} if $\flat$ is injective but not surjective and that it is \textit{pre-symplectic} if $\flat$ is neither injective nor surjective.

\noindent The problem one wants to address is to understand whether and in which sense one is able to find a solution, for a pre-symplectic $\omega$, of the equation
\be \label{Eq: canonical equation}
i_\Gamma \omega = \dd \mathcal{H}
\ee
where $\Gamma$ is a vector field on $\mathcal{M}$. 
If $\omega$ has a non-trivial kernel $K$, such an equation does not make sense unless $\dd \mathcal{H} \bigr|_K = 0$ or, equivalently, $\dd \mathcal{H}_m \in \mathbf{T}_m \mathcal{M}^\flat$ where $\mathbf{T}_m \mathcal{M}^\flat$ is the image of $\mathbf{T}_m \mathcal{M}$ via $\flat$.
Defining $$\mathbf{T}_{m}\mathcal{M}^{\perp} \,=\, \left\{\, X_m \in \mathbf{T}_m\mathcal{M} \;\; :\;\; i_{X_m}\omega = 0 \,\right\} \subset \mathbf{T}_{m}\mathcal{M},$$  the elements in  $\mathcal{M}$ for which $\dd \mathcal{H}_m \in \mathbf{T}_m \mathcal{M}^\flat$ give the space\footnote{The set defined by \eqref{Eq: M1} is not necessarily a submanifold embedded into $\mathcal{M}$.
This relies on assumptions depending on specific cases.} $\mathcal{M}_1$ defined by 
\be \label{Eq: M1}
\mathcal{M}_1 \,=\, \left\{\, m_1 \in \mathcal{M} \;\; : \;\;  i_{X_{m_1}}\dd \mathcal{H}_{m_1} \,=\, 0 \;\; \forall X_{m_1} \in \mathbf{T}_{m_1}\mathcal{M}^{\perp} \,\right\} \,.
\ee

\begin{remark}
Notice that $\mathbf{T}_{m}\mathcal{M}^{\perp}$ is the kernel of $\omega$ at $m$, i.e. $K_m$. 
We used such a notation because it is better suited to formulate the first step of the algorithmic procedure we are going to describe in the rest of the section. 
Indeed, a definition of orthosymplectic complement can be given for any immersed submanifold $\mathfrak{i}\,:\,\mathcal{N}\,\hookrightarrow\, \mathcal{M}$  as
\be \label{Eq: ortosymp}
\mathbf{T}_n\mathcal{N}^\perp \,=\, \left\{\, X_n \in \mathbf{T}_n\mathcal{M}\bigr|_{\mathcal{N}} \;\; :\;\; \mathfrak{i}^\star \left( i_{X_n}\omega \right) = 0 \,\right\} 
\ee
and, thus, the definition of $\mathbf{T}_m\mathcal{M}^\perp$ given above can be thought of as an example of  \eqref{Eq: ortosymp} in the case where the submanifold $\mathcal{N}$ is the manifold $\mathcal{M}$ itself.
\end{remark}

\noindent Equation \eqref{Eq: canonical equation} makes sense if restricted to $\mathcal{M}_1$, but nothing is said about the solution  $\Gamma$ being tangent to $\mathcal{M}_1$. Notice that this is exactly what is natural to require, also from the physical point of view, since it amounts to have an evolution which is compatible with the given constraints. 
So we impose $\Gamma$ to be a vector field on $\mathcal{M}$ being $\mathfrak{i}_1$-related to a $\Gamma_1 \in \mathfrak{X}(\mathcal{M}_1)$, where $\mathfrak{i}_1\,:\,\mathcal{M}_1\,\hookrightarrow\,\mathcal{M}$ is the immersion.
By imposing this, we are further restricting $\mathcal{M}_1$ to a submanifold\footnote{That this is a submanifold relies on specific assumptions.} that we denote by $\mathcal{M}_2$ and that can be characterized in the following way.
Assume that an $\mathcal{M}_2$ exists for which the solution $\Gamma$ of 
\be \label{Eq: canonical equation m2}
\left(\,i_\Gamma \omega - \dd \mathcal{H} \,\right)_{\mathcal{M}_2} \,=\, 0
\ee
exists and is $\mathfrak{i}_2$-related\footnote{Here $\mathfrak{i}_2$ is the immersion map of $\mathcal{M}_2$ into $\mathcal{M}$.} to a vector field $\Gamma_2 \in \mathfrak{X}(\mathcal{M}_2)$.
Then, a straightforward computation shows that if $X_m \in \mathbf{T}_m \mathcal{M}$ is such that $\mathfrak{i}_2^\star \left( i_{X_m} \omega \right) \,=\, 0$, i.e., if $X_m \in \mathbf{T}_m \mathcal{M}_2^\perp$, then $\mathfrak{i}_2^\star \left( i_{X_m}\dd \mathcal{H} \right) \,=\, 0$. 
This last condition is equivalent to $i_{X_m}\dd \mathcal{H} = 0 \;\; \forall \, X_m \in \mathbf{T}_m \mathcal{M}_2^\perp$ since, being $X_m \in \mathbf{T}_{m}\mathcal{M}\bigr|_{\mathcal{M}_2}$, $i_{X_m}\dd \mathcal{H}$ is a function defined on $\mathcal{M}_2$ and, thus, its pull-back via $\mathfrak{i}_2$ coincides with the function itself.
Therefore, $i_{X_m}\dd \mathcal{H} = 0 \;\; \forall \, X_m \in \mathbf{T}_m \mathcal{M}_2^\perp$ is a necessary condition for a solution $\Gamma$ of  \eqref{Eq: canonical equation m2} which is tangent to $\mathcal{M}_2$ to exist.
By imposing such a necessary condition we restrict to a submanifold\footnote{Again, that this is actually a submanifold is an assumption.} $\mathfrak{i}_3\,:\,\mathcal{M}_3\,\hookrightarrow\,\mathcal{M}_2$, where again we need to verify that the canonical equations
\be
\left( i_\Gamma \omega - \dd \mathcal{H} \right)_{\mathcal{M}_3} \,=\,0 \,,
\ee
are well posed, i.e., that $i_{X_m} \dd \mathcal{H} \,=\, 0 \; \forall \, X_m \in \mathbf{T}_m\mathcal{M}_3^\perp$.

\noindent Therefore, one has an algorithmic procedure which, at the $n$-th step, defines the submanifold $\mathcal{M}_n$ of the original manifold $\mathcal{M}$ given by:
\be
\mathcal{M}_n \,=\, \left\{\, m \in \mathcal{M}_{n-1} \;\; : \;\; i_{X_m} \dd \mathcal{H} \,=\, 0 \;\; \forall \,\, X_m \in \mathbf{T}_m \mathcal{M}^\perp_{n-1}   \,\right\}
\ee
where 
\be
\mathbf{T}_m \mathcal{M}^\perp_k \,=\, \left\{\, X_m \in \mathbf{T}_m\mathcal{M}\bigr|_{\mathcal{M}_k} \;\; :\;\; \mathfrak{i}_k^\star \left( i_{X_m} \omega \right) \,=\, 0 \,\right\} \,,
\ee
and where one defines  $\mathcal{M}_0 \,=\, \mathcal{M}$ and $\mathbf{T}_m \mathcal{M}_0^\perp \,=\, \mathbf{T}_m\mathcal{M}^\perp \,=\, K_m$.

\noindent Such an algorithm is called the \textit{pre-symplectic constraint algorithm} and is said to converge (or to stabilize) when there exists an $n$ such that $\mathcal{M}_n \,=\, \mathcal{M}_{n-1} \,=\, \mathcal{M}_\infty$, which is called the final manifold of the algorithm.
On the final manifold $\mathcal{M}_\infty$, the canonical equation
\be
\left( i_\Gamma \omega - \dd \mathcal{H} \right)_{\mathcal{M}_\infty} \,=\, 0 \,,
\ee
makes sense and has solution $\Gamma$ tangent to $\mathcal{M}_\infty$, that is  $\Gamma$ is $\mathfrak{i}_\infty$-related to a $\Gamma_\infty \in \mathfrak{X}(\mathcal{M}_\infty)$.
Such a $\Gamma_\infty$ satisfies the equation:
\be
i_{\Gamma_\infty} \omega_\infty \,=\, \dd \mathcal{H}_\infty \,,
\ee
where $\omega_\infty \,=\, \mathfrak{i}_\infty^\star \omega$ and $\mathcal{H}_\infty \,=\, \mathfrak{i}_\infty^\star \mathcal{H}$.
The solution $\Gamma_\infty$ of the last equation is also unique if $\omega_\infty$ is non-degenerate and, in this case, gives rise to a unique solution of the original problem \eqref{Eq: canonical equation} which is given by the integral curves of $\Gamma_\infty$ immersed into $\mathcal{M}$ via $\mathfrak{i}_\infty$.

\section*{Acknowledgement}

The authors acknowledge financial support from the
Spanish Ministry of Economy and Competitiveness, through the Severo Ochoa Programme for Centres of Excellence in RD
(SEV-2015/0554), the MINECO research project PID2020-117477GB-I00, and Comunidad de Madrid project QUITEMAD++,
S2018/TCS-A4342. GM would like to thank partial financial support provided by the Santander/UC3M Excellence Chair Program 2019/2020, and he is also a member of the Gruppo Nazionale di Fisica Matematica (INDAM), Italy. FDC thanks the
UC3M, the European Commission through the Marie Sklodowska-Curie COFUND Action (H2020-MSCA-COFUND-2017- GA
801538) for their financial support through the CONEX-Plus Programme. FMC and LS acknowledge that the work has been supported by the Madrid Government (Comunidad de Madrid-Spain) under the Multiannual Agreement with UC3M in the line
of “Research Funds for Beatriz Galindo Fellowships” (C\&QIG-BG-CM-UC3M), and in the context of the V PRICIT (Regional
Programme of Research and Technological Innovation).

\bibliographystyle{alpha}
\bibliography{Biblio}

\end{document}